\documentclass{article}

\usepackage{sdhcmds}
\usepackage{sdhpapers}

\begin{document}

\title{Policy Learning with New Treatments}
\author{Samuel D. Higbee}
\institute{Department of Economics, University of North Carolina at Chapel Hill}
\email{sdhigbee@unc.edu}
\date{July 2025}

\abstract{
    I study the problem of a decision maker choosing a policy which allocates
    treatment to a heterogeneous population on the basis of experimental data that 
    includes only a subset of possible treatment values.  
    The effects of new treatments are partially identified by shape restrictions on 
    treatment response.  
    Policies are compared according to the minimax regret criterion, and I show 
    that the empirical analog of the population decision problem has a tractable 
    linear- and integer-programming formulation.  
    I prove the maximum regret of the estimated policy converges to the lowest 
    possible maximum regret at a rate which is the maximum of $N^{-1/2}$ and the 
    rate at which conditional average treatment effects are estimated in the 
    experimental data.  
    In an application to designing targeted subsidies for electrical grid connections 
    in rural Kenya, I find that nearly the entire population should be given a 
    treatment not implemented in the experiment,
    reducing maximum regret by over $60\%$ compared to the policy that 
    restricts to the treatments implemented in the experiment.
}

\acknowledgements{
    I am grateful to
    Max Tabord-Meehan,
    Alex Torgovitsky,
    St\'ephane Bonhomme,
    Guillaume Pouliot,
    Stefan Wager,
    Chen Qiu,
    and Davide Viviano for helpful feedback.
    I would also like to thank seminar participants at 
    the University of Chicago, 
    Washington University in St. Louis,
    and the North American Winter Meetings of the Econometric Society 2024
    for helpful comments.
}

\maketitle
\newpage

\section{Introduction}
\label{section-introduction}

Heterogeneous treatment effects are often estimated with a decision problem in mind--- 
should a particular individual be treated?  
This question has fostered much research in econometrics, statistics, and machine learning.  
However, relatively less attention has been given to another important margin of the decision--- 
should the treatment itself be adjusted? 
Whether the treatment is a medical treatment, subsidy, job training, or audit probability, 
decision makers can usually entertain changing the treatment value that was observed in the data.  
Even experiments with multivalued treatments may not implement an exhaustive list of treatment values.  
This is especially true in the social sciences, where testing multiple interventions can be costly, 
and in the medical sciences, where specific treatment doses are often tested in clinical trials.  
In this paper I propose a method for allocating treatment to a population when the treatment values 
themselves can be adjusted to values never before seen in the data.  
I show how combining the data on existing treatments with economically motivated shape restrictions 
can be used to design policies that outperform those possible when only previously implemented 
treatments are considered.

I first formulate a decision problem in which the decision maker observes 
experimental data on some treatment values and seeks to construct a 
mapping, or policy, from the space of covariates to the space of treatments 
in order to maximize some objective function.  
I assume all experimentation is done before the policy is constructed.  
This setting, which is common in econometrics, is often referred to as 
treatment choice or offline policy learning.
Examples include 
\textcite{atheyPolicyLearningObservational2021}
\textcite{bhattacharyaInferringWelfareMaximizing2012},
\textcite{kitagawaWhoShouldBe2018}
and other examples mentioned in the literature review thereof, 
\textcite{liuPolicyLearningEndogeneity2024},
\textcite{mbakopModelSelectionTreatment2021},
\textcite{qianPerformanceGuaranteesIndividualized2011},
\textcite{sasakiWelfareAnalysisMarginal2024},
\textcite{zhangEstimatingOptimalTreatment2012},
and
\textcite{zhaoEstimatingIndividualizedTreatment2012}.
A distinctive feature of this paper as opposed to most policy learning problems 
is that the set of treatments that the decision maker can consider may be 
a strict superset of the support of the treatment random variable observed in the data.  
This extends policy learning to practically relevant situations in which constraints 
in the design and implementation of experiments or simply differences in the 
objectives of the experimenter versus decision maker result in only a few treatment 
values being piloted in the experiment, while the decision maker may want to consider 
many more.

Despite the lack of data on the impacts of these never-before-implemented treatments, 
I show how to bound the response to new treatments using simple, economically 
interpretable restrictions on the shape of treatment response. 
For example, a financial incentive may be assumed to have a positive effect, exhibit 
diminishing returns, or satisfy smoothness conditions. 
Such shape restrictions are often exploited to partially identify treatment effects 
(e.g. \cite{manskiIdentificationPredictionDecision2009},
\cite{mogstadUsingInstrumentalVariables2018}).
The empirical analysis of the present paper demonstrates that such bounds can be 
adequately informative for choosing whether and how to implement new treatment values.
Based on these bounds, I construct a population decision problem to choose which 
treatment to assign to each covariate value.
I use the minimax regret criterion to evaluate treatment choice under partial 
identification following \textcite{manskiMinimaxregretTreatmentChoice2007}.

As in \textcite{manskiStatisticalTreatmentRules2004a},
\textcite{kitagawaWhoShouldBe2018} and the 
subsequent literature on empirical welfare maximization methods, 
I propose a decision rule based on solving the  
empirical analog of the decision problem as a surrogate for the 
infeasible population objective.  
The resulting empirical minimax regret estimator is constructed by minimizing maximum regret
across an estimate of the partially identified set of treatment response functions.  
In this way, the resulting policy is robust to model ambiguity induced by introducing 
new treatments.
Despite involving nested, non-closed form optimization problems which characterize the 
identified set for treatment response, I show how the optimal policy can be computed 
using the same linear and integer programming tools common in the policy learning 
literature. 
The estimator is thus computationally feasible and can be implemented by widely available 
software.

I show that the proposed decision rule possesses desirable regret properties.
The maximum regret obtained under the estimated policy converges to the smallest
possible maximum regret that the decision maker could have achieved in the absence of 
sampling uncertainty-- that is, if the population identified set were observed--
uniformly across a set of data distributions.
The rate at which the regret of the estimated policy converges to its optimum depends on 
the estimation rate of the response to the treatments which were observed in the data,
and hence is an asymptotic rather than finite-sample convergence guarantee.  
In the case of discrete covariates, or more generally parametric rates of convergence for 
estimated treatment effects, the rate of convergence of maximum regret is $N^{-1/2}$.
Otherwise, maximum regret converges at the nonparametric rate.

I apply the method to data from \textcite{leeExperimentalEvidenceEconomics2020},
in which households in rural Kenya were offered one of four prices in 
$0$, $15$, $25$, or $35$ thousand shillings to connect to the electrical grid.  
I consider a decision maker able to offer prices in increments of $2.5$ thousand shillings 
based on household size and income.  
This represents a much richer set of fifteen possible treatments, allowing for finer 
targeting of personalized prices to optimize the cost-effectiveness of the subsidy program. 
To bound the takeup at these new prices, I assume demand is downward sloping and convex.  
The estimated minimax regret optimal policy assigns prices that were not implemented in the 
experiment to nearly the entire population.
Moreover, the maximum regret of the estimated policy is over $60\%$ lower than the maximum regret 
of the best policy that only implements the prices implemented in the experiment,
illustrating that constraining the decision 
maker to treatments that appear in the experimental data can result in suboptimal decisions.

\subsection{Related literature}

This paper contributes to a growing literature on statistical treatment rules in
econometrics beginning with 
\textcite{manskiStatisticalTreatmentRules2004a} and
\textcite{kitagawaWhoShouldBe2018},
which introduced the now-common empirical welfare maximization framework.  
I follow a similar strategy of constructing an empirical analog of the 
population objective, but seek to minimize the worst-case regret that can occur 
within the identified set of treatment response.  

Forecasting the effects of treatments or policies never before observed in the 
data is a fundamental goal of econometrics, especially when applied as a guide 
for public policy 
(see \textcite{heckmanChapter70Econometric2007} and
\textcite{manskiEconometricsDecisionMaking2021} 
for a deep discussion, including a historical overview).
Nonetheless, the recent literature on policy learning and treatment choice has 
generally not considered the introduction of new treatments with partially 
identified effects.

Previous literature has treatment choice
under various forms of partial identification.
\textcite{manskiSearchProfilingPartial2006} and 
\textcite{manskiVaccinationPartialKnowledge2010} consider minimax regret
treatment choice when the decision maker only observes data under
the status quo policy and uses monotonicity restrictions to partially identify
the effects of counterfactual treatment intensities on welfare.
The analysis is in the population, and issues of statistical estimation are not
considered.
\textcite{manskiMinimaxregretTreatmentChoice2007}
considers minimax regret treatment choice when some outcome data is missing 
not necessarily at random, leading to partial identification of treatment effects,
and proposes an empirical analog.
In contrast to this paper and much of the statistical policy learning
literature,
\textcite{manskiSearchProfilingPartial2006,manskiVaccinationPartialKnowledge2010,manskiMinimaxregretTreatmentChoice2007}
do not consider restricted policy classes,
yielding a decision problem that is separable in covariates.
These restricted policy classes are also important for the convergence
properties of the estimated policy.

The paper most closely related to this one is
\textcite{manskiUsingLimitedTrial2025},
which studies the question of how to allocate new dosage levels of a treatment 
given experimental evidence on a subset of possible dosage levels.
Like this paper, \textcite{manskiUsingLimitedTrial2025} uses shape restrictions to bound the
response to new treatments, and uses the minimax regret criterion to choose
a decision rule.
Unlike this paper,
\textcite{manskiUsingLimitedTrial2025} assumes population-level quantities are known and
hence does not consider statistical properties of estimated decision rules, 
nor does it consider targeting new treatments on the basis of
covariates using complexity-constrained policy classes.
The two papers also use different utility functions--
the present paper using a linear-in-outcome utility function,
while \textcite{manskiUsingLimitedTrial2025} assuming four discrete outcomes
associated with different utility levels.
Finally, \textcite{manskiUsingLimitedTrial2025} considers fractional treatment assignment
in addition to the deterministic treatment assignment considered in this paper.

\textcite{ben-michaelSafePolicyLearning2022} develops a method for learning
policies from data gathered under a deterministic policy for which strict 
overlap fails;
\textcite{zhangSafePolicyLearning2024} tailors this framework to the case where
the deterministic policy is a regression discontinuity design.
The introduction of new treatments is similar to the deterministic policy
setting considered here in that it is also a case where strict overlap fails.
\textcite{ben-michaelSafePolicyLearning2022} 
and \textcite{zhangSafePolicyLearning2024} use a maximin gain welfare criterion,
where the objective is to learn a policy that is guaranteed to weakly improve
on the status quo policy.
\textcite{khanOffpolicyEvaluationOverlap2024} studies robust policy evaluation 
using Lipschitz constraints when strict overlap fails.
\textcite{ben-michaelSafePolicyLearning2022},
\textcite{zhangSafePolicyLearning2024},
and \textcite{khanOffpolicyEvaluationOverlap2024} 
focus on partial identification through restrictions on response as a function of covariates,
while this paper focuses on shape restrictions on the response across treatment values.

Unobserved confounding can be a source of partial identification
in policy choice with observational data.
\textcite{kallusMinimaxOptimalPolicyLearning2021} studies this setting,
and use bounds on the distance between the true propensity weights and the
observed (biased) weights to partially identify the effect of policies.
Their criterion is maximum regret relative to a baseline policy such as the
status quo, and like the present paper they show that the maximum regret of the
estimated policy converges to the lowest possible maximum regret.
\textcite{puEstimatingOptimalTreatment2021} uses an instrumental variable
to partially identify treatment effects when unobserved confounding precludes point
identification, and proposes a classification-based approach
with a surrogate loss to learn the optimal policy under a maximin welfare criterion.

While unobserved confounding threatens the internal validity of the estimated
policy on the experimental population,
other papers consider threats to external validity,
where the experimental population is different from the target population.
\textcite{adjahoExternallyValidPolicy2023a} studies this setting,
using Wasserstein neighborhoods to construct the identified set,
and derive closed form expressions worst case welfare within these 
neighborhoods.
\textcite{leiPolicyLearningBiased2023} studies policy learning when the 
experimental population may self-select into the experiment on the basis 
of unobserved characteristics.
They consider maximin, maximin gain, and minimax regret policies.
They solve for a closed form when the policy class is unconstrained,
and propose an estimation method which is not a plug-in method.

\textcite{dadamoOrthogonalPolicyLearning2023} studies policy learning with a 
binary treatment where the identified set is rectangular,
meaning it is constructed by taking the product of pointwise bounds on each treatment effect.
In contrast, shape restrictions on the response across treatment values
generally yield nonrectangular identified sets.
This leads to difficulties when estimating the optimal policy in my setting because the bounds I 
identify do not in general admit a closed form.
However, the extra information provided by these shape restrictions
can lead to lower maximum regret than one would obtain using pointwise bounds.
This is illustrated in the empirical example of Section 
\ref{section-application}.
\textcite{dadamoOrthogonalPolicyLearning2023} also provides a doubly robust estimator
that can improve the convergence rate of the estimated policy under a margin condition.

\textcite{stoyeMinimaxRegretTreatment2012a} gives exact finite-sample results
for minimax regret treatment choice in Binomial and Gaussian experiments,
also with unrestricted policy classes.
\textcite{yataOptimalDecisionRules2025} gives exact finite-sample minimax regret
results in more general Gaussian settings with binary policies.
These papers do not consider treatment choice with multiple treatments.
Another difference is that the present paper only delivers asymptotic performance guarantees,
but does not require distributional assumptions.

Many of the previously mentioned works are concerned with binary treatments, 
while I am concerned with multivalued treatments.
\textcite{zhouOfflineMultiActionPolicy2023} and
\textcite{kallusPolicyEvaluationOptimization2018}
consider policy learning with multivalued treatments and continuous treatments, respectively, 
but in point-identified settings where all possible treatment values are 
implemented in the experiment. 

\textcite{atheyPolicyLearningObservational2021} extends policy learning to observational studies 
where exogeneity of treatment only holds after conditioning on high-dimensional 
covariates. 
In contrast, I am motivated by settings in which decision makers have data from 
a pilot experiment which tested a few treatment values.  
When this is the case, estimating the effects of policies involving new 
treatments only requires conditioning on the set of covariates used in the 
treatment rule, which is typically low-dimensional due to exogenous constraints 
on the policy class (\cite{kitagawaWhoShouldBe2018}).
\textcite{atheyPolicyLearningObservational2021} also considers infinitesimal, local changes to 
treatment values;
however, I consider new treatments that are sufficiently far from the support of
the data as to make local approximations or parametric extrapolations 
unreliable, necessitating a partial identification approach.

An alternative to the plug-in approach used in this paper and common in policy learning is to average across the 
parameter space according to some distribution.
\textcite{christensenOptimalDecisionRules2023}
study optimal decisions in a discrete set under partial identification where Bayes 
rules and the bootstrap distribution are used to average over the space of identified parameters, 
while a minimax approach is taken over the partially identified parameters. 
An important finding is that plug-in-rules may be dominated in the limit experiment.  
See \textcite{hiranoAsymptoticsStatisticalTreatment2009} and
\textcite{hiranoAsymptoticAnalysisStatistical2020}
for further discussion of asymptotic
optimality of statistical treatment rules.

The rest of the article is organized as follows:
Section \ref{section-decision} describes the decision problem in the population 
and shows how to incorporate information from shape restrictions.  
Section \ref{section-estimation} describes the empirical minimax regret problem 
and the algorithm for estimating the optimal policy.  
Section \ref{section-convergence} describes the convergence guarantees.   
Section \ref{section-application} applies the method to study personalized 
subsidies to connect to the electrical grid in rural Kenya.

\section{Population decision problem}
\label{section-decision}
\subsection{General framework}

A decision maker has access to experimental data and must choose a rule 
assigning individuals to treatments based on their observable covariates.  
The experimental data is described by random variables $(D, X, Y)$ taking 
values in $\mathcal{D}_0 \times \mathcal{X} \times \mathcal{Y}$ where $D$ 
is a treatment taking $|\mathcal{D}_0| = J_0$ values in the data,
$X$ are observed covariates, 
and $Y$ is a univariate outcome of interest.
$D$ is assumed to be randomly assigned, perhaps conditionally on $X$.

Although the random variable $D$ only takes values in $\mathcal{D}_0$, 
the decision maker can consider assigning individuals to any treatment 
value $d \in \mathcal{D}$ where $\mathcal{D}$ is potentially larger than 
$\mathcal{D}_0$.
Hence, I assume the existence of potential outcomes $Y(d)$ for all 
$d \in \mathcal{D}$.
The set $\mathcal{D}$ has cardinality $| \mathcal{D} | = J < \infty$, and 
its elements are denoted by $d_j$ for $j \in\{1,\dots,J\}$.
The observed outcome $Y$ is generated as $Y = Y(D)$.
Let $P$ denote the distribution of $(D, X, (Y(d))_{d \in \mathcal{D}})$.

The decision maker seeks a policy $\pi : \mathcal{X} \mapsto \mathcal{D}$ 
which assigns individuals to treatment status based on their observable 
covariates.  
The policy is chosen from some set $\Pi$ which is taken as given.
The treatment assigned to an individual with covariate values $X$ is
$\pi(X)$ and the realized outcome is $Y(\pi(X))$.

The decision maker has some utility function $u(d,x,y)$ which may depend 
on the treatment assigned, covariates, and the realized outcome of interest.  
I assume the decision maker is utilitarian and ultimately cares about the 
expected utility derived from the data realized from the policy, resulting 
in the following problem that the decision maker would like to solve
\begin{align*}
    &\max_{\pi \in \Pi} \E_P\bigg[ u(\pi(X), X, Y(\pi(X)))\bigg]
    \\
    &=\max_{\pi \in \Pi} \E_P\bigg[ v_P(\pi(X), X)\bigg]
\end{align*}
where $v_P(d,x) := \E_P[u(d,X,Y(d)) \mid X = x]$ is the conditional mean utility. 

Two sources of ignorance on the decision maker's part make this problem 
infeasible to solve in practice.  
The first is that only a sample is observed, so the population probability 
distribution is unknown.  
The second is that even if the population distribution of the data $(D,X,Y)$ 
were known, the effects of some treatments are not identified because they 
are never observed.  
In particular, the function $v$ depends on the distribution of potential 
outcomes $Y(d)$ for values of $d$ not in $\mathcal{D}_0$.  
Since data on these potential outcomes are not observed in the sample, the 
decision maker's objective is not point identified.
To deal with partial identification, I will solve a proxy problem which 
is robust to partial identification in that it achieves uniformly low regret 
across the identified set for $v$.
Since only sample data is available, I solve the empirical or plug-in version 
of this problem.

Following \textcite{manskiStatisticalTreatmentRules2004a} and much of the 
econometric literature on treatment choice, policies will be evaluated based on 
their expected regret. 
For a chosen policy $\pi$, the regret of $\pi$ is the difference in expected 
utility obtained from implementing the first-best policy versus $\pi$.  
The first-best policy maps each covariate value $x$ to $\arg\max_d v_P(d,x)$. 
For any chosen policy $\pi$, the expected regret of implementing $\pi$ versus 
implementing the first-best policy is
\begin{align*}
    R_P(\pi) &:= \E_P\bigg[ \max_d v_P(d,X) - v_P(\pi(X), X)\bigg]
    .
\end{align*}
Since $v_P$ is not identified, regret is not identified either.  
However, letting $\mathcal{V}_P$ be the identified set for $v_P$ determined by 
the experimental data (which will be characterized shortly), 
the maximum expected regret that can occur if the decision 
maker implements policy $\pi$ is given by
\begin{align*}
    \overline{R}_P(\pi) := 
    \max_{v \in \mathcal{V}_P} 
    \E_P\bigg[ \max_{d \in \mathcal{D}} v(d,X) - v(\pi(X), X) \bigg]
\end{align*}
I use the minimax regret criterion to guide the choice of policy.  
This means $\pi$ is chosen to minimize the largest regret that can occur within 
the identified set--- that is, $\overline{R}_P(\pi)$.  
Therefore, the decision maker chooses $\pi$ to minimize the worst-case expected 
regret as follows
\begin{align}
    \label{MMR}
    \pi_P^* \in \arg\min_{\pi \in \Pi} \overline{R}_P(\pi)
    = \arg \min_{\pi \in \Pi}\max_{v \in \mathcal{V}_P}
    \E_P\bigg[ \max_{d \in \mathcal{D}} v(d,X) - v(\pi(X), X) \bigg]
\end{align}
This ensures that the chosen policy minimizes regret uniformly across the identified set.  
If the minimizer is not unique, the decision maker is indifferent among them.
Since $v$ is unknown, the minimum maximum regret is generally larger than zero.

The minimax regret criterion is not the only method for comparing statistical decisions 
with partially identified effects.  
In the context of treatment choice, 
\textcite{manskiChoosingTreatmentPolicies2011}
compares the minimax 
regret criterion with the maximin welfare and subjective expected welfare criteria, two 
common alternatives. 
Under the maximin welfare criterion, the decision maker seeks to maximize the minimum 
possible level of the outcome that could be attained as opposed to the minimum gap between 
the attained and first-best level of the outcome.  
The method for construction and estimation of the optimal policy that follows can be 
applied when using the maximin welfare criterion as well. 
Indeed, it can be obtained as a simplification of what follows by replacing 
$\max_{d\in\mathcal{D}} v(d, X)$ with $0$ in (\ref{MMR}).  
However, the resulting estimator will of course have different behavior and regret properties.

In some settings the maximin criterion can be quite conservative
(see the discussion of 
\textcite{waldStatisticalDecisionFunctions1949}
found in \textcite{savageTheoryStatisticalDecision1951}).
Indeed, unless a new treatment $d\in\mathcal{D}\setminus\mathcal{D}_0$ can be guaranteed to 
outperform the original set of treatments in every possible state of the world 
$v \in \mathcal{V}_P$, the maximin welfare criterion will not implement new treatments.  
This is because under the maximin welfare criterion the decision is driven entirely by hedging 
against the least favorable state of the world.  
In contrast, the minimax regret criterion considers the suboptimality gap in all possible 
states of the world.  
The decision maker measures the performance of the policy in each state of the world according 
to the benchmark of optimality in that state of the world.  
I follow \textcite{manskiMinimaxregretTreatmentChoice2007} 
in applying the minimax regret criterion to treatment choice.  
This represents a particular choice of loss function and in turn delivers a point estimate of an 
optimal policy.

When the probability of each state of the world $v \in \mathcal{V}_P$ can be described by a 
probability distribution, the Bayesian approach to decision-making can be applied.  
This consists of setting a prior on states of the world $v \in \mathcal{V}_P$, 
using the data to form a posterior, and selecting a treatment policy which maximizes posterior 
expected welfare. 
One potential weakness of this approach in the context of introducing new treatments
is that the lack of identification means that even in large samples,
the influence of the prior on the posterior will be substantial.
Yet another possible approach to estimate the effects of treatments that lie outside the support of 
the data could be to extrapolate using a parametric model, thus circumventing entirely the need for 
partial identification.  
However, when the new treatments are sufficiently far from the support of the data, a parametric 
point-identified model substantially understates the degree of model uncertainty.  
This is illustrated in Section \ref{section-application} where a policy based on
parametric extrapolation leads to substantially higher maximum regret than
the estimated minimax regret policy.

\subsection{Imposing shape restrictions}

I now describe how a tractable characterization of the minimax regret problem 
(\ref{MMR}) can be obtained using shape restrictions on the treatment response.
This requires that the utility function is linear in the outcome of interest.
That is, there exist known functions $b$ and $c$ such that
\begin{align*}
    u(d,x,y) = b(d,x)y - c(d,x)
    .
\end{align*}
While it is often possible to avoid the assumption of linear utility by simply 
redefining $Y$ as utility, in some applications 
(such as in Section \ref{section-application}) it may be more natural 
to impose shape restrictions in terms of the original outcome variable, 
which may relate to a structural economic quantity such as a demand curve.  
In Section \ref{section-application}, $y$ will be a purchase indicator, 
$b(d,x)$ will be value of connections net of the cost of the subsidy, 
and $c(d,x)$ represents the cost of offering the subsidy regardless of takeup, 
which I take to be $0$.\footnote{
    If $b(d,x)$ and $c(d,x)$ represent preferences of a population, they may
    have to be estimated from the data.
    While this paper focuses on uncertainty about the response of $y$ to 
    new treatments, Appendix \ref{appendix-extensions} discusses how to extend
    the methods to the case where $b(d,x)$ and $c(d,x)$ are estimated.
}

Note that the assumption of linearity implies that
\begin{align*}
    v_P(d,x) = b(d,x) m_P(d,x) - c(d,x)
\end{align*}
where $m_P(d,x) := \E_P[Y(d) \mid X=x]$ is the conditional mean response function. 
Moreover, any two probability distributions which induce the same conditional mean response 
function will induce the same expected utility function.  
I therefore will also use the notation
$v_m(d,x) := b(d,x) m(d,x) - c(d,x)$ where conditional mean utilities are
indexed by conditional mean response functions rather than probability distributions.
Since $b$ and $c$ are known functions, in order to characterize maximum regret 
it is sufficient to characterize the identified set for $m_P$.

The decision maker has experimental data on the effectiveness of some treatments.  
This means that $m_P(d,\cdot)$ is identified for every $d\in\mathcal{D}_0$.
For this information on the effects of treatments in $\mathcal{D}_0$ to be informative 
about the effects of treatments in $\mathcal{D} \setminus \mathcal{D}_0$, some 
structure must be known about the mean conditional response function $m_P$.  
For example, the decision maker may know that demand is downward sloping, 
that a particular intervention features decreasing returns to scale, 
or that the treatment response exhibits some smoothness properties.  
By combining knowledge of $m_P(d,\cdot)$ for $d\in\mathcal{D}_0$ with such shape 
restrictions, the effects of new treatments may be partially identified.

Let the set of shape-restricted mean conditional response functions be denoted
by $\mathcal{S}$.
The sharp identified set for $m_P$ is
\begin{align*}
    \mathcal{M}_P &:= \mathcal{S} \cap \{ m : m(d, X) = m_P(d, X), \; P-\text{a.s.},\; \forall d \in \mathcal{D}_0 \}
\end{align*}
which represents the set of functions which obey the shape restrictions and 
match identified population means.
I assume that $\mathcal{S}$ restricts the shape of $m(d,x)$ in $d$ for any given $x$,
leaving the behavior of $m$ across $x$ unrestricted.

\begin{assumption}
    \label{assumption:pointwise_shape_restriction}
    There exist sets $\mathcal{S}_x \subset \R^J$ for each $x \in \mathcal{X}$ such that
    \begin{align*}
        \mathcal{S} = \{ m : m(\cdot, X) \in \mathcal{S}_X \; P-\text{a.s.} \}.
    \end{align*}
\end{assumption}

Under Assumption \ref{assumption:pointwise_shape_restriction}, 
a hypothetical conditional mean response function $m$ is in $\mathcal{S}$ 
if and only if $m(\cdot, X)$ satisfies some shape restrictions almost surely
in $X$.
That is, $\mathcal{S}$ encapsulates assumptions about the shape of $m_P$ across $d$ 
for fixed $x$, leaving the behavior of $m_P$ across $x$ unrestricted
(\cite{manskiMonotoneTreatmentResponse1997},
\cite{manskiSearchProfilingPartial2006}).
This means that an individual at a particular covariate value is assumed to have an 
expected treatment response that is decreasing, convex, smooth, etc. 
The sets $S_x$ may also stipulate that $m(\cdot, x)$ belongs to some parametric family,
such as polynomials.

Under Assumption \ref{assumption:pointwise_shape_restriction}, 
the maximization over $v$ (equivalently maximization over 
$m \in \mathcal{M}_P$) in (\ref{MMR}) is solved by considering each value of $x$ 
in isolation and finding the $m(\cdot, x)$ which maximizes regret.  
This allows the maximum to be interchanged with the expectation in the minimax 
regret problem (\ref{MMR})
\begin{align}
    \label{MMRLP}
    \nonumber
    \min_{\pi \in \Pi}  \max_{m \in \mathcal{M}_P} &\;  \E_P\bigg[ \max_{d \in \mathcal{D}} v_m(d,X) - v_m(\pi(X), X) \bigg]
    \\
    = \min_{\pi \in \Pi} & \; \E_P\bigg[ \max_{m \in \mathcal{M}_P} \bigg( \max_{d \in \mathcal{D}} v_m(d,X) - v_m(\pi(X), X) \bigg)\bigg]
    \nonumber
    \\
    = \min_{\pi \in \Pi} & \; \E_P\bigg[ \sum_{j=1}^J \pi_j(X) \Gamma_{j,P}(X) \bigg]
\end{align}
where $\pi_j(X) = \1[\pi(X) = d_j]$ and
\begin{align*}
    \Gamma_{j,P}(X) = \max_{m \in \mathcal{M}_P} \bigg( \max_{d \in \mathcal{D}} v_m(d,X) - v_m(d_j, X) \bigg)
\end{align*}
which can be interpreted as the contribution to maximum expected regret of 
assigning a person with covariate values $X$ to treatment $d_j$. 

The optimization problem (\ref{MMRLP}) defines the policy which is optimal 
in terms of its population minimax regret.  
The maximum regret of any policy depends on the strength of the assumptions 
encoded in $\mathcal{S}$, and their implications for the size of the identified 
set $\mathcal{M}_P$.
Larger identified sets will lead to higher maximum regret,
since it expands the set from which the worst-case response $m$ can be chosen.
The size of the identified set also depends on the relationship between the 
new and existing treatments.
If a new treatment $d_j$ lies between two existing treatments,
restrictions on $m$ such as monotonicity can provide informative bounds on
$m(d_j, x)$.
When $d_j > d$ for all $d \in \mathcal{D}_0$, monotonicity 
can leave the identified set unbounded in the absence of additional assumptions.

The benefit of imposing shape 
restrictions only on the behavior of $m_P$ across $d$ is that the 
representation (\ref{MMRLP}) is an optimization problem over a population 
expected loss defined by $\Gamma_{j,P}(X)$.  
This problem possesses a form similar to decision problems presented in 
\textcite{atheyPolicyLearningObservational2021}, 
\textcite{dadamoOrthogonalPolicyLearning2023}
and others, with the key 
distinction that the covariate-level loss $\Gamma_{j,P}(X)$ is itself the solution 
to an optimization problem which generally will not have a closed-form solution.  
Nonetheless, the minimax regret problem (\ref{MMRLP}) can be cast in terms of 
the empirical welfare maximization framework of 
\textcite{kitagawaWhoShouldBe2018}.
In the following section, I discuss how to set up the empirical analog of the 
nested optimization problem (\ref{MMRLP}) and provide a computationally 
attractive algorithm for solving it.

\section{Estimation}
\label{section-estimation}
The optimization problem (\ref{MMRLP}) is infeasible for the decision maker 
because in practice only a sample \\ $\{(D_i, X_i, Y_i)\}_{i=1}^N$ is 
observed.
Instead, I propose solving the empirical analog of (\ref{MMRLP}) to obtain
an estimate of the population optimal policy.
Insofar as the constraints of this problem are constructed from consistent
estimators, the optimal policy will inherit similar properties.

In this section I describe the empirical analog of (\ref{MMRLP}) and
provide a solution procedure.
It consists of first estimating the effects of the treatments which were 
implemented in the experimental data, then constructing estimates of 
$\Gamma_{j,P}(X_i)$ for every observation $i$ and treatment $j$, and finally 
plugging these estimates into the empirical analog of (\ref{MMRLP}) where the 
sample mean is used instead of the population expectation.  
I show how these estimates of $\Gamma_{j,P}(X_i)$ can be computed using linear programming, 
resulting in a mixed integer-linear programming formulation for (\ref{MMRLP}) 
for many policy classes $\Pi$.

First, I estimate the mean conditional response function for every 
$d \in \mathcal{D}_0$, denoted $\hat m_0(d,x)$.  
Except for high level conditions on the accuracy of the estimate detailed in 
Section \ref{section-convergence}, I remain agnostic about how the estimate is 
constructed.  
The estimate $\hat m_0$ is used to construct an estimate of the identified set 
for $m_P(\cdot, X_i)$ for each $i$ as a function of $d$ in all of $\mathcal{D}$, which 
represents covariate-level bounds on the effects of new treatments. 
The empirical analog of $\mathcal{M}_P$ is the set of functions which obey the 
shape restrictions and match estimated sample means, and is denoted by 
$\hat{\mathcal{M}}:= \mathcal{S} \cap \{ m : m(d, X_i) = \hat m_0(d, X_i) \; \forall i,\; \forall d \in \mathcal{D}_0 \}$.  
I assume it is nonempty.  
As discussed in Section \ref{section-convergence}, estimates which violate the 
shape restrictions and hence yield an empty $\hat{\mathcal{M}}$ can be projected 
onto the set of functions which satisfy the shape restrictions.  
Since these are assumed to hold in the population, imposing such shape 
restrictions on estimators typically improves performance in finite samples 
(\cite{chetverikovEconometricsShapeRestrictions2018}).

This is then used to construct estimates $\hat \Gamma_j(X_i)$ of the 
covariate-level loss $\Gamma_{j,P}(X_i)$, for every observation $i$ and treatment $j$.  
That is,
\begin{align}
    \label{gammaj}
    \hat \Gamma_j(X_i) = \max_{ m \in \hat{\mathcal{M}}} \bigg( \max_{d \in \mathcal{D}} v_m(d,X_i) - v_m(d_j, X_i) \bigg)
\end{align}
These estimates are then used in the program
\begin{align}
    \label{empiricalMMR}
    \hat \pi := \arg \min_{\pi\in\Pi} \sum_{i=1}^{n} \sum_{j=1}^J \pi_{ij} \hat\Gamma_j(X_i)
\end{align}
where $\pi_{ij} = \1[\pi(X_i) = d_j]$. 

Having defined the estimator for the minimax regret optimal policy 
(\ref{empiricalMMR}), I turn to computationally convenient methods for 
estimating $\hat \Gamma_j(X_i)$ and thereby the policy $\hat \pi$.  
This is achieved by expressing $\hat \Gamma_j(X_i)$ through linear 
programs and considering policy classes $\Pi$ which can be expressed 
using linear and integer constraints.
In doing so, I impose some additional structure on the set of shape 
restricted functions $\mathcal{S}$ specified in Assumption \ref{assumption:pointwise_shape_restriction}.
Specifically, I assume the a priori knowledge on shape restrictions can be
summarized through $\ell$ linear inequalities on the treatment response vector
for almost every $x$.

\begin{assumption}
    \label{linear}
    There exists a matrix $S \in \R^{\ell \times J}$ and 
    a vector $r \in \R^\ell$ such that the conditional 
    mean response function $m(d,x)$ is in the set $\mathcal{S}$ if and 
    only if $S m(\cdot,X) \leq r$ $P-$almost surely, 
    where $m(\cdot,x) := (m(d_1,x), \dots, m(d_J,x))'$.
\end{assumption}

This assumption strengthens Assumption 2.1 by requiring that the shape 
restrictions are linear in the treatment response vector.
Such linear restrictions can accommodate a wide range of shape restrictions 
that may be used in practice.
For example, restrictions on the first, second, or higher differences of the mean 
conditional response can be expressed this way, allowing for $m_P$ to be constrained 
to be decreasing, Lipschitz, convex, or obey higher order smoothness conditions 
(\cite{mogstadUsingInstrumentalVariables2018}).\footnote{
    The analysis can be extended to allow $S$ and $r$ to depend on $x$.
    I have focused on the case where the shape restrictions are the same for all $x$
    for simplicity.
}
Upper and lower bounds on $m_P$ can also be expressed through such constraints.  
Appendix \ref{appendix-computation} describes in detail how the restrictions of 
decreasing demand and diminishing responsiveness to the subsidy are applied to the empirical 
example in Section \ref{section-application}.

The following examples convey the practical use of the assumption.

\begin{example}
    Suppose $\mathcal{D} = \{1,2,3,4\}$ and $m_P(d, x)$ is assumed to be 
    increasing and concave in $d$.  
    Then $ m \in \mathcal{S}$ if and only if 
    $S m(\cdot, X) \leq r$ $P$-almost surely where
    \begin{align*}
        S = \begin{bmatrix}
            1 & -1 & 0 & 0\\
            0 & 1 & -1 & 0\\
            0 & 0 & 1 & -1\\
            1 & -2 & 1 & 0\\
            0 & 1 & -2 & 1\\
        \end{bmatrix}
        \quad \quad
        r = \begin{bmatrix}
            0 \\ 0 \\ 0 \\ 0 \\ 0
        \end{bmatrix}
    \end{align*}
    Here the first three rows of $S$ ensure that $m$ is increasing, 
    and the second two rows of $S$ ensure concavity.
\end{example}

\begin{example}
    Suppose $m_P(\cdot,x)$ is assumed to be a polynomial of degree 
    $L$ where $J_0-1 \leq L \leq J-1$,
    in addition to shape restrictions such as boundedness or monotonicity.
    These restrictions can be imposed by using Bernstein polynomials.
    That is,
    \begin{align*}
        \mathcal{S}_x = \bigg\{
            m \in \R^J :
            m_j = B(d_j, L)' \beta, \;
            T \beta \leq r \;
            \beta \in \R^{L+1}
        \bigg\}
    \end{align*}
    where $B(d_j, L)$ is an $L+1$ vector of Bernstein polynomials of degree
    $L$ in $d_j$,
    and $\beta$ is a vector of coefficients.
    The matrix $T$ encodes shape restrictions on $m$ through
    constraints on the coefficients $\beta$.

    Let $B$ be the $J \times (L+1)$ matrix of Bernstein polynomials of degree $L$.
    Since $B$ has linearly independent columns, it has a left inverse $B^\dagger$
    and $m \in S_x$ if and only if $S m \leq r$ where $S = T B^\dagger$.
\end{example}

To ensure that $\hat{\mathcal{M}}$ consists of functions which match sample 
analogs of identified means on $\mathcal{D}_0$, I introduce the 
$J_0 \times J$ matrix $F$ where $F_{kj} = 1$ if $d_j$ is the $k$th 
element of $\mathcal{D}_0$ and $F_{kj} = 0$ otherwise.  
Then the identified set $\mathcal{M}_P$ is the set of all $m$ such that 
$S m(\cdot,X) \leq r$ and $F m(\cdot,X) = m_{0,P}(\cdot, X)$ almost surely, 
where $m_{0,P}(\cdot, X) = (m_P(d,X))_{d \in \mathcal{D}_0}$.
That is,
\begin{align*}
    \mathcal{M}_P 
    &= \{ m : S m(\cdot,X) \leq r, \; F  m(\cdot,X) = m_{0,P}(\cdot, X), \; P-\text{a.s.}\}
\end{align*}
which allows the empirical analog $\hat{\mathcal{M}}$ to be expressed as
\begin{align*}
    \hat{\mathcal{M} }
    &= \{m : S m(\cdot,X_i) \leq r \; \forall i, F  m(\cdot,X_i) = \hat m_0(\cdot,X_i) \; \forall i\}
\end{align*}
where $\hat m_0(\cdot,x)' = (\hat m_0(d, x))_{d\in\mathcal{D}_0}$.
By expressing $\hat{\mathcal{M}}$ this way it is possible to express 
$\hat \Gamma_j(X_i)$ as the maximum of $J$ linear programs.
Define the estimate
\begin{align*}
    \hat \Gamma_{jk}(X_i) := \max_{m \in \hat{\mathcal{M}}} v_m(d_k, X_i) - v_m(d_j, X_i)
\end{align*}
which measures the contribution to expected regret of assigning an individual $d_j$ 
instead of assigning them $d_k$, conditional on $X_i$.
Recall that $b(d,x)$ and $c(d,x)$ parametrize the linear utility function.
For each observation $i$ and treatments $j$ and $k$, construct $b_{jk}(X_i)$ as a vector 
in $\R^J$ with $b(d_k,X_i)$ in the $k^{th}$ entry and $b(d_j,X_i)$ in the $j^{th}$ entry, 
and zeros everywhere else.  
Construct $c_{jk}(X_i)$ likewise.  
Then
\begin{align}
    \label{gammalp}
    \begin{split}
    \hat \Gamma_{jk}(X_i) = \max_{m \in \R^J} \quad   b_{jk}(X_i)' m - & c_{jk}(X_i)
    \\ s.t. \quad S m &\leq r
    \\  Fm &= \hat m_0(\cdot,X_i).
    \end{split}
\end{align}
After defining $\hat \Gamma_j(X_i) = \max_k \hat\Gamma_{jk}(X_i)$, 
these estimates can be used in the program (\ref{empiricalMMR}).

Despite the linear programming representation of $\hat \Gamma_{jk}(X_i)$, 
computing $\hat \Gamma_j(X_i)$ for all $i$ and $j$ appears to require 
$N J^2$ linear programs total (one for each $i$, $j$, and $k$ combination).
However, a dual formulation detailed in Appendix \ref{appendix-computation} 
demonstrates that $\hat\Gamma_{j}(X_i)$ can be computed with a single 
linear program.
Moreover, the linear programs for all observations can be stacked together and 
solved simultaneously.
This can be done prior to or in conjunction with the policy optimization over $\pi$.
Additionally, in the case of discrete covariates, it is only necessary to 
construct $\hat \Gamma_j(X_i)$ for unique values of $X_i$, which may be 
substantially smaller than the sample size $N$.

Having computed the covariate-level loss estimates $\hat \Gamma_j(X_i)$
that appears in (\ref{empiricalMMR}), optimization of the policy $\pi$ 
over the set $\Pi$ can be performed according to established methods in policy learning.  
In many cases, $\Pi$ can be represented by linear and integer constraints.  
Examples include linear eligibility scores, decision trees, and treatment sets with 
piecewise linear boundaries 
(\cite{kitagawaWhoShouldBe2018},
\cite{mbakopModelSelectionTreatment2021},
\cite{zhouOfflineMultiActionPolicy2023}).
When this is the case, the problem (\ref{empiricalMMR}) is a mixed integer-linear program 
for which highly optimized solvers are readily available.  
In Section \ref{section-application}, I use a class of linear eligibility 
score policies, which is described using linear and integer constraints in 
Appendix \ref{appendix-computation}.  

Since optimization of $\pi$ using mixed integer-linear programming is standard practice 
in policy learning problems, the only additional computational burden resulting from 
considering new treatments is that of solving the linear programs corresponding to 
$\hat \Gamma_j(X_i)$ as described above.  
For the example in Section \ref{section-application} I found the computation 
time for $\hat \Gamma_j(X_i)$ to be at most a similar order of magnitude as that of the 
estimation of $\hat \pi$ and sometimes much shorter, depending on the complexity of the 
policy class.
Constructing $\hat \Gamma_j(X_i)$ can often benefit from parallelization so that the 
overall computational burden is not much larger than the point identified case.

\section{Regret convergence}
\label{section-convergence}
In this section I investigate theoretical guarantees on the performance of the 
estimated policy $\hat \pi$.
Following \textcite{manskiStatisticalTreatmentRules2004a}, I evaluate the performance of 
policies in terms of their statistical regret.
In particular, I show that the regret of the estimated policy $\hat \pi$ 
converges to the lowest possible maximum regret the decision maker could achieve
if the population identified set under distribution $P$ were observed, uniformly across $P$.
Specifically, the regret guarantees will be of the form
\begin{align*}
    \sup_P \left(\E_P[\overline{R}_P(\hat\pi)] - \overline{R}_P(\pi^*_P) \right)
        \leq \mathcal{O}(N^{-1/2} \vee \rho_N^{-1})
\end{align*}
where $P$ ranges across an appropriate set defined below, implying that
\begin{align*}
    \sup_P \E_P[\overline{R}_P(\hat \pi)] \leq \sup_P \overline{R}_P(\pi^*_P) + \mathcal{O}(N^{-1/2} \vee \rho_N^{-1})
\end{align*}
for an appropriate sequence $\rho_N \to \infty$.
Since the estimated policy $\hat \pi$ is constructed using consistent estimates of the
partially identified set, these guarantees will generally be asymptotic in nature.
Above, the expectation only averages across realizations of the estimator $\hat \pi$ because
$\overline{R}_P(\cdot)$ is defined using the population probability measure $P$.

The interpretation of this bound is that in large samples the performance of the estimated policy
$\hat \pi$ as measured by its maximum regret across distributions $P$
(and the identified set under $P$) approaches the performance of the population optimal
policy $\pi^*_P$.
This bound on the difference between the maximum regret of $\hat \pi$ and the
best-in-class policy $\pi^*_P$ is similar to the bounds often obtained in the 
empirical welfare maximization or empirical risk minimization literature
in the point-identified case
(e.g. \textcite{kitagawaWhoShouldBe2018})
after replacing (unidentified) welfare with maximum regret.
Since $0 \leq \overline{R}_P(\pi^*_P) \leq \overline{R}_P(\hat\pi)$, 
this means that averaging across realizations of the estimate $\hat \pi$, 
the worst-case expected regret of $\hat \pi$ is growing arbitrarily close to 
$\overline{R}_P(\pi^*_P)$, the lowest possible maximum regret the decision maker
could achieve in the absence of sampling uncertainty.
In general no policy $\pi$ can achieve zero maximum regret across the entire identified
set, resulting in $\overline{R}_P(\pi^*_P) \geq 0$ typically holding with strict inequality for the
population optimal $\pi^*_P$.

I now discuss assumptions sufficient for such guarantees.
The main assumptions on the joint distribution of the data are random assignment 
of treatment and boundedness of the components of utility.

\begin{assumption}
    \label{DGP}
    $\{(Y_i, D_i, X_i)\}_{i=1}^n$ are i.i.d. copies of $(Y, D, X)$, generated by
    $P$ which satisfies
    \begin{enumerate}
        \item $D \perp Y(d) \mid X$ for all $d \in \mathcal{D}$
        \item There exists $C < \infty$ such that $m_P(d, X)$, $b(d,X)$, and $c(d,X)$ 
            are all bounded in absolute value by $C$ almost surely, for each $d\in\mathcal{D}$
    \end{enumerate}
\end{assumption}

Assumption \ref{DGP}.1 reflects the standard exogeneity condition that holds in 
the randomized experiment settings I use as a motivating example.  
It may also hold in observational studies, in which case it may be a strong assumption.  
In many randomized experiments the stronger condition $D \perp (Y(d), X)$ is satisfied.  
When this is true, it is sufficient to estimate $m_P(d,x) = \E_P[Y(d) \mid \tilde X=x]$, 
where $\tilde X$ is a subset of covariates $X$ which directly enter the policy.  
$\tilde X$ may be of a much lower dimension than $X$ since policies are often 
restricted to be relatively simple (\cite{kitagawaWhoShouldBe2018}).  
In Section \ref{section-application}, there are two covariates which enter 
the policy.  
Henceforth, I do not distinguish between the covariates required for 
Assumption \ref{DGP}.1 and the covariates used for the policy.
Assumption \ref{DGP}.2 restricts decision maker preferences 
by requiring that the mean response function is bounded, 
as well as the parameters of the linear utility function.
In the example of Section \ref{section-application} the outcome $Y$ is 
bounded while $b(d,x)$ is constant in $x$ and $c(d,x)=0$, satisfying this condition 
trivially.

The estimate $\hat m$ also must be sufficiently accurate in the following sense
\begin{assumption}
    \label{rootnconvergence}
    For some sequence $\rho_N \to \infty$ and some class of distributions $\mathcal{P}$, 
    the estimate $\hat m$ satisfies 
    \begin{enumerate}
        \item $\lim \sup_{N \to \infty} \sup_{P \in \mathcal{P}}
            \rho_N \E_P\bigg[ \frac{1}{N} \sum_{i=1}^N \lVert \hat m_0(\cdot,X_i) - m_{0,P}(\cdot,X_i) \rVert \bigg] < \infty$
        \item $\hat{\mathcal{M}} = \{ m : S m(\cdot,X) \leq r, \; F m(\cdot,X) = \hat m_0(\cdot,X)\}$ is nonempty, almost surely,
            for all $P \in \mathcal{P}$.
    \end{enumerate}
\end{assumption}
One common setting in which Assumption \ref{rootnconvergence}.1 holds with 
$\rho_N = N^{1/2}$ is when the covariates are discrete and sample averages may be used.  
Alternatively, $m_P(d, \cdot)$ may be assumed to belong to a parametric family, for each 
value of $d \in \mathcal{D}_0$.  Since $m_P(d,\cdot)$ is identified holding $d\in\mathcal{D}_0$ 
fixed, parametric assumptions on the relationship between covariates and the outcome of 
interest conditional on treatment values observed in the data may be weaker assumptions 
than the kinds of parametric assumptions that would allow one to extrapolate to new treatments,
in the sense that the former are testable.
\textcite{kitagawaWhoShouldBe2018} provides more general conditions under which 
Assumption \ref{rootnconvergence}.1 is satisfied when $\hat m$ is constructed via local 
polynomial regression.

Assumption \ref{rootnconvergence}.2 is not a restrictive assumption.  
Unless $m_P$ is on the boundary of $\mathcal{S}$, the estimated set $\hat{\mathcal{M}}$ 
will typically be nonempty with high probability as $N$ grows even if 
\ref{rootnconvergence}.2 is not assumed. 
In finite samples, an estimator that yields an empty $\hat{\mathcal{M}}$ can be 
projected onto the set of all $\hat m$ such that $\hat{\mathcal{M}}$ is nonempty.  
Since $\hat m_0(\cdot, X)$ is a vector in $\R^{J_0}$ and $\mathcal{S}$ 
is described by linear inequalities, this is a convex minimum norm problem that
can be solved by quadratic programming.

Finally, the choice set $\Pi$ is assumed to satisfy a standard condition on its complexity.
\begin{assumption}
    \label{policy}
    For each $d \in \mathcal{D}$, the class of sets 
    $
        \{ x: \pi(x) = d, \; \pi \in \Pi\}
    $
    is a VC-class of sets with VC dimension at most $V < \infty$.
\end{assumption}
For a formal definition of the VC dimension, see 
\textcite{vandervaartWeakConvergence1996}.
The assumption of finite VC dimension limits the complexity of the class $\Pi$; 
specifically, Assumption \ref{policy} ensures that $\Pi$ cannot be so flexible 
as to assign any arbitrary subset of a collection of $V+1$ points in $\mathcal{X}$ 
to treatment $d$.
This assumption is commonly invoked in offline policy learning settings as a way 
to express the constraints faced by decision makers (\cite{kitagawaWhoShouldBe2018}); 
this may be for the sake of interpretation, fairness, ease of implementation, 
political constraints, etc.  
The types of rules discussed in Section \ref{section-estimation} which can be 
expressed using linear and integer constraints, like linear eligibility scores and 
decision trees, satisfy this assumption under bounds on the number of inputs to the 
eligibility score or the depth of the decision tree.  
The assumption of VC dimension also plays an important role in the convergence of 
the regret of the optimal policy by ensuring the policy does not overfit the sample data.  
This assumption can be relaxed by instead using a holdout validation sample which
regularizes estimation of the policy (\cite{mbakopModelSelectionTreatment2021}).

Under these assumptions, the following regret bound is obtained:
\begin{theorem}
    \label{regretconvergence}
    Let $\mathcal{P}_C$ be a set of distributions for which  (1) Assumptions \ref{DGP} holds with
    constant $C$ and (2) Assumption \ref{rootnconvergence} holds.
    Under Assumptions \ref{linear} and \ref{policy},
    \begin{align*}
        \sup_{P\in\mathcal{P}_C} 
        \left(\E_P[\overline{R}_P(\hat\pi)] - \overline{R}_P(\pi^*_P) \right)
        \leq \mathcal{O}(N^{-1/2} \vee \rho_N^{-1})
    \end{align*}
    is satisfied.  As a result,
    \begin{align}
        \sup_{P\in\mathcal{P}_C} \E_P[\overline{R}_P(\hat \pi)] \leq \sup_{P\in\mathcal{P}_C} \overline{R}_P(\pi^*_P) + \mathcal{O}(N^{-1/2} \vee \rho_N^{-1})
    \end{align}
\end{theorem}

The rate of convergence of the maximum regret is the slower of two rates:
$N^{-1/2}$, and the estimation rate of $\hat m_0$ in Assumption \ref{rootnconvergence}.  
This first rate is driven by the convergence of an empirical process uniformly over the 
policy class, which is $N^{-1/2}$ under Assumption \ref{policy} (\cite{vandervaartWeakConvergence1996}).  
The second rate reflects that
the regret of the estimated policy depends on the behavior of the linear program (\ref{gammalp}), 
the constraints of which depend on identified moments of the data and must be estimated.  
In turn, the value of the linear program can be shown to converge to its population counterpart 
at the same rate as the constraints 
(see \textcite{hoffmanApproximateSolutionsSystems1952}
and \textcite{rockafellarVariationalAnalysis2009};
related results in econometrics include 
\textcite{fangInferenceLargeScaleLinear2023} and
\textcite{freybergerIdentificationShapeRestrictions2015}).
Because Assumption \ref{rootnconvergence} is only a condition on the rate of convergence of
this estimator, the bound of Theorem \ref{regretconvergence} is a rate result.
If non-asymptotic bounds on the estimator $\hat m_0$ are available,
(for example if covariates are discrete and outcomes are bounded),
then the proof of Lemma \ref{lpconvergence} can be used to obtain non-asymptotic regret bounds.
This is developed in more detail in Appendix \ref{appendix-extensions}.

I give a heuristic sketch of the proof and defer the details to Appendix \ref{appendix-proofs}.  
I first define the quantities
\begin{align*}
    \tilde{R}_{N,P}(\pi) &:= \frac{1}{N}\sum_{i=1}^N  \sum_{j=1}^J \pi_{ij}\Gamma_{j,P}(X_i)
    \\
    \overline{R}_{N}(\pi) &:= \frac{1}{N}\sum_{i=1}^N \sum_{j=1}^J \pi_{ij} \hat\Gamma_j(X_i)
\end{align*}
$\tilde R_{N,P}(\pi)$ measures the in-sample or empirical maximum regret of policy $\pi$, 
supposing the true $m_P$ and hence the true $\Gamma_{j,P}$ were known.  
$\overline{R}_{N}(\pi)$ is the objective function of the empirical minimax regret problem 
(\ref{empiricalMMR}).
The difference between the maximum regret of the estimated policy and that of the 
minimax regret optimal policy can then be decomposed in terms of these quantities as follows:

\begin{align}
    \label{decomposition}
    \begin{split}
    0 \leq \overline{R}_P(\hat\pi) - \overline{R}_P(\pi^*) 
        &= \overline{R}_P(\hat\pi) - \tilde{R}_{N,P}(\hat\pi)
        \\ &\quad + \tilde{R}_{N,P}(\hat\pi) - \overline{R}_{N}(\hat\pi)
        \\ &\quad + \overline{R}_{N}(\hat\pi) - \overline{R}_{N}(\pi^*)
        \\ &\quad + \overline{R}_{N}(\pi^*) - \tilde{R}_{N,P}(\pi^*) 
        \\ &\quad + \tilde{R}_{N,P}(\pi^*) - \overline{R}_P(\pi^*)
    \end{split}
\end{align}

The first and last lines of (\ref{decomposition}) each concern the difference between a 
sample mean and the population expectation, holding the policy and distribution fixed and assuming 
$\Gamma_{j,P}(X)$ is known.  
They are each bounded by
$$
\sup_{\pi\in\Pi} \bigg|  \overline{R}_P(\pi) 
- \tilde{R}_{N,P}(\pi) \bigg|
$$
The second and fourth lines of (\ref{decomposition}) concern the difference between 
sample means of the true quantities $\Gamma_{j,P}(X_i)$ and their estimated counterparts, 
holding the policy and distribution fixed.  
They are each bounded by
$$
    \sup_{\pi\in\Pi} \bigg| \tilde{R}_{N,P}(\pi) 
    -  \overline{R}_N(\pi)\bigg|
$$
The third line of (\ref{decomposition}) concerns the difference between the in-sample 
performances of $\hat \pi$ and $\pi^*$. 
This is always negative because $\hat\pi$ is optimal for the empirical minimax regret 
problem (\ref{empiricalMMR}). 
Hence, the decomposition (\ref{decomposition}) yields
\begin{align} 
    \overline{R}_P(\hat\pi) - \overline{R}_P(\pi^*)
    &\leq 2\sup_{\pi\in\Pi} \bigg|  \overline{R}_P(\pi) 
    - \tilde{R}_{N,P}(\pi) \bigg|
    \label{firsttermregret}
    \\
    &\quad + 2\sup_{\pi\in\Pi} \bigg| \tilde{R}_{N,P}(\pi) 
    -  \overline{R}_N(\pi)\bigg|
    \label{secondtermregret}
\end{align}

Term (\ref{firsttermregret}) is the sup-$\Pi$ norm of a centered empirical process.  
Its expectation can be shown to converge uniformly at $N^{-1/2}$ rate using techniques 
in empirical process theory.
\begin{lemma}
    \label{eplemma}
    Under assumptions \ref{linear}, \ref{DGP}, and \ref{policy},
    $$
    \sup_{P\in\mathcal{P}_C}
        \E_P\bigg[\sup_{\pi\in\Pi} \bigg| \overline{R}_P(\pi) 
        - \tilde{R}_{N,P}(\pi) \bigg| \bigg]
        \leq K \sqrt{\frac{V}{N}}
    $$
    for some constant $K$ depending only on $C$ and $J$.
\end{lemma}
The constant $K$ hides a dependence on the number of treatments $J$.  
This dependence represents a cost to introducing arbitrarily large sets of new treatments.  
Just as Assumption \ref{policy} restricts the complexity of the sets of covariate values assigned 
to each treatment, the assumption of a fixed $J$ represents an exogenous constraint on the 
overall complexity of the policy.  

Term (\ref{secondtermregret}) concerns the difference between the value of the linear program 
defining $\hat\Gamma_{jk}(X_i)$, in which the constraints are estimated, versus the linear 
program defining $\Gamma_{jk,P}(X_i)$, in which the true value of the constraint vector is used.  
When the estimated constraints converge at $N^{-1/2}$ rate, the value of the linear program 
can be shown to exhibit similar convergence uniformly across $\Pi$.  
More generally, the value of the linear program converges at the same rate as the estimated constraints. 
This is because the feasible set of a linear program is Lipschitz in its constraints with respect 
to the Hausdorff metric.
\begin{lemma}
    \label{lpconvergence}
    Under assumptions \ref{linear}, \ref{DGP}, and \ref{rootnconvergence},
    $$
    \sup_{P\in\mathcal{P}_C}\E_P\bigg[\sup_{\pi\in\Pi} \bigg|\tilde{R}_{N,P}(\pi) 
    -  \overline{R}_N(\pi) \bigg| \bigg]\leq \mathcal{O}(\rho_N^{-1})
    $$
\end{lemma}
Taking the expectation of the bound given by (\ref{firsttermregret}) and (\ref{secondtermregret}) and combining this with
Lemmas \ref{eplemma} and \ref{lpconvergence} yields the bound of Theorem \ref{regretconvergence}.

\section{Application to rural electrification}
\label{section-application}
Investment in energy infrastructure is an important focus of development aid 
and there is a large body of research in development economics devoted to its study 
(reviews include 
\cite{leeDoesHouseholdElectrification2020},
\textcite{petersImpactsRuralElectrification2016},
and
\textcite{vandewalleLongTermGainsElectrification2015}).
\textcite{leeExperimentalEvidenceEconomics2020}
examines the relationship between the price of connections 
to the electrical grid and takeup in rural Kenya.  
This particular setting provides a compelling use case for the procedure outlined in this paper.  
There are only four prices observed in the data, leading to substantial model ambiguity in the 
form of partial identification of the demand curve outside these four prices. 
Further, the treatments are subsidies valued at hundreds of US dollars,
making subsequent experimentation with new treatments expensive.  
In this section, I take experimental data collected to study the economics of rural 
electrification 
(\cite{leeDataArchiveExperimental2020})
and illustrate how the method outlined in the present paper can be used to 
design cost-effective targeted subsidy policies to maximize household takeup.

Prices of $d$-thousand Kenyan shillings for $d \in \mathcal{D}_0 = \{0, 15, 25, 35\}$ are 
randomly offered to households, who have an eight-month period in which to decide whether 
to purchase the connection at the offered price.  
After the period is over, households continue to have the option to connect at the full price 
of $35$ thousand shillings.  
Here $D$ is price, $Y$ is a takeup indicator, and $X$ is a two-dimensional random vector 
containing household size and income.

Given this experimental data, I consider a decision maker able to offer subsidies to households.  
However, the decision maker has no reason to restrict themselves to the four prices that appear 
in the data.  
In my baseline analysis, I examine an expanded treatment set of 
$\mathcal{D} = \{0, 2.5, 5, \dots, 35\}$ thousand shillings.  
The sensitivity of results to coarser and finer treatment sets is reported in 
Appendix \ref{appendix-robustness}.
I assume the decision maker values each connection at $\alpha$-thousand Kenyan shillings 
and must pay the value of the subsidy if the recipient purchases a connection.  
There is no fixed cost for offering the subsidy.  
This means $u(d,x,y) = (\alpha - (35 - d)) y $ so that $b(d,x) = \alpha - (35 - d)$ and $c(d,x) = 0$.  
As a baseline specification, I take $\alpha$ to be the full market price of 
$35$ thousand shillings and explore policies under other valuations in 
Appendix \ref{appendix-robustness}. 

\subsection{Optimal policy without covariates}

Before estimating the optimal policy mapping covariate values to prices,
I illustrate the method for the simple case of no covariates.
Ignoring covariates for the time being makes the process easier to visualize,
and transparently demonstrates how combining the experimental data, 
shape restrictions, and the minimax regret criterion drives 
the choice of whether and how to implement new treatments.  
In the next subsection, where I consider policies which target prices
on the basis of covariates,
the worst-case regret is computed for each covariate value similarly
to the no-covariate case of this subsection.

I first estimate the average takeup at each price.  
Using these first stage estimates, I construct bounds for the effects 
of each new treatment and explain the difference between these pointwise 
bounds on outcomes and the estimated identified set $\hat{\mathcal{M}}$.  
Then I consider a fixed policy which assigns a single price to the entire 
population and find the regret-maximizing demand curve.  
I find the minimax regret optimal policy by finding the policy for which 
the maximum regret is as small as possible.

\begin{figure}
    \center
    \caption{Takeup and utility by price}
    \includegraphics[width=0.49\textwidth]{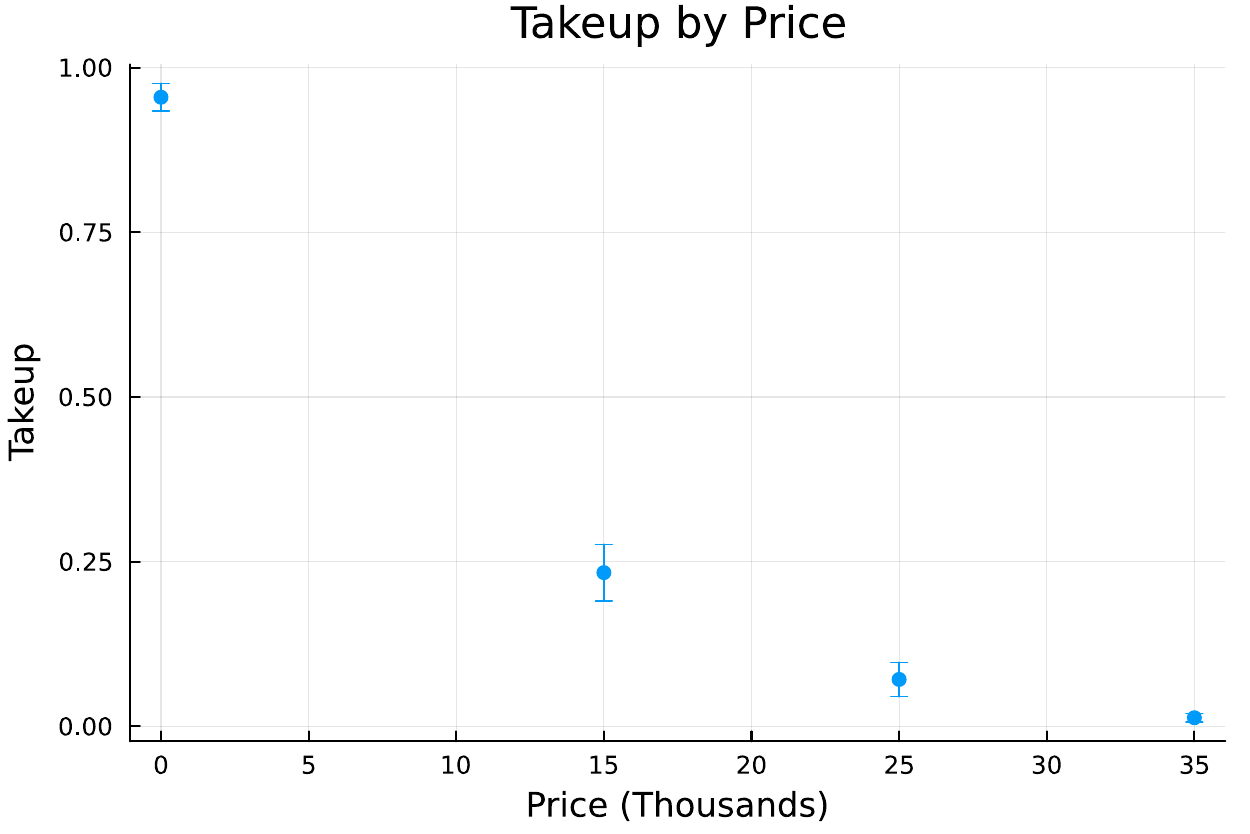}
    \includegraphics[width=0.49\textwidth]{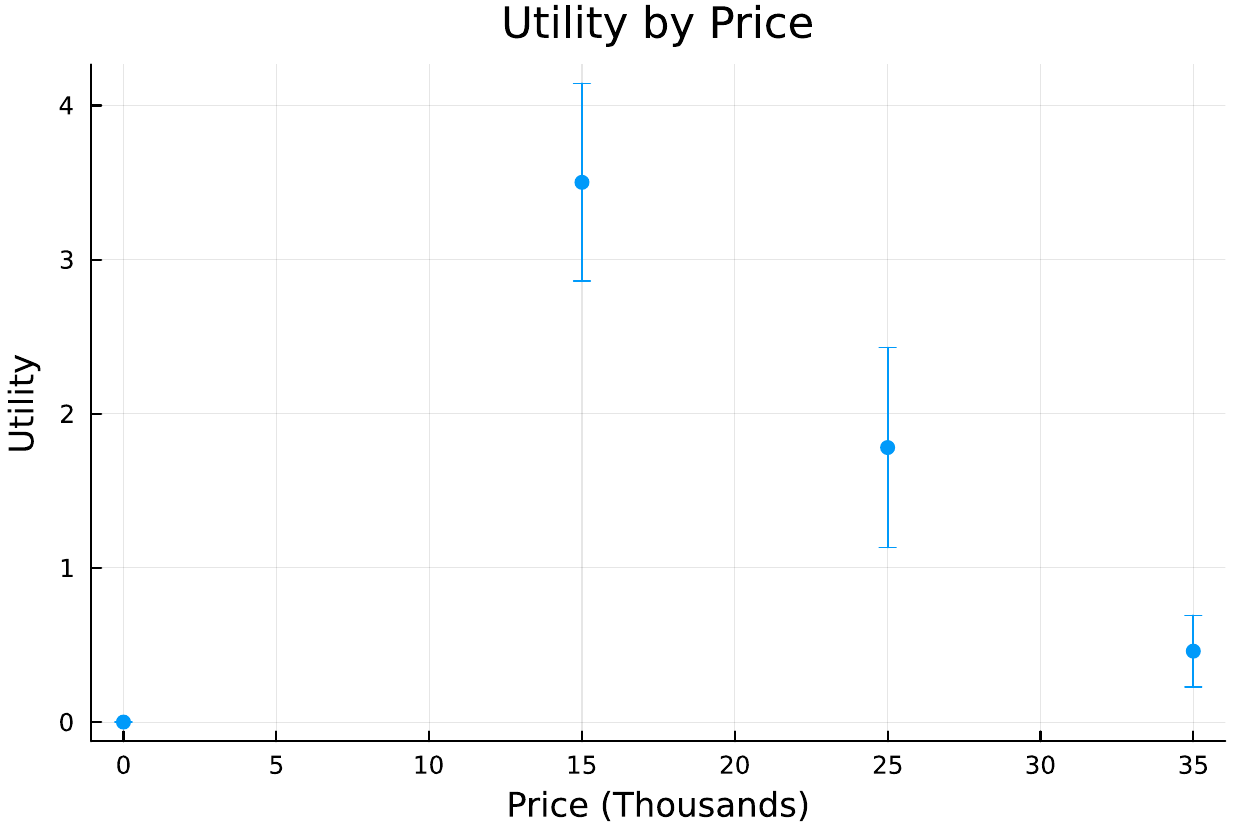}
    \label{experimentalresults}
    \\
    \small \raggedright
    Mean estimated takeup and utility at each price included in the experiment of 
    \textcite{leeExperimentalEvidenceEconomics2020}.  $95\%$ confidence intervals shown in bars.
\end{figure}

For each $d$ in the experimental data, I plot the mean takeup and utility in 
Figure \ref{experimentalresults}.  
Mean takeup $m_{0,P}(d) = \E_P[Y(d)]$ is identified from the experimental data for 
$d\in\mathcal{D}_0$, and estimated mean takeup $\hat m_0(d)$ is simply the 
sample mean at each price.  
Expected utility for experimental subsidy values is given by 
$v_{m_P}(d) = (\alpha - (35 - d)) m_P(d)$ and is estimated for $d \in \mathcal{D}_0$ 
by plugging in $\hat m_0(d)$. 
The price $d=0$ represents a fully subsidized connection, which is clearly 
undesirable from the decision maker's perspective because the decision maker 
will receive $0$ utility, which is the minimum possible, regardless of whether 
the household connects.  
Amongst the treatment values that appear in the data, $d=15$ achieves the 
highest utility on average.  
While not shown here, this is largely true of estimated mean utility conditional 
on $X$ as well.  
Indeed, setting $\mathcal{D} = \mathcal{D}_0$ and solving the empirical welfare 
maximization problem as in \textcite{kitagawaWhoShouldBe2018} with a linear eligibility 
score as the policy class assigns all individuals to a price of $15$.

A key question from the decision maker's perspective is whether prices not in the 
support of $D$ in the data could yield higher utility, and how data from the 
experiment can provide information on the magnitude of such gains.  
To answer this, I impose shape restrictions which imply bounds on takeup at new prices.  
The shape restrictions I study here are that demand is downward sloping and the price 
subsidy exhibits diminishing returns.
Takeup is also bounded between zero and one. 
Downward sloping demand is expected to be satisfied in all but a few exceptional markets, 
and represents one of the weaker assumptions a researcher may impose.  
Diminishing sensitivity to treatment may be more context specific, and can be motivated 
by a simple binary choice model where the density of valuations is decreasing on the 
support of treatments.  
Another setting where such a restriction may be applied is the analysis of production functions 
(\cite{manskiMonotoneTreatmentResponse1997}). 
The shape restrictions I impose, which can be expressed as linear inequalities involving 
the $J$-dimensional vector $m$ as shown in Appendix \ref{appendix-computation}, 
define the constraint $S m \leq r$ in the linear program (\ref{gammalp}).

\begin{figure}
    \center 
    \caption{Bounds on takeup and utility at each price}
    \includegraphics[width=0.49\textwidth]{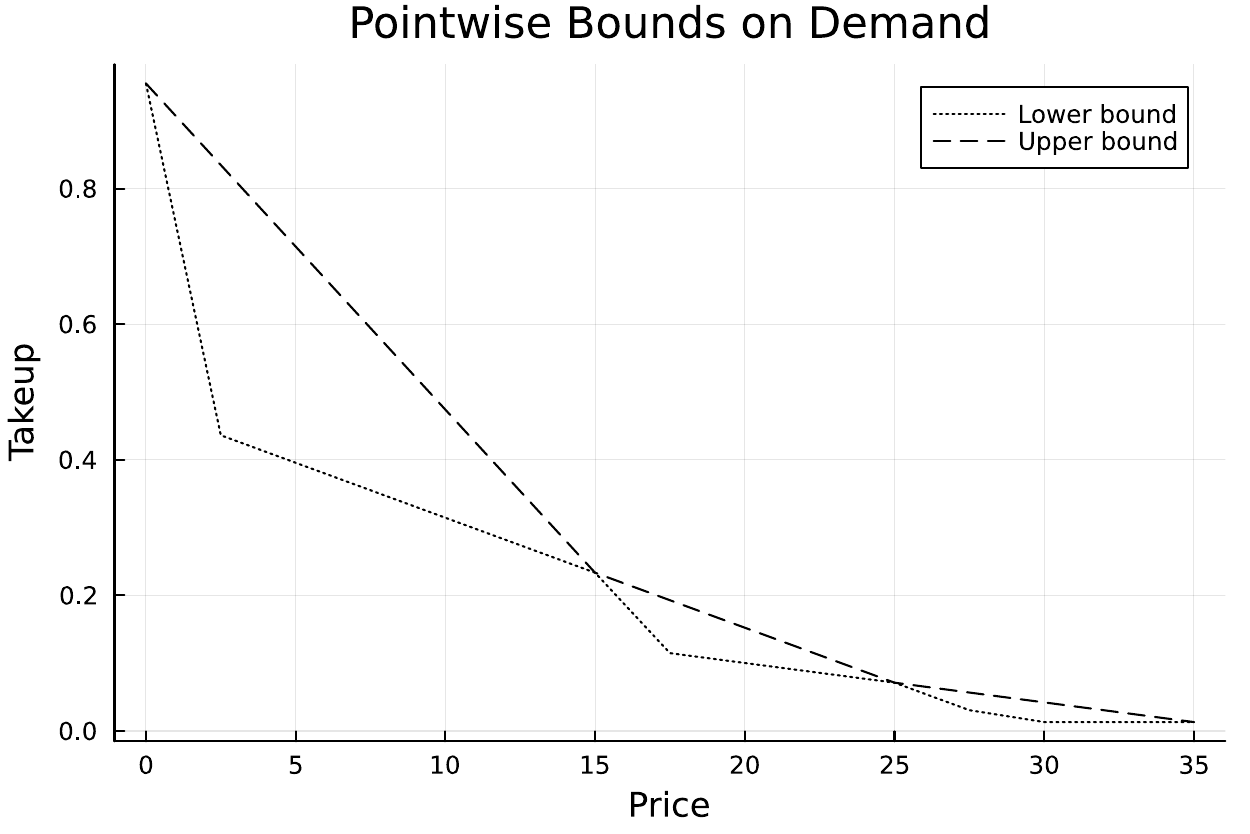}
    \includegraphics[width=0.49\textwidth]{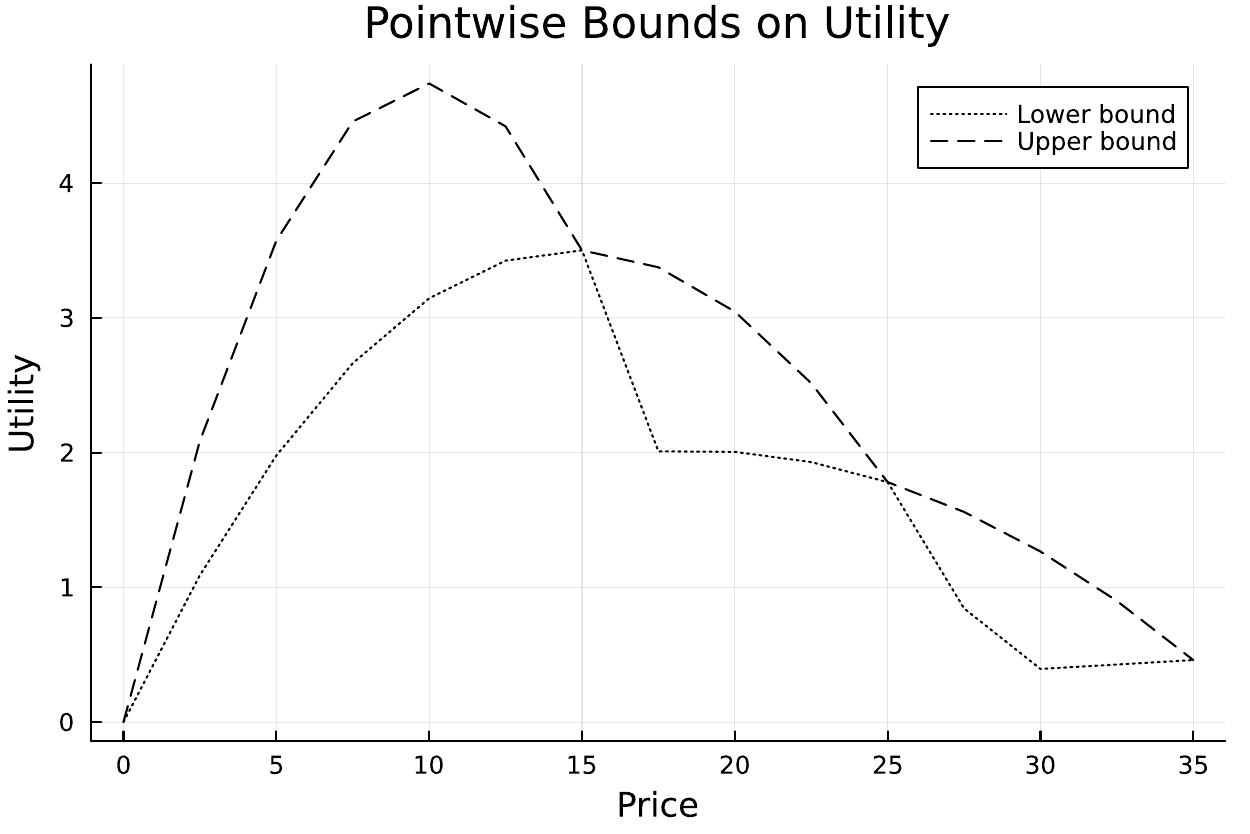}
    \label{worstcaseexample}
    \\
    \small \raggedright
    The maximal and minimal possible expected takeup at each price, 
    and the corresponding bounds on expected utility generated by these bounds on takeup.
\end{figure}

To explore the potential effects of new treatments informally in the simple case 
of no covariates, in Figure \ref{worstcaseexample} I plot pointwise upper and lower bounds 
on takeup at each possible price.  
These are obtained by calculating $\min_{m \in \hat{\mathcal{M}}} m(d)$ and 
$\max_{m \in \hat{\mathcal{M}}} m(d)$ for each $d$.  
Note that the lower bound is not convex.  
This is an illustration of the non-rectangularity induced by the shape restrictions-- 
there is no $m \in \hat{\mathcal{M}}$ that simultaneously minimizes takeup for all prices $d$. 
More generally, not every curve that lies within the pointwise bounds of Figure \ref{worstcaseexample} 
satisfies the shape restrictions.  
This can be expressed formally as 
\begin{align*} 
    \hat{\mathcal{M}}
    \subset
    \bigg\{\tilde m \in \R^J : \min_{m \in \hat{\mathcal{M}}} m \leq \tilde m_j \leq \max_{m \in \hat{\mathcal{M}}} m, \forall j \bigg\} 
\end{align*}
with strict containment.  
This difference is key for the informativeness of the linear program (\ref{gammalp}) 
because regret is defined by comparing the outcomes under the chosen policy 
to those of the first-best policy under the same demand curve $m$.  
If the chosen policy achieves low utility for one demand curve and the first-best policy 
achieves high utility only for a different demand curve, this does not contribute to high regret.

Along with the bounds on takeup in Figure \ref{worstcaseexample}, I also plot the bounds on 
expected utility $v_{m}$ generated by the bounds on takeup.  
These curves illustrate a range of possible outcomes that may result from implementing new treatments. 
The upper bounds on utility illustrate the potential for much better outcomes as a result of 
implementing new treatments, especially in the range of $7.5$ to $12.5$.  
The lower bounds imply the possibility of worse outcomes as well.  
A maximin welfare approach to this problem would not assign a price in 
$\mathcal{D} \setminus \mathcal{D}_0$ to anyone for whom that price was not guaranteed to outperform 
the prices in $\mathcal{D}$.  
This ends up assigning a price of $15$ to the entire sample, 
which seems excessively conservative in this example\footnote{
    It is not generally true that the maximin welfare policy restricts to the original set of 
    treatments.
    See Appendix \ref{appendix-lipschitz} for an example in which the maximin welfare policy
    assigns a new treatment to the population.
}.
On the other hand, the minimax regret approach considers losses relative to the ex-post 
optimal decision in each state of the world represented by $m \in \hat{\mathcal{M}}$.

Given the bounds in Figure \ref{worstcaseexample}, one could imagine naively constructing 
$\hat \Gamma_{jk}$ by comparing the worst possible $v_{m}(d_j)$ to the best possible 
$v_{m}(d_k)$. 
For example, taking $d_k = 7.5$ and $d_j = 10$ would result in an estimate of about $4.5 - 3.2 = 1.3$.  
However, recall that these bounds on $v_{m}$ were constructed from the bounds on $m$.  
Observing the bounds on $m$, it can be seen that the demand curve $m$ which achieves 
maximal takeup at $d=7.5$ and minimal takeup at $d=10$ is not convex, and thus the regret estimate 
obtained by comparing the pointwise bounds is unnecessarily pessimistic.  
Likewise, taking $d_k = 12.5$ and $d_j = 10$ and comparing the pointwise bounds would yield an estimate 
of about $4.4 - 3.2 = 1.2$, but a demand curve which achieves these bounds is not decreasing.  
Formally, 
$\max_{m} [v_{m}(d_k) -  v_{m}(d_j)] \leq \max_{m} v_{m}(d_k) - \min_{m} v_{m}(d_j)$. 
Thus, it is necessary to construct the regret estimates by finding a demand curve $m$ 
which maximizes regret while satisfying the shape restrictions.
This illustrates that the linear program (\ref{gammalp}) defining $\hat \Gamma_{jk}$, 
while requiring more computations than pointwise bounds for each $d_j \in \mathcal{D}$, 
carries additional useful information.

\begin{figure}
    \center
    \caption{Regret-maximizing demand curves for alternative policies}
    \includegraphics[width=0.49\textwidth]{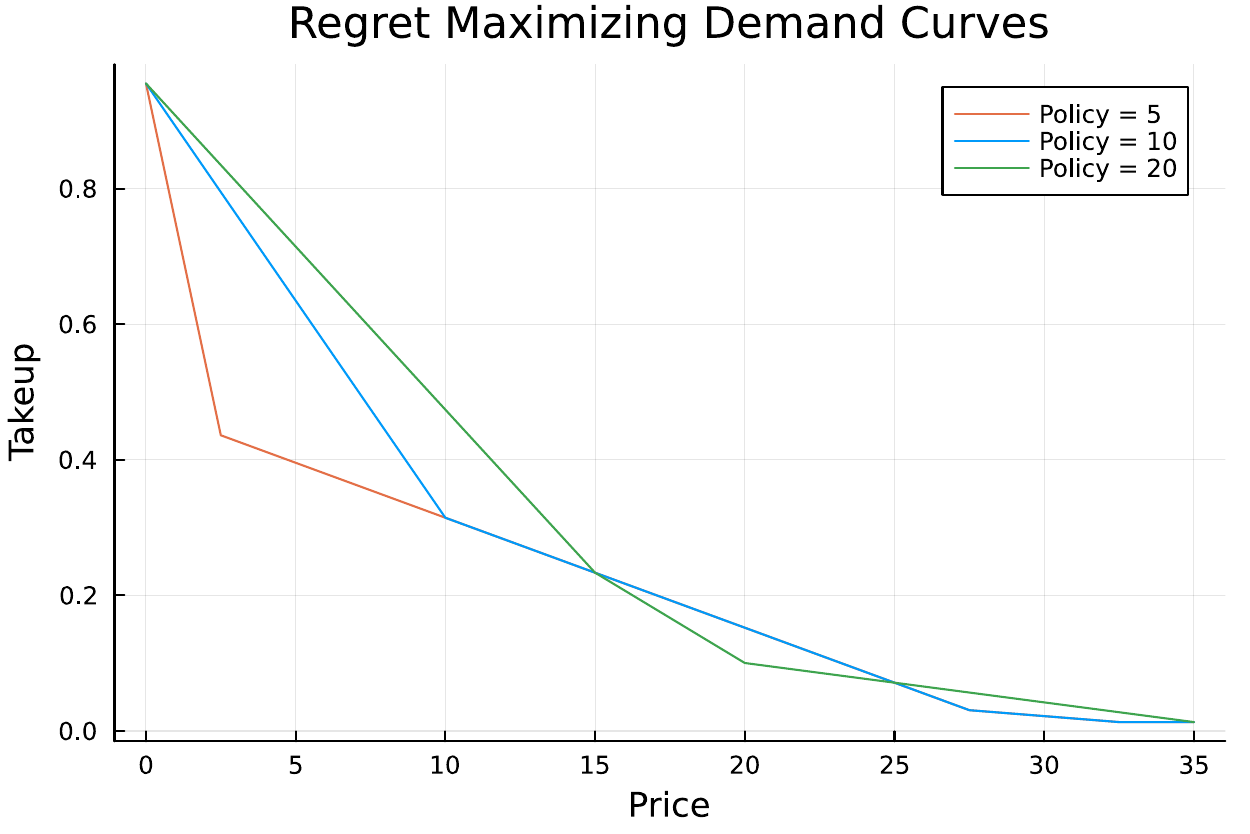}
    \includegraphics[width=0.49\textwidth]{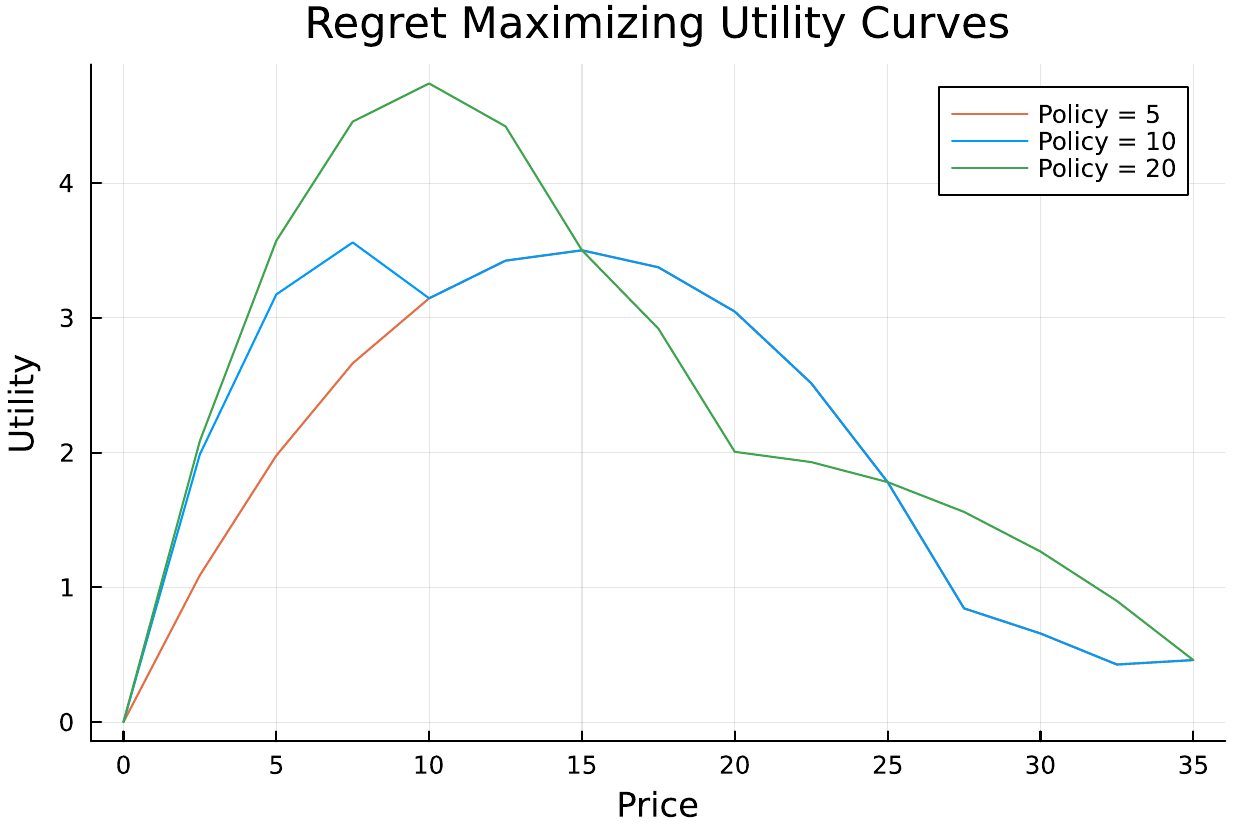}
    \label{maxregret}
    \\
    \small \raggedright
    The demand curves that maximize regret for each of three policies-- 
    assigning $d=5$ to everyone, assigning $d=10$ to everyone, 
    and assigning $d=20$ to everyone.  
    Also plotted are the utility curves generated by these regret-maximizing demand curves.  
    Given a policy, the regret-maximizing demand curve is chosen to achieve low utility 
    at the chosen price but high utility elsewhere. 
\end{figure}

To understand how maximal regret is computed for each policy, 
I plot regret-maximizing demand curves for each of three different policies in Figure \ref{maxregret}.  
In this case, a policy is a single value of $d$ that will be assigned to the entire population. 
The regret-maximizing demand curve is the vector $m$ which solves (\ref{gammaj}) 
with no covariates.  
To compute it, I solve (\ref{gammalp}) for each $k$ and find the $m$ corresponding 
to the optimal $k$. 
Supposing the decision maker assigns a price of $d=5$ to the entire population, 
the regret-maximizing demand curve is chosen to yield low expected utility when $d=5$ 
but high utility for some other price, 
thus incurring high regret in the sense that the chosen policy of $d=5$ was ex-post 
a poor policy compared to, say, a price of $d=15$.  
The same process is enacted for the policies which assign $d=10$ to the entire population 
and $d=20$ to the entire population.  
Under the policy $d=20$, regret is very high because the difference between 
expected utility at $d=20$ and the optimal expected utility 
under the regret-maximizing demand curve is very large.  
Comparatively, the maximum regret incurred under the policy $d=10$ is small.  
Importantly, the regret-maximizing demand curves which generate these worst-case utility curves 
obey the shape restrictions, as can be seen in the left-hand pane of Figure \ref{maxregret}.

\begin{figure}
    \center
    \caption{Regret-maximizing demand curves for $d=10$}
    \includegraphics[width=0.49\textwidth]{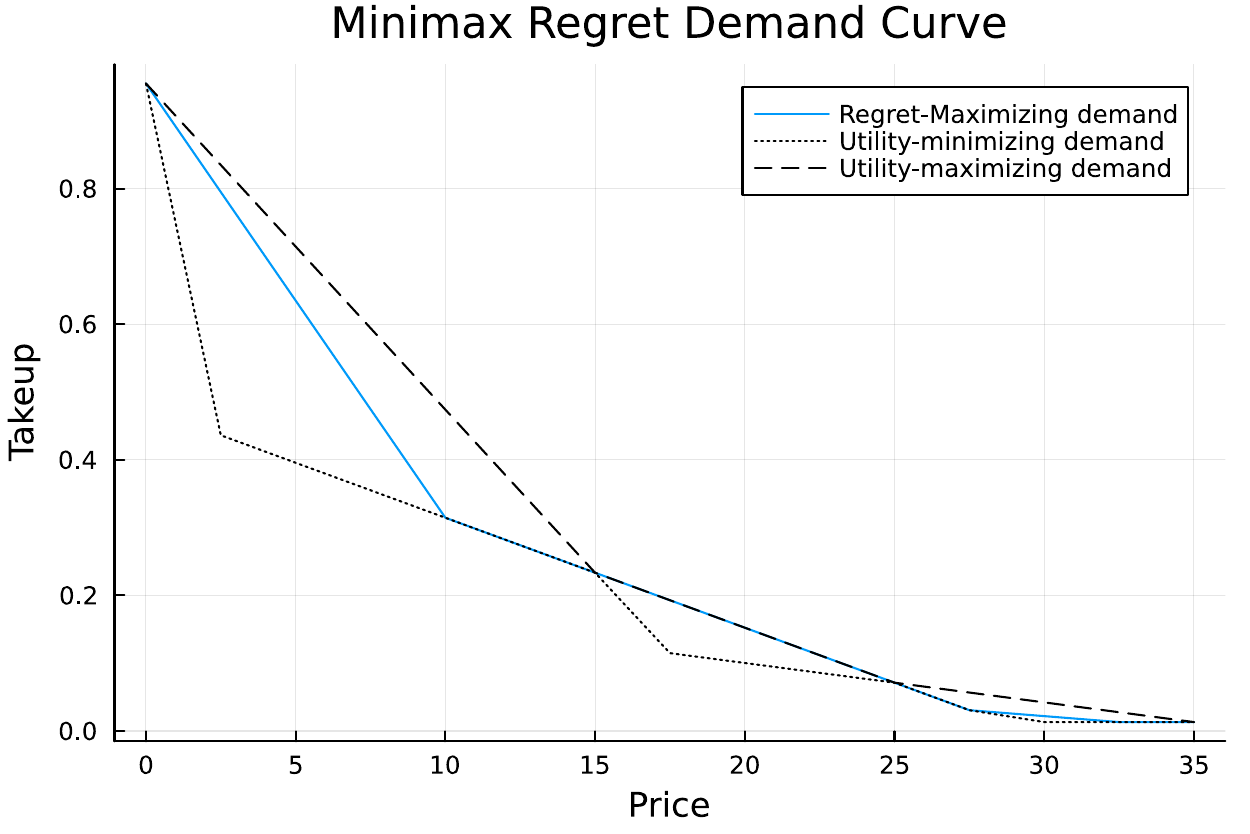}
    \includegraphics[width=0.49\textwidth]{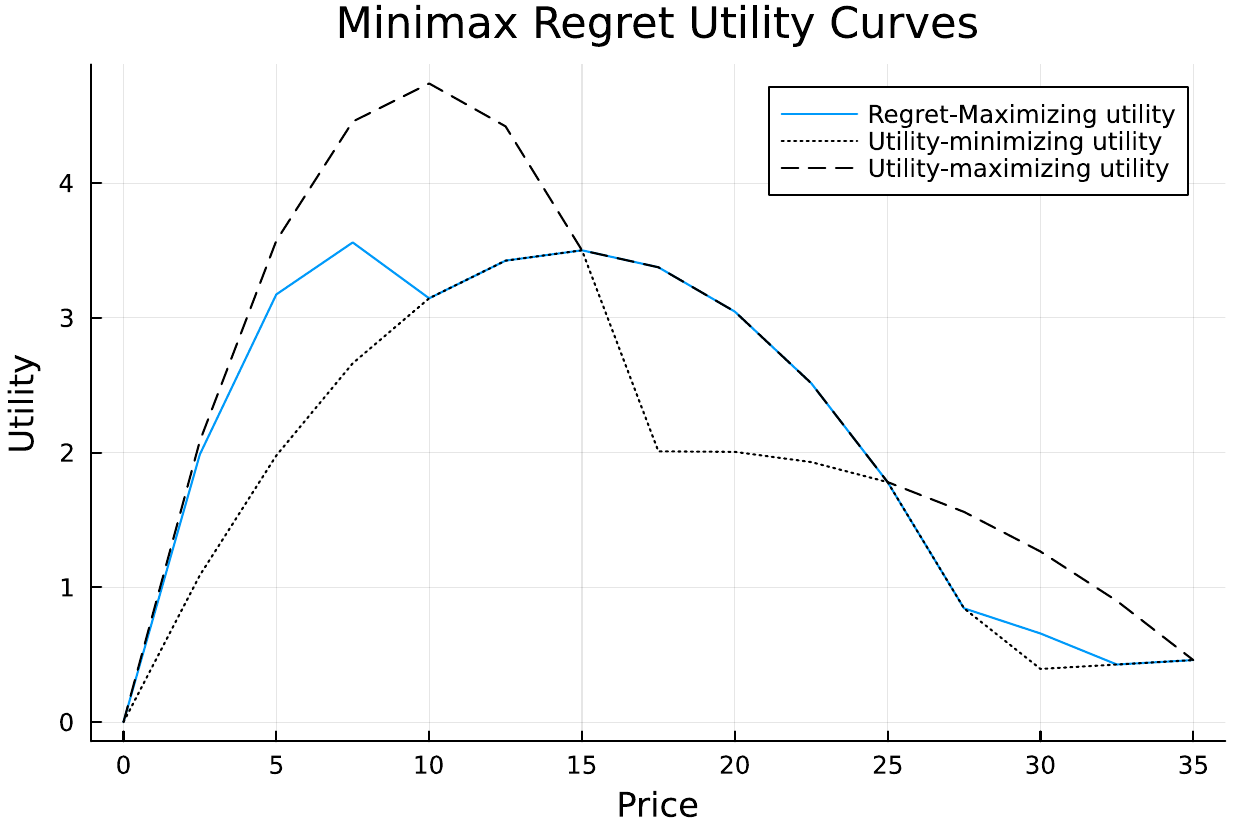}
    \label{optpolicy}
    \\
    \small \raggedright
    The demand curve that achieves maximum regret under the 
    minimax regret-optimal price of $d=10$ and the resulting welfare curve. 
    The curves lie between the pointwise bounds examined in Figure \ref{worstcaseexample}.
\end{figure}

Finally, I compute the optimal policy which does not target based on covariates.
The solution to the empirical minimax regret problem without covariates is given by 
the policy which assigns the price $d=10$ to the population.
This means that across all demand curves $m \in \hat{\mathcal{M}}$, $v_{m}(10)$ 
is uniformly as close as possible to $\max_d v_{m}(d)$.  
To visualize this, in Figure \ref{optpolicy} I overlay the regret-maximizing demand curve 
for the policy $d=10$ on top of the bounds on takeup and welfare plotted in Figure \ref{worstcaseexample}.  
The $m$ which maximizes regret is the one which maximizes utility at $d=7.5$ 
but performs somewhat worse when $d=10$.
This difference between the best possible outcome and the outcome realized under the chosen policy 
is the regret that nature seeks to maximize through the choice of $m$ and the decision maker 
seeks to minimize through the choice of $\pi$.  
Observe that maximum regret, given by $\max_{m} [v_{m}(7.5) -  v_{m}(10)]$, 
is much smaller than a naive comparison of the bounds.  
Hence, an adversarially chosen demand curve in $\hat{\mathcal{M}}$ can make a price of $d=10$ 
perform only mildly suboptimally. 

\subsection{Optimal policy with covariates}

Having illustrated the method for estimating $\hat m_0$, constructing $\hat{\mathcal{M}}$, 
and constructing $\hat \pi$ in the simple case of no covariates, 
I now solve for the optimal policy when the decision maker can target subsidies 
based on household size and income. 
I construct the estimate $\hat m_0(d,x)$ using a Lasso-penalized 
logistic regression of takeup on a dictionary of Chebyshev polynomials in 
household size and income, for each $d\in\mathcal{D}_0$.
As before, I use the shape restrictions that demand is decreasing and convex in $d$
for every $x$.
For some observations, the estimates $\hat m_0(d, X_i)$ violate
these shape restrictions.  
When this happens, I replace the estimates with 
$\arg\min_m \lVert m(d) -\hat m_0(d, X_i)  \rVert$, 
where the minimum is taken over all $m(d) \in \R^{J_0}$ 
that are decreasing and convex in $d$ 
and bounded between $0$ and $1$.  
This ensures that $\hat{\mathcal{M}}$ is nonempty.
These estimates are used to obtain $\hat \Gamma_{j}(X_i)$ for each $i$ and $j$.  

Finally, to estimate the optimal policy, I consider a policy class of 
linear eligibility score rules where each treatment shares the same eligibility score, 
but different cutoffs.
The decision maker chooses a vector of covariate weights $\beta$ and a vector of 
increasing cutoffs $\{c_j\}_{j=0}^{J-1}$.  
A household with covariates $X_i$ receives treatment $d_j$ if 
$c_{j-1} < X_i' \beta \leq c_j$, 
where $c_0 = -\infty$ and $c_J = \infty$. 
I impose that the eligibility score increases with income, implying that poorer 
households receive lower prices.  Formally,
\begin{align} 
    \label{policyclass}
    \Pi = \bigg\{
        \pi : 
        \pi(x) = \sum_{j=1}^{J-1} (d_{j+1} - d_{j}) \1[ X_i ' \beta > c_j], \; 
        c_{j-1} \leq c_j, \; \beta_1 > 0 
    \bigg\}
\end{align}
Appendix \ref{appendix-computation} discusses how this class can be formulated with 
linear and integer constraints, resulting in a mixed integer-linear program formulation 
for the empirical minimax regret problem (\ref{empiricalMMR}).  

\begin{figure}[ht]
    \center 
    \caption{Optimal policy}
    \includegraphics[width=0.8\textwidth]{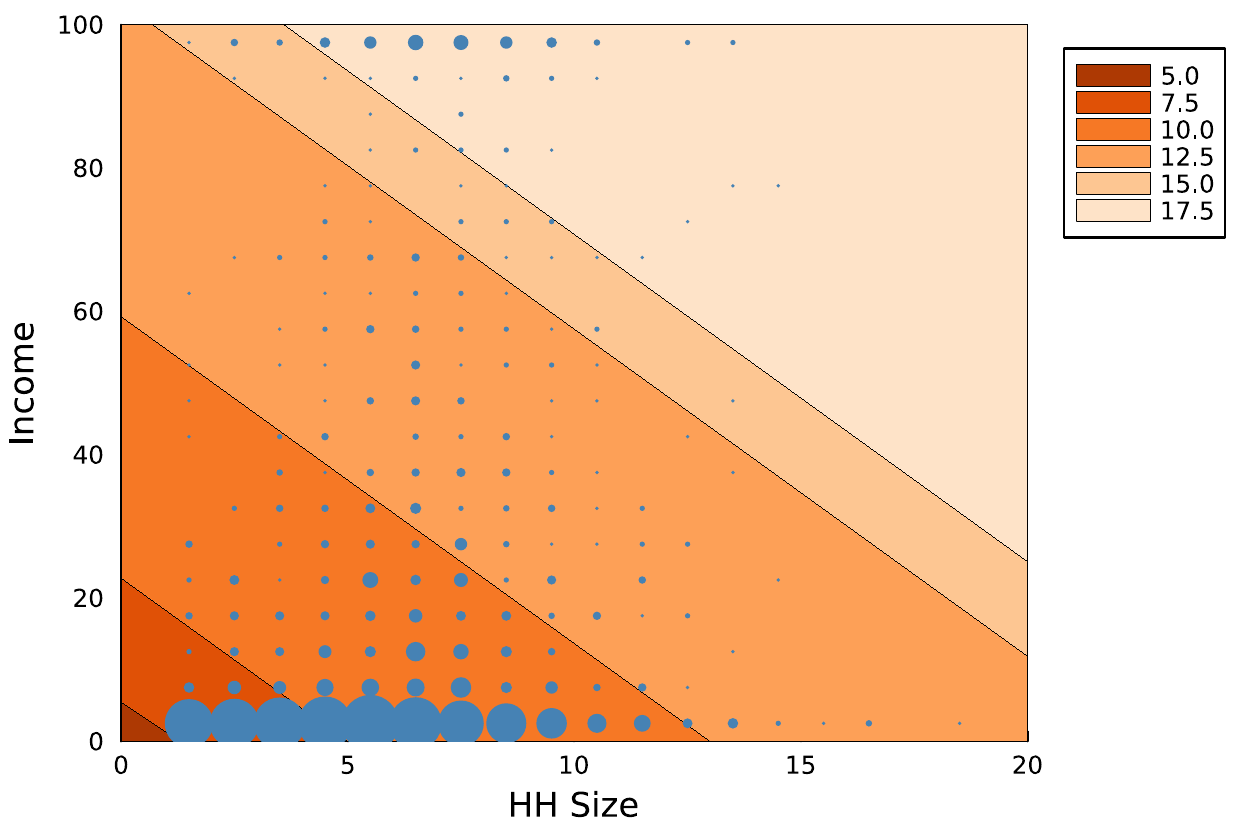}
    \label{results}
    \\
    \small \raggedright
    The estimated optimal treatment allocation as a function of household size and earnings.  
    The size of the dots is proportional to the number of people at each value of covariates.  
    The shaded regions indicate which covariate values are assigned to each treatment.
\end{figure}

\begin{table}[ht]
    \caption{Optimal policy}
    \label{resultstable}
    \centering
    \begin{tabular}{c c c c c c c }
$d$ & 5.0 & 7.5 & 10.0 & 12.5 & 15.0 & 17.5 \\ \hline
\\
\% Treated & 6.6\% & 27.4\% & 52.2\% & 8.3\% & 1.3\% & 4.2\% \\
\\
Cutoff & 1.48 & 6.1 & 15.81 & 27.57 & 31.09 & 42.23 \\
($\beta = [0.266, 1.221]$)
\\ \hline
\\
\end{tabular}

    \\ \small \raggedright 
    Percent of population assigned to each treatment and eligibility score cutoff 
    for each treatment for which a nonzero share of the population was assigned. 
    Households were assigned to treatment $j$ if their score was below cutoff $j$ 
    and above cutoff $j-1$. 
\end{table}

The optimal allocation is illustrated in Figure \ref{results}, and exact estimates of the 
optimal policy along with the fraction of the population assigned to each treatment are 
presented in Table \ref{resultstable}.
The optimal allocation assigns poor, small households the lowest prices as they have the 
lowest willingness to pay.  
Almost the entire population is assigned a price not observed in the experimental data, 
with only $1.3\%$ of the population being assigned to the price $d=15$ which was optimal 
amongst the prices that were used in the experiment.
Most of the population is assigned to a price of $d=7.5$ or $d=10$.

\subsection{Comparison with other policies}

I compare the optimal policy in Figure \ref{results} with two other policies that 
a decision maker might use in the absence of the method proposed in this paper.
I evaluate the performance of these policies in terms of their estimated maximum regret,
and compare this to the estimated maximum regret of the policy proposed in this paper.
These heuristic policies, which are not designed to control maximum regret,
have the potential to perform substantially worse than the minimax regret optimal policy.

For the first benchmark, I estimate the optimal policy using an ad-hoc parametric interpolation
that a decision maker might use to forecast the effects of new treatments.
I estimate takeup using OLS with linear and quadratic terms in household
size and income, motivated by the near-quadratic response to price observed in 
Figure \ref{experimentalresults}.
This means that the identified set $\hat{\mathcal{M}}$ is a singleton,
making the maximization over $\hat{\mathcal{M}}$ trivial.
These takeup estimates are then used to construct a policy which maximizes
estimated utility.
The resulting treatment allocation is shown in Figure \ref{fig:parametric},
and the associated policy is given in Table \ref{table:parametric}.

\begin{figure}[ht]
    \center
    \caption{Optimal policy with parametric interpolation}
    \includegraphics[width=0.8\textwidth]{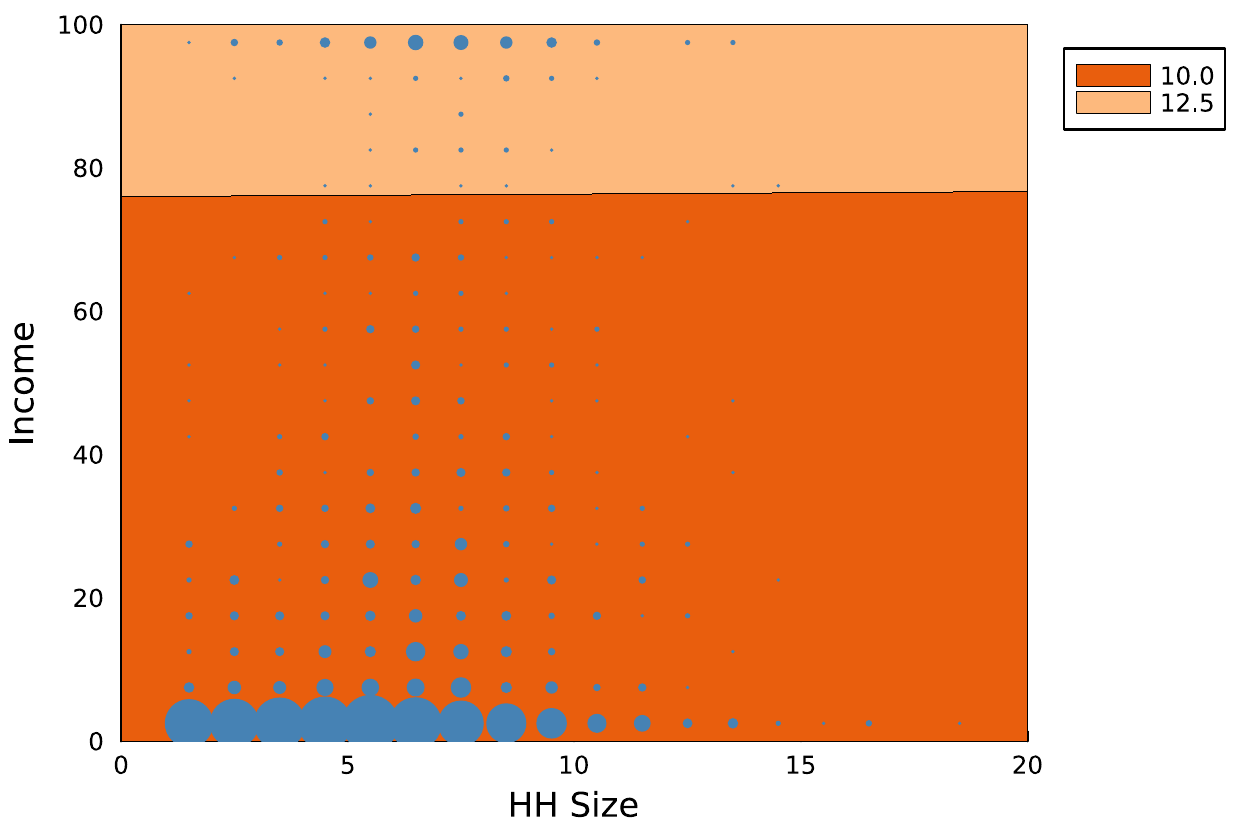}
    \label{fig:parametric}
    \\
    \small \raggedright
    The estimated optimal treatment allocation under the parametric interpolation 
    as a function of household size and earnings.  
    The size of the dots is proportional to the number of people at each value of covariates.  
    The shaded regions indicate which covariate values are assigned to each treatment.
\end{figure}

\begin{table}[ht]
    \caption{Optimal policy with parametric interpolation}
    \label{table:parametric}
    \centering
    \begin{tabular}{c c c }
$d$ & 10.0 & 12.5 \\ \hline
\\
\% Treated & 95.1\% & 4.9\% \\
\\
Cutoff & 2.99 & 3.89 \\
($\beta = [0.039, -0.001]$)
\\ \hline
\\
\end{tabular}

    \\ \small \raggedright 
    Percent of population assigned to each treatment and eligibility score cutoff 
    under the parametric interpolation,
    for each treatment for which a nonzero share of the population was assigned. 
    Households were assigned to treatment $j$ if their score was below cutoff $j$ 
    and above cutoff $j-1$.
\end{table}

For the second benchmark, I estimate the optimal policy under the restriction that the policy 
cannot assign new treatments.
This reflects what a decision maker might do if they did not have a method
for evaluating and choosing policies which involve new treatments.
That is, I take the regret estimates $\hat{\Gamma}_{j}$ from the Lasso model
and use it to construct a minimax regret policy
under the additional restriction that the policy cannot assign new treatments.
This reflects the maximum regret that a policymaker would incur by choosing
not to assign new treatments, even if it were possible to do so.
The resulting treatment allocation is shown in Figure \ref{fig:restricted},
and the associated policy is given in Table \ref{table:restricted}.

\begin{figure}[ht]
    \center
    \caption{Optimal policy with no new treatments}
    \includegraphics[width=0.8\textwidth]{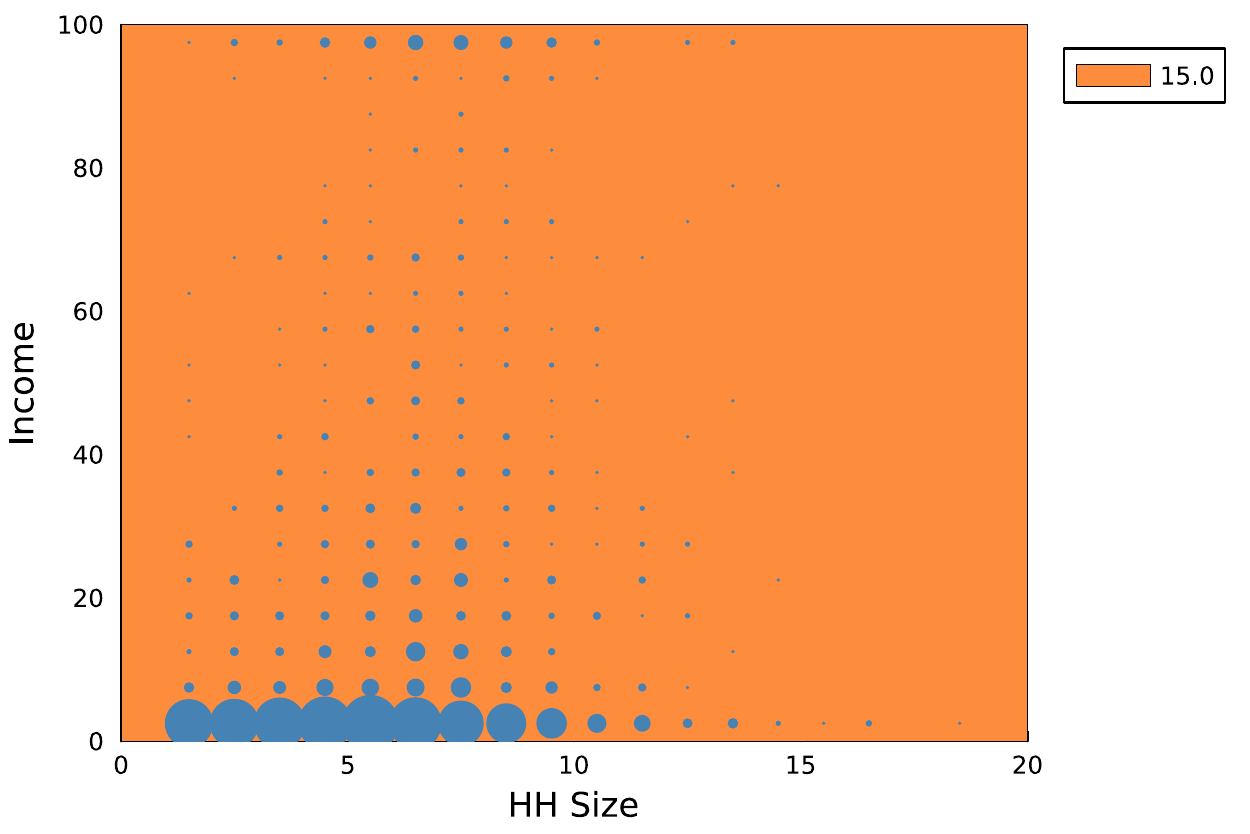}
    \label{fig:restricted}
    \\
    \small \raggedright
    The estimated optimal treatment allocation under the restriction that the policy 
    cannot assign new treatments,
    as a function of household size and earnings.  
    The size of the dots is proportional to the number of people at each value of covariates.  
    The shaded regions indicate which covariate values are assigned to each treatment.
\end{figure}

\begin{table}[ht]
    \caption{Optimal policy with no new treatments}
    \label{table:restricted}
    \centering
    \begin{tabular}{c c }
$d$ & 15.0 \\ \hline
\\
\% Treated & 100.0\% \\
\\
Cutoff & -5.26 \\
($\beta = [0.0, -5.263]$)
\\ \hline
\\
\end{tabular}

    \\ \small \raggedright 
    Percent of population assigned to each treatment and eligibility score cutoff 
    under the restriction that the policy cannot assign new treatments,
    for each treatment for which a nonzero share of the population was assigned. 
    Households were assigned to treatment $j$ if their score was below cutoff $j$ 
    and above cutoff $j-1$.
\end{table}

In Table \ref{regrettable}, I compare the estimated maximum regret of the 
estimated policy $\hat{\pi}$ with the estimated maximum regret of the other 
two policies described above.
By construction, the estimated policy $\hat{\pi}$ minimizes the estimated 
maximum regret.

The parametric extrapolation results in higher estimated regret than the
method proposed in this paper, which takes into account non-identification
of the effects of new treatments.
Thus, whether the decision maker should use the parametric extrapolation
is sensitive to how much they trust the parametric form.
If the parametric model is used only for convenience and the actual identified
set is described by $\mathcal{M}$, then the parametric model may lead to
higher regret than the robust method proposed in this paper.
When the decision maker uses the parametric extrapolation,
they fail to consider the worst-case effects of new treatments,
and hence are overconfident in the benefits of new treatments.

Restricting to the support of the experimental data greatly increases maximum 
regret.
When the decision maker chooses not to assign new treatments,
the worst-case utility function $v_m$ will be high on the set of new treatments 
to make the regret of the chosen policy large
(in the case without covariates, the worst-case utility curve coincides with the
upper bound on utility at $d=10$ in Figure \ref{worstcaseexample}).
By implementing new treatments, the decision maker can ensure they don't miss out
on potentially large gains from these new treatments.
However, it is not guaranteed that the decision maker will choose to implement
new treatments, as they must also ensure the potential downside of implementing
new treatments is not so large as to make worst-case regret higher than
restricting to old treatments.

\begin{table}
    \caption{Comparison of maximum regret of different policies}
    \label{regrettable}
    \centering
    \begin{tabular}{c c c}
Policy & Maximum Regret & Percent Increase \\
\hline
\\
Lasso & 488.19 &  \\
\\
Parametric & 620.75 & 27.15 \% \\
\\
Restricted & 1381.38 & 182.96 \% \\
\hline \\ 
\end{tabular}

    \\ \small \raggedright 
    Estimated maximum regret of the optimal policy, the policy with parametric interpolation, 
    and the policy with no new treatments.
\end{table}

\section{Conclusion}
\label{section-conclusion}
Experiments may not pilot all possible treatments a decision maker may consider.  
The existing literature on policy learning and treatment choice does not offer 
much guidance for how to use data on some treatment values to design policies 
involving new treatment values.  
I use data on previously observed treatments, 
partial identification, and the minimax regret criterion to extend 
empirical welfare maximization methods to settings where new treatments 
may be considered.
Since the effects of new treatments are partially identified, a single policy 
is chosen to uniformly minimize regret across the identified set.  
The empirical minimax regret estimator is computationally tractable and possesses 
favorable regret convergence properties.
In the setting of targeting subsidies to connect to the electrical grid, 
the estimator takes information on a small set of treatments and provides 
informative bounds on the effects of a much richer set new treatments, 
resulting in policies that implement new treatments which are uniformly 
close to optimal in every state of the world.

\pagebreak
\printbibliography
\pagebreak

\appendix

\section{Proof of Theorem \ref{regretconvergence}}
\label{appendix-proofs}
\subsection{Intermediate results}

I first introduce some notation and state existing results that I will use.  
Given a class of functions $\mathcal{F}$, the Rademacher complexity of $\mathcal{F}$ is defined as 
$$
    \mathcal{R}_N(\mathcal{F}) := \E\bigg[ \sup_{f\in\mathcal{F}} \bigg| \frac{1}{N} \sum_i \epsilon_i f(X_i) \bigg| \bigg] 
$$
where $\epsilon_i$ are i.i.d. Rademacher random variables.  

We say $\mathcal{T}$ is a $\delta$-cover of a metric space $(\mathcal{F} , h)$ 
if for every $f \in \mathcal{F}$, there is some $f_i \in \mathcal{T}$ such that $h(f_i, f)\leq \delta$.  
The cardinality of the smallest $\delta$-cover of $\mathcal{F}$ is called the 
\textit{$\delta$-covering number} of $\mathcal{F}$ and is denoted $N(\delta, \mathcal{F}, h)$.

I make use of the following existing results:
\begin{lemma}{(\textcite{kitagawaWhoShouldBe2018} Lemma A.1)}
    \label{subgraph}
    Let $\mathcal{G}$ be a VC-class of subsets of $\mathcal{X}$ with VC dimension $V<\infty$.  
    Let $g$ and $h$ be two given functions from $\mathcal{Y} \times \mathcal{D} \times \mathcal{X}$ to $\R$. Then 
    \begin{align*} 
        \mathcal{F} = \{f : f(y, d, x) = g(y,d,x) \1\{x \in G\} + h(y,d,x) \1\{x \not\in G\}, G \in \mathcal{G} \}
    \end{align*}
    is a VC subgraph class of functions with VC dimension less than or equal to $V$.
\end{lemma}

\begin{lemma}{(Symmetrization) (\textcite{vandervaartWeakConvergence1996} Lemma 2.3.1)}
    \label{Symmetrization}
    For a class of measurable functions $\mathcal{F}$ and i.i.d random variables $X_1, \dots, X_N$,
    $$
    \E\bigg[\sup_{f\in\mathcal{F}} \bigg| \frac{1}{N}\sum_{i=1}^N f(X_i) - \E[f(X_i)] 
    \bigg| \bigg]\leq 2 \mathcal{R}_N(\mathcal{F})
    $$
\end{lemma} 

\begin{lemma}{(Dudley's entropy integral inequality) (\textcite{vandervaartWeakConvergence1996} Corollary 2.2.8)}
    \label{Dudley}
    Let $(Z_f)_{f\in\mathcal{F}}$ be a separable process with sub-Gaussian increments. 
    Then for some constant $K$, we have for any $f_0$
    $$
        \E[\sup_{f\in\mathcal{F}} |Z_f|] \leq \E[|Z_{f_0}|] + K \int_0^{\infty} \sqrt{\log N(t, \mathcal{F}, h)} dt
    $$
    for some constant $K$.
\end{lemma}

\begin{lemma}{(\textcite{vandervaartWeakConvergence1996} Theorem 2.6.7)}
    \label{VC}
    Suppose $\mathcal{F}$ is a VC subgraph class with VC dimension at most $V < \infty$ 
    and suppose $\mathcal{F}$ has a measurable envelope function $F$.  
    For $q\geq 1$ let $P$ be a probability measure such that $\lVert F \rVert_{q,P} >0$. Then 
    $$
    N(\delta \lVert F \rVert_{L_q(P)}, \mathcal{F}, L_q(P)) \leq KV(16e)^V(1/\delta)^{q(V-1)}
    $$
    for some constant $K$ and $0<\delta<1$.
\end{lemma}

I now state and prove a useful result which establishes the relationship between
the complexity of the policy class $\Pi$ and the complexity of the class of 
regret estimates.  
Define the following function classes:
\begin{align*}
    & \mathcal{F} := \{ f :  f(x) = \sum_{j=1}^J \pi_j(x) \Gamma_j(x), \pi \in \Pi\}
    \\
    & \Pi_j := \{\pi_j : \pi_j(x) = \1[\pi(x) = d_j], \pi\in\Pi\}, \quad \forall j \in \{1,\dots, J \}
\end{align*}
The assumption that $\mathcal{F}$ is separable will be maintained.

\begin{lemma}\label{covering lemma}
    Under Assumption $\ref{DGP}.2$,
    \begin{align*}
        N(\delta, \mathcal{F}, L_2(P_N)) &\leq \prod_{j=1}^J N(\epsilon, \Pi_j, L_2(P_N))
    \end{align*}
    where $\epsilon = \delta/(B\sqrt{J})$ and $B$ is the bound on $v(d,x)$ implied by Assumption \ref{DGP}.2.
\end{lemma}

\begin{proof}
    Let $f \in \mathcal{F}$ be given. From the discussion in Section \ref{section-decision}, $f$ can be written as 
    \begin{align}
        f(x) &= \max_k \max_{m \in \mathcal{M}_P} \bigg( v_{m}(d_k,x) - \sum_{j=1}^J \pi_j(x)  v_{m}(d_j,x) \bigg)
        \label{f}
    \end{align} 
    for some $\pi \in \Pi$. 
    For each treatment $j$, let $\mathcal{T}_j$ be an $\epsilon$-cover of $\Pi_j$ 
    and let $\tilde \pi_j$ be an element of $\mathcal{T}_j$ satisfying 
    $\E_N[(\pi_j(X) - \tilde \pi_j(X))^2]^{1/2} \leq \epsilon$.  
    Define the approximating function $\tilde f$ by 
    \begin{align}
        \tilde f(x) = \max_k \max_{m \in \mathcal{M}_P}  v_{m}(d_k,x) - \sum_{j=1}^J \tilde \pi_j(x) v_{m}(d_j,x) 
        \label{tilde_f}
    \end{align}  
    Finally, let $k^*, m^*(x)$ be maximizers of (\ref{f}) and let $\tilde k, \tilde m(x)$ be maximizers of (\ref{tilde_f}).

    For each $x$ we have by the optimality of $k^*$ and $m^*$
    \begin{align*} 
        f(x) - \tilde f(x)
        &= v_{m^*}(d_{k^*},x) - v_{\tilde m}(d_{\tilde k},x) 
        + \sum_{j=1}^J \bigg[ \tilde \pi_j(x) v_{m^*}(d_j,x) - \pi_j(x)  v_{\tilde m}(d_j,x) \bigg]
        \\
        &\geq  v_{\tilde m}(d_{\tilde k},x) - v_{\tilde m}(d_{\tilde k},x) 
        + \sum_{j=1}^J \bigg[ \tilde \pi_j(x) v_{\tilde m}(d_j,x) - \pi_j(x)  v_{\tilde m}(d_j,x) \bigg]
        \\
        &= \sum_{j=1}^J \bigg[ (\tilde \pi_j(x) - \pi_j(x)) v_{\tilde m}(d_j, x)\bigg]
    \end{align*}
    Likewise, optimality of $\tilde m $ and $\tilde k$ imply 
    \begin{align*} 
        f(x) - \tilde f(x)
        &\leq \sum_{j=1}^J \bigg[ (\tilde \pi_j(x) - \pi_j(x)) v_{m^*}(d_j, x)\bigg]
    \end{align*}
    Together, we have 
    \begin{align*}
        |f(x) - \tilde f(x)| 
        &\leq \max\bigg\{\bigg| \sum_{j=1}^J (\tilde \pi_j(x) - \pi_j(x)) v_{\tilde m}(d_j,x)  \bigg|, \bigg|  \sum_{j=1}^J (\tilde \pi_j(x) - \pi_j(x)) v_{m^*}(d_j,x)\bigg| \bigg\}
        \\
        &\leq \max \bigg\{ \lVert (\tilde \pi_j(x) - \pi_j(x))_{j=1}^j \rVert \lVert v_{\tilde m}(\cdot,x) \rVert , \lVert(\tilde \pi_j(x) - \pi_j(x))_{j=1}^j \rVert \lVert v_{m^*}(\cdot, x) \rVert \bigg\} 
        \\
        &\leq B \bigg( \sum_j (\pi_j(x) - \tilde \pi_j(x))^2\bigg)^{1/2}
    \end{align*}
    where the second line is by the Cauchy-Schwartz inequality. 
    Squaring and integrating over $x$,
    \begin{align*}
        \E_{N,P}[(f(X) - \tilde f(X))^2]
        &\leq B^2\E_{N,P}\bigg[\sum_j (\pi_j(X) - \tilde \pi_j(X))^2\bigg]
        \\ &\leq B^2 J \epsilon^2
        \\ &\leq \delta^2
    \end{align*}
    Taking the square root of both sides shows that $\lVert f - \tilde f \rVert_{L_2(P_N)} \leq \delta$.
    Consider the set of all such functions constructed this way,
    $$
        \mathcal{T} := \bigg\{\tilde f : \tilde f(x) = \max_k \max_{m \in \mathcal{M}}  v_{m}(d_k,x) - \sum_{j=1}^J \tilde \pi_j(x) v_{m}(d_j,x) , \; 
        \tilde \pi = (\tilde \pi_1, \dots, \tilde \pi_J), \tilde \pi_j\in \mathcal{T}_j \bigg\}
    $$
    we see that $| \mathcal{T}| = \prod_j |\mathcal{T}_j|$, and $\mathcal{T}$ is a $\delta$-cover of $\mathcal{F}$.
\end{proof}

We can now prove the main results of the paper.

\subsection{Proof of Lemma \ref{eplemma}}

\begin{proof}
    To simplify notation, for now we consider $P$ fixed and suppress dependence on $P$.
    By the definition of $\mathcal{F}$, we have 
    $$
    \E\bigg[\sup_{\pi\in\Pi} \bigg| \overline{R}(\pi) 
        - \tilde{R}_N(\pi) \bigg| \bigg] 
        = \E\bigg[\sup_{f\in\mathcal{F}} \bigg| \frac{1}{N}\sum_{i=1}^N f(X_i) - \E[f(X_i)]\bigg| \bigg]
    $$
    Hence, we can apply Lemma \ref{Symmetrization} to obtain
    \begin{align*} 
        \E\bigg[\sup_{\pi\in\Pi} \bigg| \overline{R}(\pi) 
        - \tilde{R}_N(\pi) \bigg| \bigg]
        &\leq
        2 \mathcal{R}_N(\mathcal{F})
    \end{align*}
    Now, define $Z_f = \frac{1}{\sqrt{N}}\sum_{i=1}^N\epsilon_i f(X_i)$. 
    The increments of $(Z_f)_{f\in\mathcal{F}}$ are given by 
    $\frac{1}{\sqrt{N}}\sum_{i=1}^N\epsilon_i (f(X_i) - g(X_i))$.  
    Conditional on $(X_i)_{i=1}^N$, we can apply Hoeffding's inequality 
    (e.g. \textcite{vandervaartWeakConvergence1996} Lemma 2.2.7) to establish that this is 
    sub-Gaussian with parameter 
    $\left(\frac{1}{N} \sum_{i=1}^N (f(X_i) - g(X_i))^2 \right)^{1/2} = \lVert f - g \rVert_{L_2(P_N)}$.
    We can then apply Lemma \ref{Dudley} conditional on $(X_i)_{i=1}^N$ to obtain
    \begin{align*} 
        \frac{1}{\sqrt{N}} 
        \E_\epsilon[\sup_{f\in\mathcal{F}}|Z_f|] 
        & \leq \frac{1}{\sqrt{N}}\E_\epsilon[|Z_{f_0}|] +  \frac{K}{\sqrt{N}} \int_0^{\infty} \sqrt{\log N(t, \mathcal{F}, L_2(P_N))} dt
    \end{align*}
    for some $f_0\in\mathcal{F}$.  
    Moreover, since $f_0$ is bounded by $2B$, $|Z_{f_0}| \leq |\sum_{i=1}^N \epsilon_i \frac{2B}{\sqrt{N}}|$. 
    Again by Hoeffding's inequality, this implies that $Z_{f_0}$ is sub-Gaussian with parameter 
    $K'\left(\frac{1}{N}\sum_{i=1}^N (2B)^2\right)^{1/2} = K'2B$ which does not depend on $N$. 
    Basic properties of sub-Gaussian random variables 
    (e.g. \textcite{vershyninHighdimensionalProbabilityIntroduction2018} Proposition 2.5.2) 
    imply that $\E_\epsilon[|Z_{f_0}|] \leq K''$ for some constant $K''$. Thus,
    \begin{align*}
        \frac{1}{\sqrt{N}} 
        \E_\epsilon[\sup_{f\in\mathcal{F}}|Z_f|]&\leq 
        \frac{1}{\sqrt{N}}\left( K'' + K \int_0^\infty \sqrt{\log N(t, \mathcal{F}, L_2(P_N))} dt \right)
        \\
        &=
        \frac{1}{\sqrt{N}}\left( K'' + K \int_0^{4B} \sqrt{\log N(t, \mathcal{F}, L_2(P_N))} dt \right)
    \end{align*}
    where we have used the fact that since the diameter of $\mathcal{F}$ is $2B$, 
    the integrand is $0$ for $t>4B$. 
    By Lemma \ref{subgraph}, each class $\Pi_j$ is a VC subgraph class of functions 
    with VC dimension at most $v$.  
    Then applying Lemma \ref{covering lemma} and a change of variables yields
    \begin{align*} 
        \frac{1}{\sqrt{N}} 
        \E_\epsilon[\sup_{f\in\mathcal{F}}|Z_f|]&\leq \frac{1}{\sqrt{N}}\left( K'' + K \int_0^{4/\sqrt{J}} \sqrt{\sum_j \log N(t, \Pi_j, L_2(P_N))} B \sqrt{J} dt \right)
    \end{align*}
    Apply Lemma \ref{VC} to each VC subgraph class of functions $\Pi_j$, 
    which have envelope $1$, and allow $K$ to subsume other constants to obtain
    \begin{align}
        \frac{1}{\sqrt{N}} 
        \E_\epsilon[\sup_{f\in\mathcal{F}}|Z_f|]&\leq 
        \frac{1}{\sqrt{N}}\left( K'' + K \int_0^{4/\sqrt{J}} \sqrt{ \log K''' V(16e)^V(1/t)^{2(V-1)}} dt \right)
        \\
        &\leq
        \frac{1}{\sqrt{N}}\left( K'' + K\sqrt{V}\right)
        \\ &\leq K\sqrt{\frac{V}{N}}
    \end{align}
    Finally, since $\mathcal{R}_N(\mathcal{F}) = \E[\frac{1}{\sqrt{N}} 
    \E_\epsilon[\sup_{f\in\mathcal{F}}|Z_f|]]$, we obtain
    \begin{align*} 
        \E\bigg[\sup_{\pi\in\Pi} \bigg| \overline{R}(\pi) 
        - \tilde{R}_N(\pi) \bigg| \bigg]
        &\leq K\sqrt{\frac{V}{N}}
    \end{align*}
    Note that the constant $K$ does depend on $B$ and $J$.
\end{proof}

\subsection{Proof of Lemma \ref{lpconvergence}}

\begin{proof}
    To simplify notation, for now we consider $P$ fixed and suppress dependence on $P$.
    For every $1 \leq j,k \leq J$, let $\gamma_{jk} : \R^{J_0} \mapsto 2^{\R^J}$ be the 
    identified set for covariate-level regret, viewed as a set-valued mapping from the 
    first stage conditional mean response vector to subsets of $\R^J$. 
    That is, hold $x$ fixed and define 
    $\gamma_{jk}(w) = \{ b_{jk}(x)' m - c_{jk}(x) : Fm = w, Sm \leq r\}$.  
    For this proof, we view $\Gamma_{jk}$ as a function of the first stage conditional 
    mean response $m_0$ to consider how $\Gamma_{jk}$ changes with perturbations to $m_0$.  
    Thus, $\Gamma_{jk}(m_0(\cdot,X_i)) = \max\{\gamma_{jk}(m_0(\cdot,X_i))\}$.

    For any matrix $A$, let $A^\dag : \text{null}(A)^\perp \mapsto \text{range}(A)$ 
    denote the Moore-Penrose pseudoinverse operator, where $\text{null}(A)$ and 
    $\text{range}(A)$ denote the null space and range of $A$, respectively.  
    For any $w \in \R^{J_0}$,
    \begin{align*} 
        \{m : Fm = w\} = \{m : m = F^\dag w + y, \; y \in \text{null}(F)\}
    \end{align*}
    Let $\tilde w \in \R^{J_0}$ be given.  
    For any $J$, let $\mathcal{B}_J$ be the unit ball in $\R^J$, 
    $\mathcal{B}_J := \{ w \in \R^J : \lVert w \rVert \leq 1 \}$. 
    Since $\tilde w \in w + \lVert w - \tilde w \rVert \mathcal{B}_{J_0}$,
    \begin{align*} 
        \{m : m = F^\dag \tilde w + y, \; y \in \text{null}(F)\}
        &\subseteq
        \{m : m = F^\dag z + y, \; y \in \text{null}(F), \; z \in w + \lVert w - \tilde w \rVert \mathcal{B}_{J_0} \}
        .
    \end{align*}
    Let $\lVert A \rVert$ denote the operator norm of a matrix $A$.  
    Since
    $F^\dag \lVert w - \tilde w \rVert \mathcal{B}_{J_0}
    \subseteq \lVert F^\dag \rVert \lVert w - \tilde w \rVert \mathcal{B}_{J}$,
    we have
    \begin{align*}
        &\{m : m = F^\dag z + y, \; y \in \text{null}(F), \; z \in w + \lVert w - \tilde w \rVert \mathcal{B}_{J_0} \}
        \\
        &\subseteq \{m : m = F^\dag w + y, \; y \in \text{null}(F)\} + \lVert F^\dag \rVert \lVert \tilde w - w \rVert  \mathcal{B}_J 
        .
    \end{align*}
    This implies
    \begin{align*} 
        \{m : Fm = \tilde w\} \cap \{m : Sm \leq r\} 
        \subseteq \left(\{m : m = F^\dag w + y, y \in \text{null}(F)\} + \lVert F^\dag \rVert \lVert \tilde w - w \rVert  \mathcal{B}_J\right) \cap \{m : Sm \leq r\}
    \end{align*}
    and since $\gamma_{jk}(\tilde w)$ consists of scalars of the form $b_{jk}(x)' m - c_{jk}(x)$ for $m$ in this set, we have
    \begin{align*}
        \gamma_{jk}(\tilde w) 
        &\subseteq \gamma_{jk}(w) + \lVert b_{jk}(x) \rVert  \lVert F^\dag \rVert \lVert \tilde w - w \rVert  \mathcal{B}_1
    \end{align*}
    and likewise 
    \begin{align*} 
        \gamma_{jk}(w) \subseteq \gamma_{jk}(\tilde w) + \lVert b_{jk}(x) \rVert  \lVert F^\dag \rVert \lVert \tilde w - w \rVert  \mathcal{B}_1
    \end{align*}
    Therefore, the correspondence $\gamma_{jk}$ is Lipschitz with respect to the Hausdorff distance, 
    with Lipschitz constant $\sup_x \lVert b_{jk}(x) \rVert \lVert F^\dag \rVert := \kappa < \infty$ 
    (\cite{rockafellarVariationalAnalysis2009}).
    Importantly, this implies 
    $| \Gamma_{jk}(\hat m_0(\cdot,X_i) - \Gamma_{jk}(m_0(\cdot, X_i)) | \leq \kappa \lVert \hat m_0(\cdot, X_i) - m_0(\cdot,X_i) \rVert$.  
    While not essential to our analysis, we note that in our case $\lVert F^\dag \rVert = 1$.

    We can now prove the main claim of the lemma.
    \begin{align*} 
        &\E\bigg[\sup_{\pi\in\Pi} \bigg|\tilde{R}_N(\pi) 
        -  \overline{R}_N(\pi) \bigg| \bigg]
        \\
        &=\E\bigg[\sup_{\pi\in\Pi} \bigg| \frac{1}{N}\sum_{i=1}^N \sum_{j=1}^J \pi_{ij} (\max_k \Gamma_{jk}(\hat m_0(\cdot,X_i)) - \max_k \Gamma_{jk}(m_0(\cdot,X_i))) \bigg| \bigg]
        \\
        &\leq \E\bigg[\frac{1}{N}\sum_{i=1}^N \max_{j,k} \bigg| \Gamma_{jk}(\hat m_0(\cdot,X_i)) - \Gamma_{jk}(m_0(\cdot,X_i)) \bigg| \bigg]
        \\
        &\leq \E\bigg[\frac{1}{N}\sum_{i=1}^N  \kappa \lVert \hat m_0(\cdot, X_i) - m_0(\cdot,X_i) \rVert \bigg]
    \end{align*}
    Finally, we bound this uniformly in $P$ by Assumption \ref{rootnconvergence}. We conclude that 
    \begin{align*}
        \sup_{P\in\mathcal{P}_C}\E_P\bigg[\sup_{\pi\in\Pi} \bigg|\tilde{R}_{N,P}(\pi) 
        -  \overline{R}_N(\pi) \bigg| \bigg]
        &\leq \mathcal{O}(\rho_N^{-1})
    \end{align*}
\end{proof}

\subsection{Proof of Theorem \ref{regretconvergence}}

\begin{proof}
    Combining the bounds in Lemmas \ref{eplemma} and \ref{lpconvergence}
    establishes Theorem \ref{regretconvergence}.
\end{proof}

\section{Extensions of theoretical results}
\label{appendix-extensions}
\subsection{Finite-sample bounds}

In this section we consider strengthening the asymptotic result of Theorem 
\ref{regretconvergence} to hold for finite samples.
We replace Assumption \ref{rootnconvergence} with the following assumption
that the estimate $\hat m_0$ has finite-sample error bounds.

\begin{assumption}
    \label{rootnbound}
    There exists a constant $K > 0$ and a sequence $\rho_N \to \infty$ such that
    for all $N \geq 1$,
    the estimate $\hat m_0$ satisfies 
    \begin{enumerate}
        \item $\sup_{P \in \mathcal{P}}
            \E_P\bigg[ 
                \frac{1}{N} \sum_{i=1}^N 
                \lVert \hat m_0(\cdot,X_i) - m_{0,P}(\cdot,X_i) \rVert 
            \bigg] 
            \leq K \rho_N^{-1}$
        \item $\hat{\mathcal{M}} = 
            \{ m : S m(\cdot,X) \leq r, \; F m(\cdot,X) = \hat m_0(\cdot,X)\}$ 
            is nonempty, almost surely,
            for all $P \in \mathcal{P}$.
    \end{enumerate}
\end{assumption}

This assumption leads immediately to an extension of 
Lemma \ref{lpconvergence}.
\begin{corollary}
    \label{lpconvergencefinite}
    Under assumptions \ref{DGP}, \ref{linear}, and \ref{rootnbound},
    \begin{align*}
        \sup_{P \in \mathcal{P}_C} \E_P\bigg[ 
            \sup_{\pi \in \Pi} 
            \bigg| \tilde{R}_{N,P}(\pi) - \overline{R}_N(\pi) \bigg| 
        \bigg]
        \leq K \rho_N^{-1}
        .
    \end{align*}
\end{corollary}

\begin{proof}
    By the exact same argument as in the proof of Lemma \ref{lpconvergence},
    we have that for any $P \in \mathcal{P}$,
    \begin{align*}
        \E_P\bigg[ 
            \sup_{\pi \in \Pi} 
            \bigg| \tilde{R}_{N,P}(\pi) - \overline{R}_N(\pi) \bigg| 
        \bigg]
        \leq
        \E_P\bigg[ 
            \frac{1}{N} \sum_{i=1}^N 
            \kappa \lVert \hat m_0(\cdot,X_i) - m_{0,P}(\cdot,X_i) \rVert 
        \bigg]
    \end{align*}
    and the conclusion follows by Assumption \ref{rootnbound}.
\end{proof}

Combining the bounds of Lemma \ref{eplemma} and Corollary \ref{lpconvergencefinite},
we obtain the following simple extension of Theorem \ref{regretconvergence}.
\begin{corollary}
    \label{regretconvergencefinite}
    Let $\mathcal{P}_C$ be a set of distributions for which  (1) Assumptions \ref{DGP} holds with
    constant $C$ and (2) Assumption \ref{rootnbound} holds.
    Under Assumptions \ref{linear} and \ref{policy},
    \begin{align*}
        \sup_{P\in\mathcal{P}_C} 
        \left(\E_P[\overline{R}_P(\hat\pi)] - \overline{R}_P(\pi^*_P) \right)
        \leq 
        K ( N^{-1/2} \vee \rho_N^{-1})
    \end{align*}
    for some constant $K$ depending only on $C$ and $J$.
\end{corollary}

We now show that the finite-sample bound in Assumption \ref{rootnbound} 
is satisfied with $\rho_N = N^{1/2}$ when covariates are discrete, and therefore 
Corollary \ref{regretconvergencefinite} is satisfied with $\rho_N = N^{1/2}$.

\begin{proposition}
    \label{rootnbounddiscrete}
    Let $\mathcal{P}$ be a set of distributions for which 
    \begin{enumerate}
        \item $|\mathcal{X}| = M < \infty$ and $Y_i$ is binary
        \item $P(X_i = x, D_i = d) \geq \delta > 0$ for all 
            $x \in \mathcal{X}$ and $d \in \mathcal{D}_0$.
    \end{enumerate}
    Then there exists a sample size $\bar N$ such that for all $N \geq \bar N$,
    \begin{align*}
        \sup_{P \in \mathcal{P}} \E_P\bigg[ 
            \frac{1}{N} \sum_{i=1}^N 
            \lVert \hat m_0(\cdot,X_i) - m_{0,P}(\cdot,X_i) \rVert 
        \bigg] 
        \leq 
        \left(\frac{J_0 \times |\mathcal{X}|}{\delta}\right)^{1/2} N^{-1/2}
        .
    \end{align*}
\end{proposition}

\begin{proof}
    Define $Z_i = (D_i, X_i)$, and let $L = J_0 \times M$ 
    denote the number of possible combinations of $D_i$ and $X_i$.
    For any $P \in \mathcal{P}$, let $p_\ell = \P[Z_i = \ell]$. 
    By Chernoff's inequality 
    (\cite{vershyninHighdimensionalProbabilityIntroduction2018} Theorem 2.3.1),
    \begin{align}
        \label{probzbad}
        P\left(
            \sum_{i=1}^N \1[Z_i = \ell] \leq \frac{N p_\ell}{2}
        \right)
        \leq
        \exp\left(
            (\ln(2)-1) \frac{N p_\ell}{2}
        \right)
        .
    \end{align}

    For any $\ell \in [L]$, let $m_\ell = \E[Y_i \mid Z_i = \ell]$.
    Conditional on $\sum_{i=1}^N \1[Z_i = \ell]$,
    we have
    \begin{align*}
        \E\left[
            \left(
                \frac{1}{\sum_{i=1}^N \1[Z_i = \ell]}
                \sum_{i=1}^N y_i \1[Z_i = \ell]
                - m_\ell
            \right)^2
        \right]
        &=
        \frac{m_\ell(1-m_\ell)}{\sum_{i=1}^N \1[Z_i = \ell]}
        \\
        &\leq
        \frac{1}{4 \sum_{i=1}^N \1[Z_i = \ell]}
    \end{align*}
    for any $\ell \in [L]$, since $m_\ell \in [0,1]$.

    Since $p_\ell \geq \delta$ by assumption, we have that conditional on
    $\sum_{i=1}^N \1[Z_i = \ell] > \frac{N p_\ell}{2}$,
    \begin{align}
        \label{meanmgood}
        \E\left[
            \left(
                \hat m_\ell - m_\ell
            \right)^2
        \right]
        \leq
        \frac{1}{2 N \delta}
    \end{align}
    and therefore
    \begin{align*}
        \E\left[
            \left(
                \hat m_\ell - m_\ell
            \right)^2
        \right]
        &\leq
        \E\left[
            \left(
                \hat m_\ell - m_\ell
            \right)^2
            \mid
            \sum_{i=1}^N \1[Z_i = \ell]
            > \frac{N p_\ell}{2}
        \right]
        \P\left(
            \sum_{i=1}^N \1[Z_i = \ell]
            > \frac{N p_\ell}{2}
        \right)
        \\ &\quad
        +
        \E\left[
            \left(
                \hat m_\ell - m_\ell
            \right)^2
            \mid
            \sum_{i=1}^N \1[Z_i = \ell]
            \leq \frac{N p_\ell}{2}
        \right]
        \P\left(
            \sum_{i=1}^N \1[Z_i = \ell]
            \leq \frac{N p_\ell}{2}
        \right)
        \\
        &\leq
        \frac{1}{2 N \delta}
        + \exp\left(
            (\ln(2)-1) \frac{N \delta}{2}
        \right)
    \end{align*}
    where in the last lines we have used the bounds
    (\ref{probzbad}) and (\ref{meanmgood}),
    and the fact that $(\hat m_\ell - m_\ell)^2 \leq 1$.
    There exists $\bar N$ such that for all $N \geq \bar N$,
    the second term is less than $\frac{1}{2 N \delta}$
    and therefore
    \begin{align*}
        \E\left[
            \left(
                \hat m_\ell - m_\ell
            \right)^2
        \right]
        &\leq
        \frac{1}{N \delta}
        .
    \end{align*}
    
    By Jensen's inequality, the overall error is bounded by
    \begin{align*}
        \E\left[
            \frac{1}{N} \sum_{i=1}^N 
            \lVert \hat m_0(\cdot,X_i) - m_{0,P}(\cdot,X_i) \rVert
        \right]
        &\leq
        \left(
            \sum_{\ell=1}^L
            \E\left[
                \left(
                    \hat m_\ell - m_\ell
                \right)^2
            \right]
        \right)^{1/2}
    \end{align*}
    and therefore
    \begin{align*}
        \E\left[
            \frac{1}{N} \sum_{i=1}^N 
            \lVert \hat m_0(\cdot,X_i) - m_{0,P}(\cdot,X_i) \rVert
        \right]
        &\leq
        \left(
            \frac{L}{N \delta}
        \right)^{1/2}
    \end{align*}
    for all $N \geq \bar N$.
    Since $L$ and $\delta$ do not depend on $P$,
    this bound is uniform over $\mathcal{P}$.
\end{proof}

\subsection{Estimation of the parameters of the utility function}

In this section we relax the assumption that $b(d,x)$ and $c(d,x)$ are known.
We instead assume that estimates of $b(d,x)$ and $c(d,x)$ are available.
This assumption is stronger than the assumption that $\hat m_0(d,x)$ converges
to $m_0(d,x)$, because it implies a consistent estimate of $b(d,x)$ and $c(d,x)$
is available for new treatments as well as existing ones.
This assumption may be satisfied if $b(d,x) = b(x)$ and $c(d,x) = c(x)$
are known to not depend on $d$.
For example, $b(x)$ may represent the value of a good to a household with characteristics $x$,
and $c(x)$ may represent the cost of delivering the good to the household.
Alternatively, it is also possible that $b(d,x)$ and $c(d,x)$ depend on $d$ in a known way,
such as price minus marginal cost, $b(d,x) = d - b(x)$.

Redefine the regret score analogously to (\ref{gammalp}),
using the estimated utility function.
\begin{align*}
    \hat \Gamma_{jk}(X_i) = 
    \max_{m \in \R^J}
    \quad
    \hat{b}_{jk}(X_i)' m - 
    &\hat{c}_{jk}(X_i)
    \\
    \text{s.t.} \quad
    S m & \leq r
    \\
    F m & = \hat m_0(X_i)
\end{align*}

We assume that the estimated parameters converge to the true parameters at the same rate
as the estimated mean of the treatment effect.
We also assume that mean conditional response is bounded,
as it is in Section \ref{section-application}.

\begin{assumption}
    \label{rootnboundutility}
    For some sequence $\rho_N \to \infty$ and some class of distributions 
    $\mathcal{P}$,
    \begin{enumerate}
        \item $\limsup_{N \to \infty} \sup_{P \in \mathcal{P}}
            \rho_N
            \E_P\bigg[ 
                \frac{1}{N} \sum_{i=1}^N 
                \lVert \hat m_0(\cdot,X_i) - m_{0,P}(\cdot,X_i) \rVert 
            \bigg] 
            \leq \infty$
        \item $\limsup_{N \to \infty} \sup_{P \in \mathcal{P}}
            \rho_N
            \E_P\bigg[ 
                \frac{1}{N} \sum_{i=1}^N 
                \lVert \hat b(\cdot,X_i) - b(\cdot,X_i) \rVert 
            \bigg] 
            \leq \infty$
        \item $\limsup_{N \to \infty} \sup_{P \in \mathcal{P}}
            \rho_N
            \E_P\bigg[ 
                \frac{1}{N} \sum_{i=1}^N 
                \lVert \hat c(\cdot,X_i) - c(\cdot,X_i) \rVert 
            \bigg] 
            \leq \infty$
        \item $\hat{\mathcal{M}} = 
            \{ m : S m(\cdot,X) \leq r, \; F m(\cdot,X) = \hat m_0(\cdot,X)\}$ 
            is nonempty, almost surely,
            for all $P \in \mathcal{P}$.
        \item There exists a bounded set $\overline{\mathcal{M}} \subset \R^J$
            such that $\hat{\mathcal{M}} \subset \overline{\mathcal{M}}$
            almost surely, for all $P \in \mathcal{P}$.
    \end{enumerate}
\end{assumption}

We show that under this assumption, the conclusions of 
Lemma \ref{lpconvergence} and Theorem \ref{regretconvergence} hold.

\begin{proposition}
    \label{regretconvergenceutility}
    Under Assumptions \ref{linear}, \ref{policy}, and \ref{rootnboundutility},
    \begin{align*}
        \sup_{P \in \mathcal{P}} \E_P\left[
            \overline{R}_P(\hat \pi) - \overline{R}_P(\pi^*_P)
        \right]
        \leq \mathcal{O}(N^{-1/2} \vee \rho_N^{-1})
    \end{align*}
\end{proposition}

\begin{proof}
    We write the difference between the estimated and true regret scores as
    \begin{align*}
        \bigg|
            \hat \Gamma_{jk}(X_i) - \Gamma_{jk}(X_i) 
        \bigg| 
        &= 
        \bigg|
            \max_{m \in \hat{\mathcal{M}}}
            \hat{b}_{jk}(X_i)' m - \hat{c}_{jk}(X_i) - 
            \left(
                \max_{m \in \mathcal{M}}
                b_{jk}(X_i)' m - c_{jk}(X_i)
            \right)
        \bigg|
        \\
        &\leq
        \bigg|
            \max_{m \in \hat{\mathcal{M}}}
            \hat{b}_{jk}(X_i)' m -
            \max_{m \in \hat{\mathcal{M}}}
            b_{jk}(X_i)' m
        \bigg|
        \\
        &\quad +
        \bigg|
            \max_{m \in \hat{\mathcal{M}}}
            b_{jk}(X_i)' m -
            \max_{m \in \mathcal{M}}
            b_{jk}(X_i)' m
        \bigg|
        \\
        &\quad +
        \bigg|
            \hat{c}_{jk}(X_i) - c_{jk}(X_i)
        \bigg|
    \end{align*}
    The first term is bounded by
    \begin{align*}
        \bigg|
            \max_{m \in \hat{\mathcal{M}}}
            (\hat{b}_{jk}(X_i) - b_{jk}(X_i))' m
        \bigg|
        &\leq
        \bigg|
            \max_{m \in \overline{\mathcal{M}}}
            (\hat{b}_{jk}(X_i) - b_{jk}(X_i))' m
        \bigg|
        \\
        &\leq
        \lVert \hat{b}_{jk}(X_i) - b_{jk}(X_i) \rVert
        \overline{M}
        \\
        &\leq
        2 \lVert \hat{b}(\cdot,X_i) - b(\cdot,X_i) \rVert
        \overline{M}
    \end{align*}
    where $\overline{M} = \max_{m \in \overline{\mathcal{M}}} \lVert m \rVert$.
    The second term is bounded by 
    \begin{align*}
        \kappa \lVert \hat{m}_0(X_i) - m_{0,P}(X_i) \rVert
    \end{align*}
    by Lemma \ref{lpconvergence}.
    The third term is bounded by
    \begin{align*}
        | c_j(X_i) - \hat{c}_j(X_i) | + | c_k(X_i) - \hat{c}_k(X_i) |
        \leq
        2 \lVert \hat{c}(\cdot,X_i) - c(\cdot,X_i) \rVert
    \end{align*}
    and therefore all three terms are $\mathcal{O}(\rho_N^{-1})$
    uniformly in $P \in \mathcal{P}$.

    Together, we have that
    \begin{align*}
        \sup_{P \in \mathcal{P}} \E_P\bigg[
            \sup_{\pi \in \Pi}
            \bigg| \tilde{R}_N(\pi) - \overline{R}_N(\pi) \bigg|
        \bigg]
        &\leq
        \sup_{P \in \mathcal{P}} \E_P\bigg[
            \frac{1}{N}
            \sum_{i=1}^N
            \max_{j,k}
            \bigg|
                \hat \Gamma_{jk}(X_i) - \Gamma_{jk}(X_i)
            \bigg|
        \bigg]
        \\
        &\leq
        \mathcal{O}(\rho_N^{-1})
    \end{align*}
    
    Combining this with Lemma \ref{eplemma} yields the conclusion.
\end{proof}

\section{Simulation study}
\label{appendix-simulations}
This section presents the results of a simulation study that examines
how the choice of estimator $\hat m_0$ affects the performance of the
resulting estimated policy.

We consider a firm choosing a price to maximize profit.
The firm faces logit demand given by
\begin{align*}
    P(y_i = 1 | z_i) =
    \frac{\exp(z_i' \beta)}{1 + \exp(z_i' \beta)}
\end{align*}
and constant marginal cost of $1$.
Here $z_i = (1, x_i, p_i, x_i p_i)$ and $x_i \sim U[0,1]$ is a continuous
covariate.
The price $p_i$ is randomly assigned from $\mathcal{D}_0 = \{1, 3, 5\}$ 
independently of $x_i$.
The set of possible prices is $\mathcal{D} = \{1, 2, 3, 4, 5\}$.
The conditional mean response is assumed to be decreasing in price,
holding $x_i$ fixed.

For the simulation, we set $\beta = (2.25, 1, -0.25, -1)$.
The purchase probabilities and profits are plotted in Figure 
\ref{fig:purchase-prob-montecarlo}.

\begin{figure}[ht]
    \centering
    \caption{Purchase probabilities and profits}
    \includegraphics[width=0.4\textwidth]{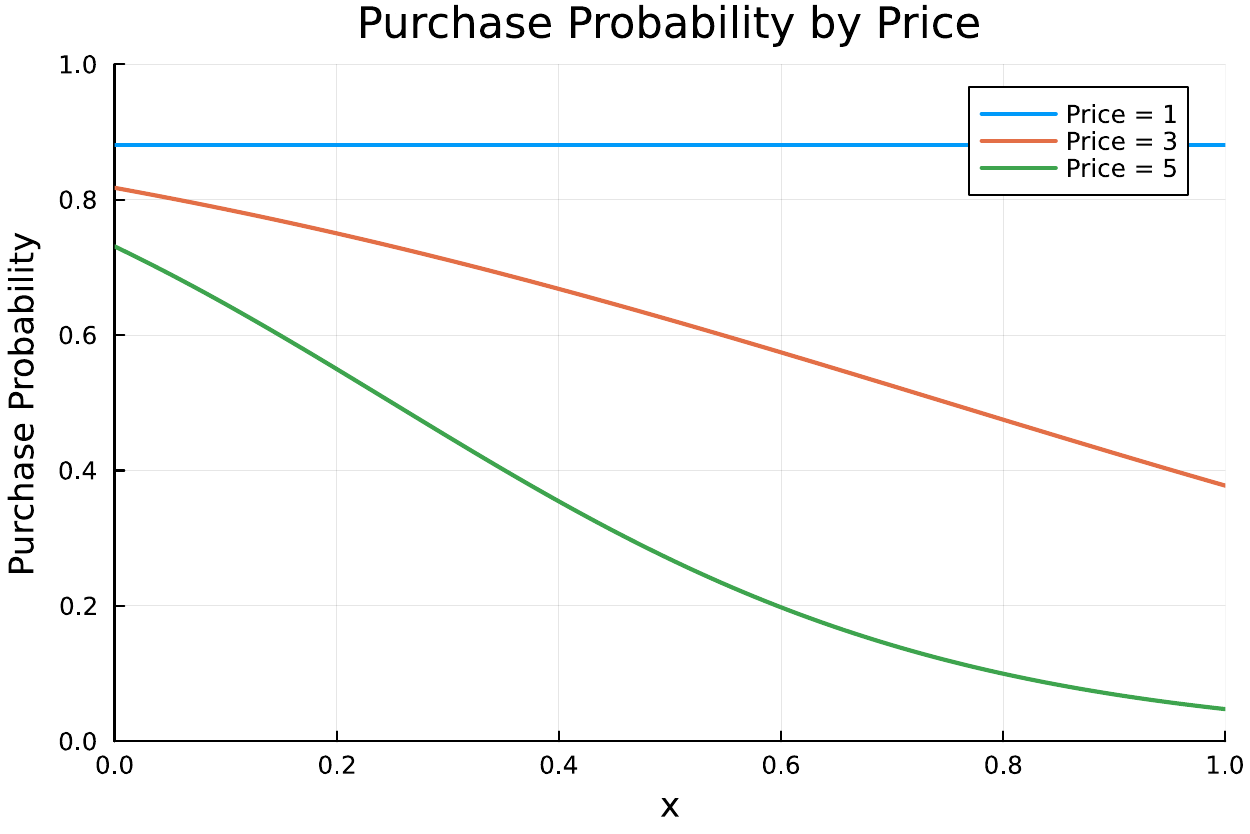}
    \includegraphics[width=0.4\textwidth]{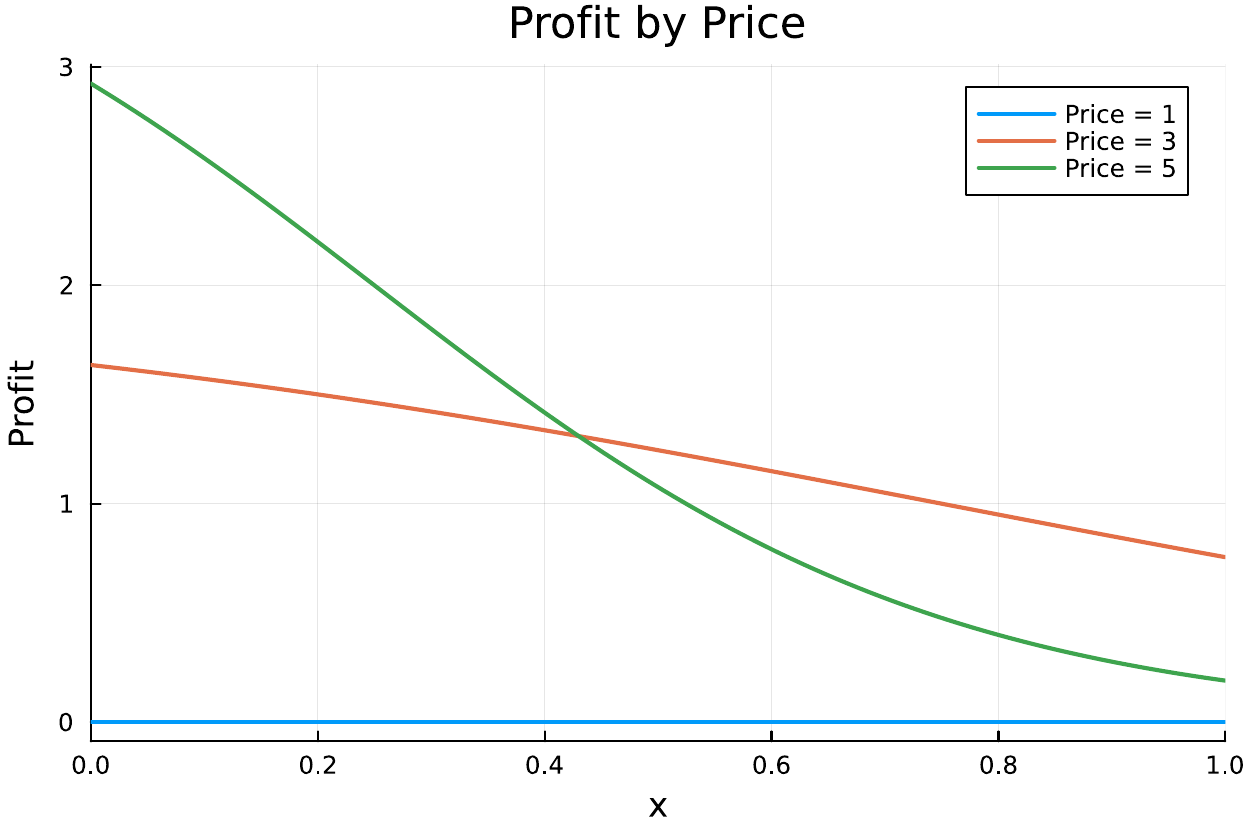}
    \label{fig:purchase-prob-montecarlo}
\end{figure}

A policy is defined by a scalar $\beta$ and a vector of cutoffs $\{\kappa_j\}_{j=1}^4$
with $\kappa_{j-1} \leq \kappa_j$ for all $j$.
Unit $i$ is assigned to price $j$ if 
$\kappa_{j-1} < \beta x_i \leq \kappa_j$ 
(letting $\kappa_0 = -\infty$ and $\kappa_5 = \infty$).
Only the sign of $\beta$ matters, since the cutoffs must be increasing.

We consider $N \in \{100, 250, 500, 750, 1000\}$ observations.
We perform $200$ Monte Carlo simulations for each value of $N$.
For each simulation, we estimate two models for the response function:
a logit model in $x_i$ separately for each price, and a Lasso model with 
a logit link function on a dictionary of Chebyshev polynomials in $x_i$,
separately for each price.
The regularization parameter is selected by cross-validation.
The maximum regret (estimated on holdout data) is plotted in Figure 
\ref{fig:regret-montecarlo}.

\begin{figure}[ht]
    \centering
    \includegraphics[width=0.8\textwidth]{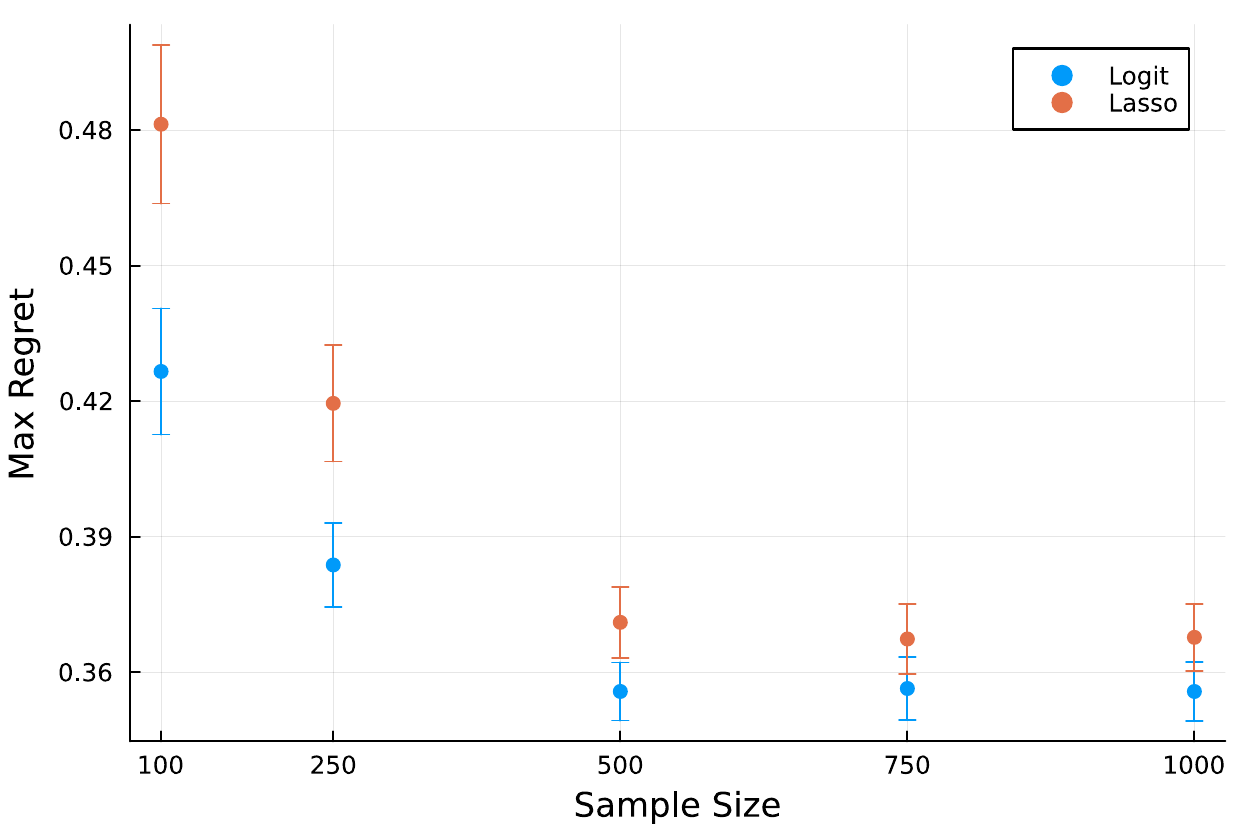}
    \caption{Maximum regret}
    \label{fig:regret-montecarlo}
\end{figure}

As predicted by Theorem \ref{regretconvergence},
the Lasso estimator achieves higher regret than the correctly specified
parametric model.

\section{Computational details}
\label{appendix-computation}

\subsection{Shape constraints}

Mean takeup as a function of price, holding covariates fixed, is bounded between $0$ and $1$.  It is assumed to be downward sloping and the subsidy is assumed to exhibit decreasing returns to scale, so that takeup is convex.  These constraints can be expressed as $S \tilde m(.,X) \leq r$ where
\begin{gather*}
    S_1 = \begin{bmatrix}
        \frac{-1}{d_2 - d_1} & \frac{1}{d_2 - d_1} & 0 & \dots \\
        0 & \frac{-1}{d_3 - d_2} & \frac{1}{d_3 - d_2} & 0 & \dots \\
        \vdots \\
        & & & & \dots & 0 & \frac{-1}{d_J - d_{J-1}} & \frac{1}{d_J - d_{J-1}}
    \end{bmatrix}
    \\
    S_2 = \begin{bmatrix}
        \frac{-1}{d_2 - d_1} & \frac{1}{d_2 - d_1} + \frac{1}{d_3 - d_2} & \frac{-1}{d_3 - d_2} & 0 & \dots \\
        0 & \frac{-1}{d_3 - d_2} & \frac{1}{d_3 - d_2} + \frac{1}{d_4 - d_3} & \frac{-1}{d_4 - d_3} & 0 & \dots \\
        \vdots\\
        & & & & \dots & 0 & \frac{-1}{d_{J-1} - d_{J-2}} & \frac{1}{d_{J-1} - d_{J-2}} + \frac{1}{d_J - d_{J-1}} & \frac{-1}{d_J - d_{J-1}}
    \end{bmatrix}
    \\
    S = \begin{bmatrix}
        -I \\ I \\ S_1 \\ S_2
    \end{bmatrix}
    \quad \quad
    r = \begin{bmatrix}
        0_J \\ -1_J \\ 0_{J-1} \\ 0_{J-2}
    \end{bmatrix}
\end{gather*}

\subsection{Dual representation of maximum regret}

For each individual $i$, the maximum regret is obtained by 
\begin{align*} 
    \hat\Gamma_j(X_i) &= \max_{k,m} b_{jk}(X_i)' m - c_{jk}(X_i) \quad s.t. \quad Sm \leq r, \quad Fm = \hat m(.,X_i)
\end{align*}
For now, suppress the dependence on $X_i$ and $j$.  Let $\lambda$ be the vector of Lagrange dual variables associated with the inequality constraint, and let $\eta$ be the vector of Lagrange dual variables associated with the equality constraint.  We can rewrite the linear program as 
\begin{align*} 
    &\min_{\mu} \mu \quad s.t. \quad \mu + c_k \geq \max_m \{ b_k' m \quad s.t. \quad Sm \leq r, Fm = \hat m_0\} \quad \forall k
    \\
    &\min_{\mu} \mu \quad s.t. \quad \mu + c_k \geq \min_{\lambda, \eta} \{\lambda'r + \eta'\hat m_0 \quad s.t. \quad \lambda' S + \eta'F \geq b_k, \quad \lambda \geq 0 \} \forall k
    \\
    &\min_{\mu, \lambda, \eta} \mu \quad s.t. \quad \mu + c_k \geq \lambda'r + \eta'\hat m_0, \quad \lambda' S + \eta'F \geq b_k, \quad\lambda \geq 0 \quad \forall k
\end{align*}
Thus, computing each $\hat \Gamma_j(X_i)$ is a single linear program.  
Since the programs are independent across $i$, they can be solved simultaneously 
by summing the objective across individuals.

Since the dual formulation above is expressed as a minimization problem,
it can be solved jointly with the minimization over $\pi$ in (\ref{empiricalMMR}).
Alternatively, since there are finitely many $\hat\Gamma_j(X_i)$,
we can solve these linear programs before performing the policy minimization,
and plug the values into (\ref{empiricalMMR}).
In the application of Section \ref{section-application},
solving the problem jointly did not decrease computation time.

\subsection{MILP Formulation of MMR Problem}

Consider the set of policies given by (\ref{policyclass}) in the problem 
(\ref{empiricalMMR}).
To express this as a mixed integer-linear program, 
introduce the binary variables $g_{ij}$ for $i=1,\dots,N$ and $j=1,\dots,J-1$ 
to indicate whether $X_i'\beta$ is above cutoff $j$.
For notational convenience, set $g_{i0} = 1$ and $g_{iJ} = 0$.

We introduce constraints to ensure that $g_{ij}$ is one if and only if 
$X_i'\beta$ is above cutoff $j$ using the ``big-M'' method.
These constraints are 
\begin{align*} 
    X_i' \beta - c_j &\leq M g_{ij}
    \\
    X_i' \beta - c_j &\geq - M (1 - g_{ij}) + \epsilon
\end{align*}
where $M$ is a sufficiently large constant and
$\epsilon$ is a sufficiently small numerical error tolerance.
The first constraint ensures that $g_{ij} = 1$ if $X_i'\beta > c_j$,
and the second constraint ensures that $g_{ij} = 0$ if 
$X_i'\beta \leq c_j + \epsilon$.
The constant $M$ must be chosen large enough to ensure that 
$| X_i'\beta - c_j | \leq M$ for all $i$ and $j$.
The tolerance $\epsilon$ is introduced to imitate a strict inequality 
constraint, which is not possible to impose exactly in a mixed integer-linear 
program.
Otherwise, a solution of $\beta = 0, c_j = 0$ for all $j$
would permit $g_{ij}$ to be either $0$ or $1$ for all $i$ and $j$.
For any optimal policy $(\beta, c)$, the policy $ (t\beta, tc)$ is also optimal
for $t > 0$.
Thus, without loss of generality, we may set $M = 1$.

This leads to the following mixed integer-linear program:
\begin{align*} 
    \min_{g, \beta, c} &\sum_{i=1}^N \sum_{j=1}^J (g_{ij} - g_{i,j-1}) \hat\Gamma_j(X_i)
    \\
    s.t. \quad & c_1 \leq c_2 \leq \dots \leq c_{J-1}
    \\
    & g_{ij} \geq X_i'\beta - c_j + \epsilon, \quad i = 1,\dots,N, \quad j = 1,\dots,J-1
    \\
    & g_{ij} \leq 1 + X_i'\beta - c_j, \quad i = 1,\dots,N, \quad j = 1,\dots,J-1
    \\
    & \beta_1 \geq 0
\end{align*}
which we can solve using a standard mixed integer-linear programming solver.

\section{Robustness}
\label{appendix-robustness}
\subsection{Maximin welfare with Lipschitz constraint}
\label{appendix-lipschitz}

I first illustrate the effect of shape restrictions on the maximin welfare policy.
I give an example of shape restrictions that ensure the maximin welfare policy
assigns a new treatment to the population.
The maximin welfare policy without covariates is given by the price that
maximizes the lower bound on welfare in Figure \ref{worstcaseexample}.
This policy assigns a price of $15$ thousand shillings to the population.
Although the maximin welfare policy restricts to the original set of 
treatments, this is not true in general and depends on the utility function
and the shape restrictions imposed.

To illustrate this, I give an example in which the maximin welfare policy
assigns a new treatment to the population.
Specifically, in the example of Section \ref{section-application},
impose in addition that takeup is Lipschitz continuous,
with a Lipschitz constant of $0.06$.
The resulting bounds on takeup and welfare are plotted in Figure \ref{fig:welfarebounds_lip}.

\begin{figure}
    \center 
    \caption{Bounds on takeup and utility at each price with Lipschitz constraint}
    \includegraphics[width=0.49\textwidth]{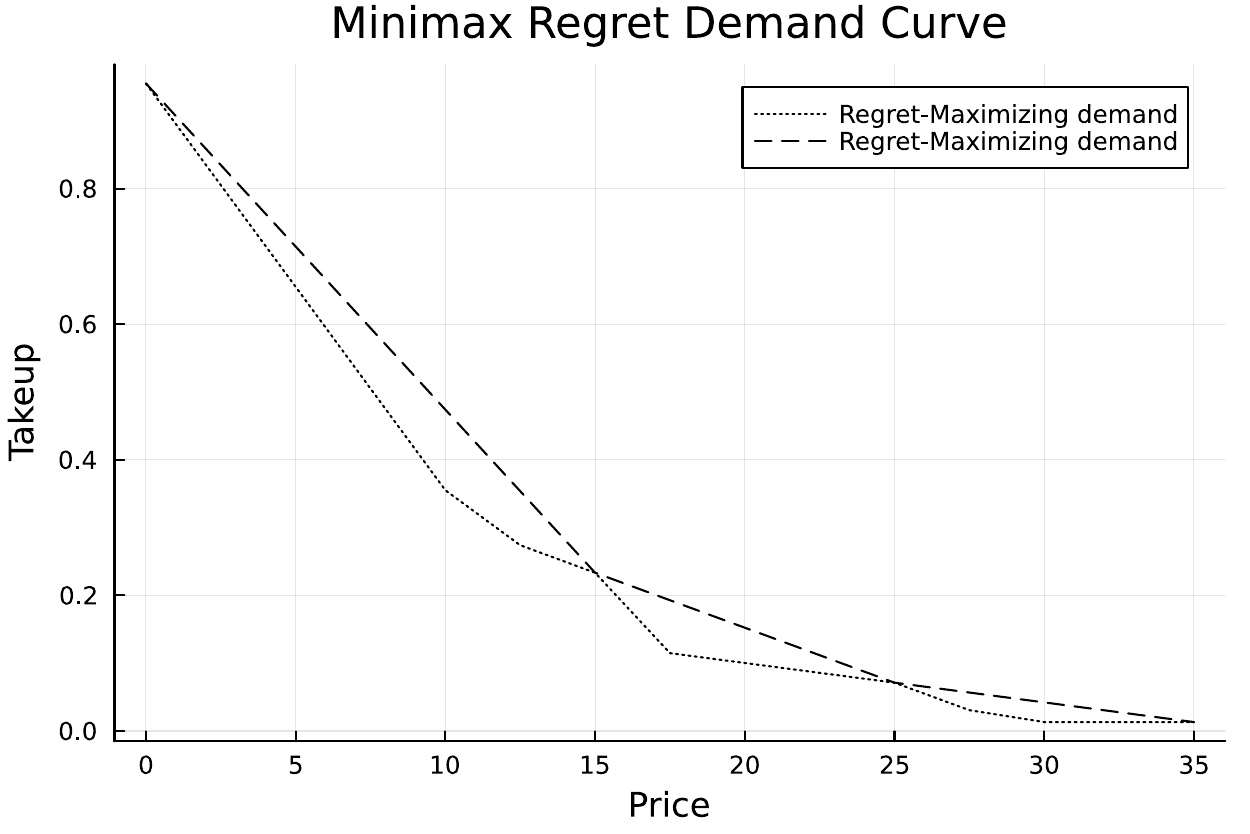}
    \includegraphics[width=0.49\textwidth]{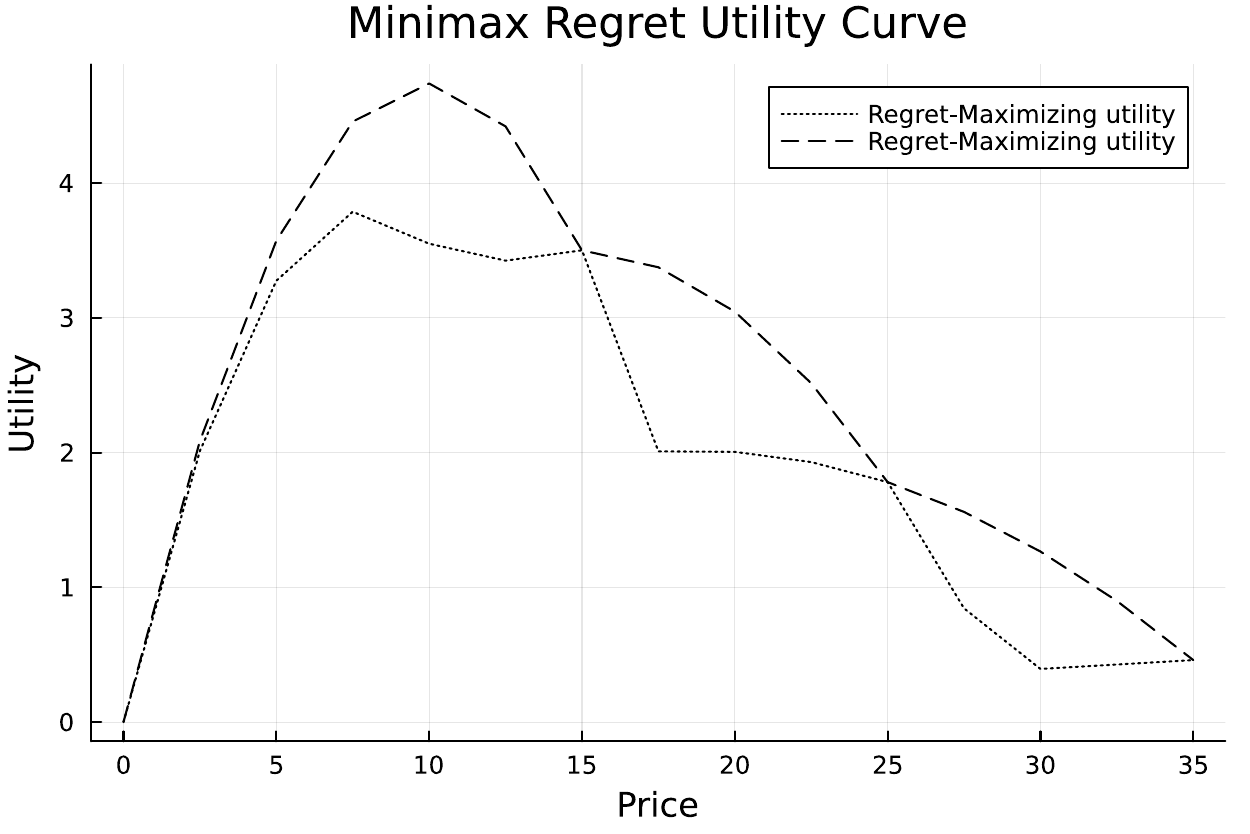}
    \label{fig:welfarebounds_lip}
    \\
    \small \raggedright
    The maximal and minimal possible expected takeup at each price, 
    and the corresponding bounds on expected utility generated by these bounds on takeup.
\end{figure}

From Figure \ref{fig:welfarebounds_lip}, we see that the maximin welfare policy
assigns a price of $7.5$ thousand shillings to the population.

\subsection{Robustness to utility parameters and treatment set}

In these section I present the optimal policy under different specifications of 
$\mathcal{D}$ and $\alpha$.
I consider three specifications of $\mathcal{D}$ constructed from equally spaced
prices between $0$ and $35$ thousand shillings,
where $d_{j+1} - d_j = \Delta$ for $\Delta \in \{5, 2.5, 1\}$.
I consider three values of $\alpha \in \{25, 35, 45\}$.
The results are plotted in Figures \ref{firstfig}-\ref{lastfig},
except for $\Delta = 2.5$ and $\alpha = 25$, which is shown in Section 
\ref{section-application}.

\begin{figure}
    \center 
    \caption{$\Delta = 5$, $\alpha=25$}
    \label{firstfig}
    \includegraphics[width=0.8\textwidth]{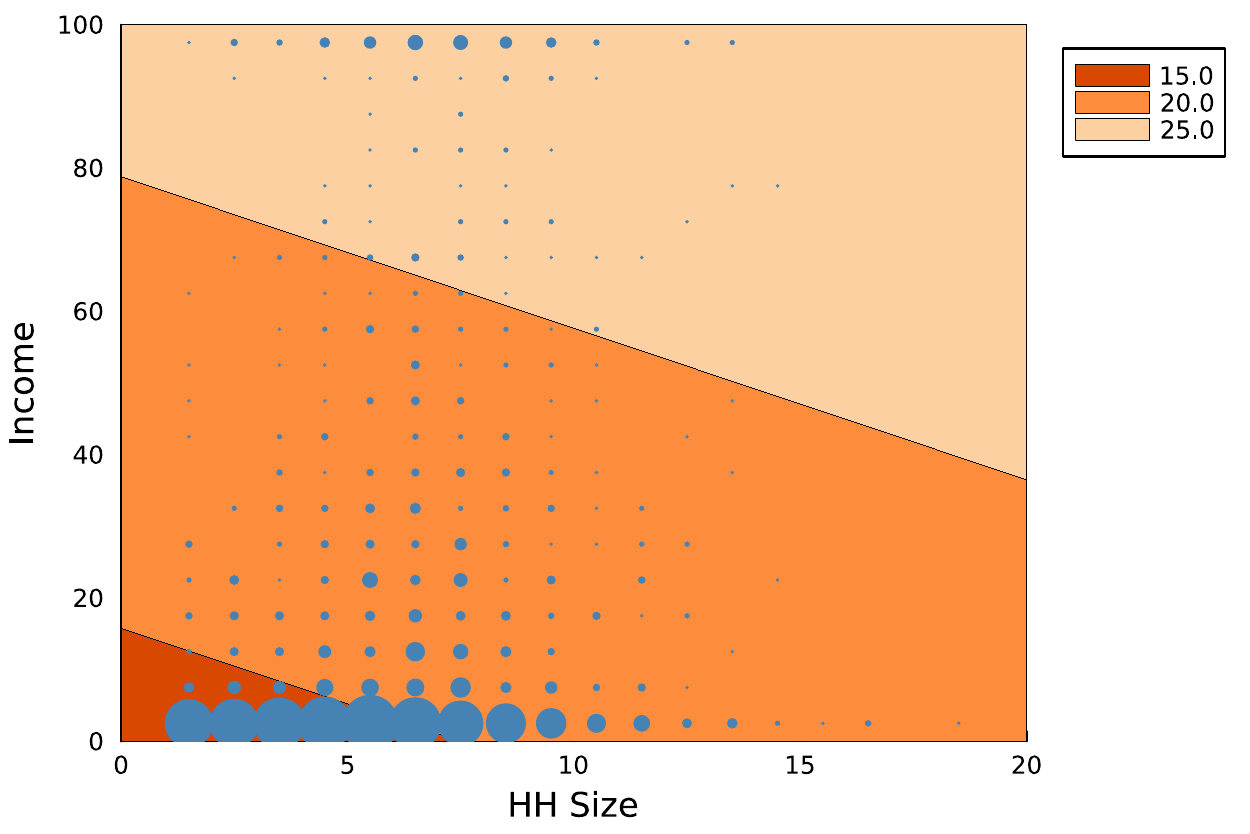}
    \\
    \small \raggedright
    The estimated optimal treatment allocation as a function of household size 
    and earnings.  The size of the dots is proportional to the number of people 
    at each value of covariates.  
    The shaded regions indicate which covariate values are assigned to each treatment.
\end{figure}
\begin{figure}
    \center 
    \caption{$\Delta = 5$, $\alpha=35$}
    \includegraphics[width=0.8\textwidth]{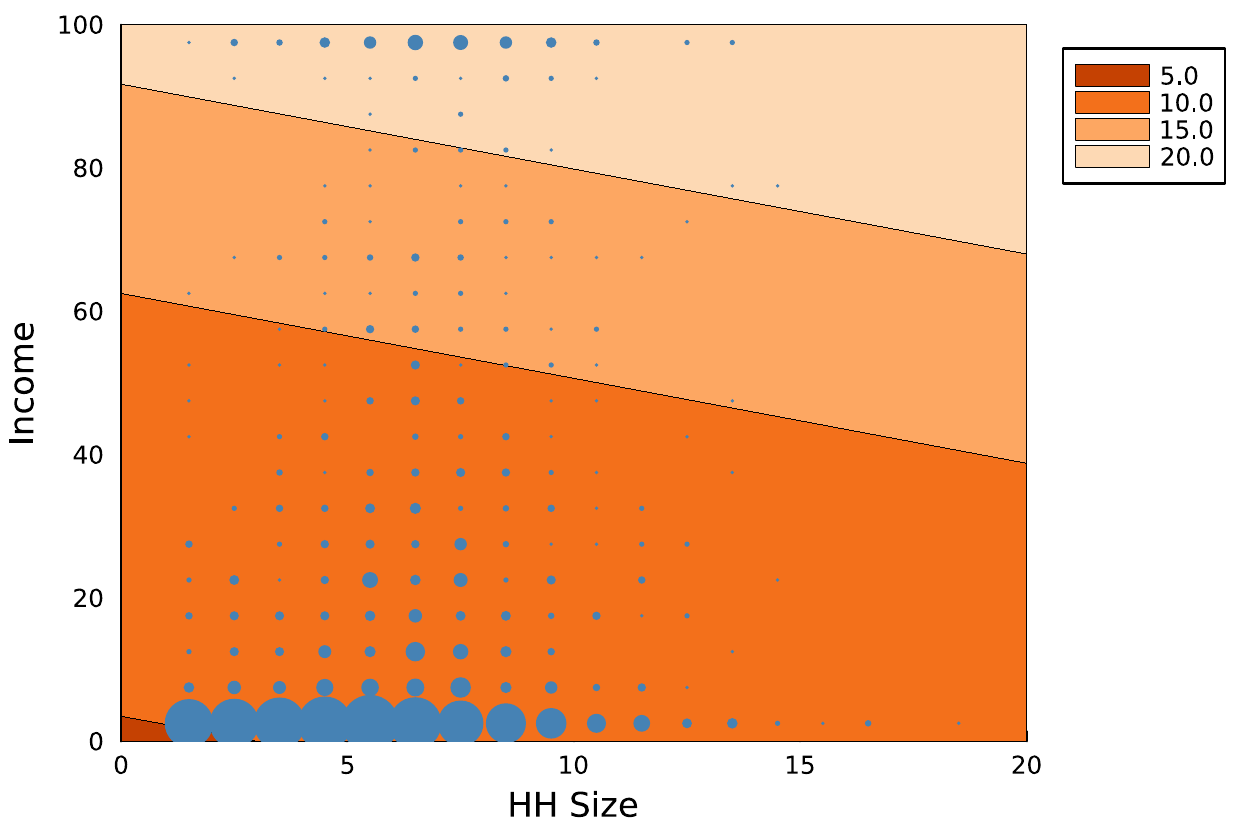}
    \\
    \small \raggedright
    The estimated optimal treatment allocation as a function of household size 
    and earnings.  The size of the dots is proportional to the number of people 
    at each value of covariates.  
    The shaded regions indicate which covariate values are assigned to each treatment.
\end{figure}
\begin{figure}
    \center 
    \caption{$\Delta = 5$, $\alpha=45$}
    \includegraphics[width=0.8\textwidth]{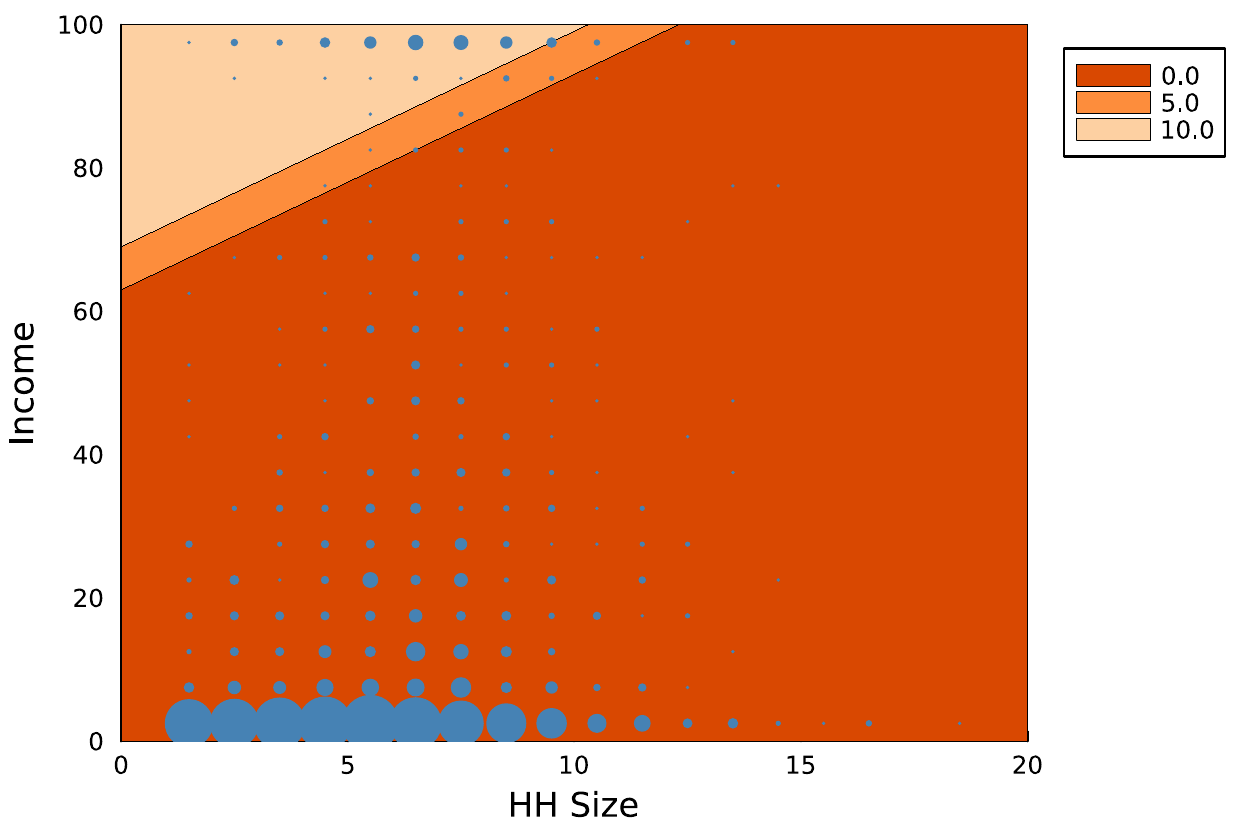}
    \\
    \small \raggedright
    The estimated optimal treatment allocation as a function of household size 
    and earnings.  The size of the dots is proportional to the number of people 
    at each value of covariates.  
    The shaded regions indicate which covariate values are assigned to each treatment.
\end{figure}

\begin{figure}
    \center 
    \caption{$\Delta = 2.5$, $\alpha=25$}
    \includegraphics[width=0.8\textwidth]{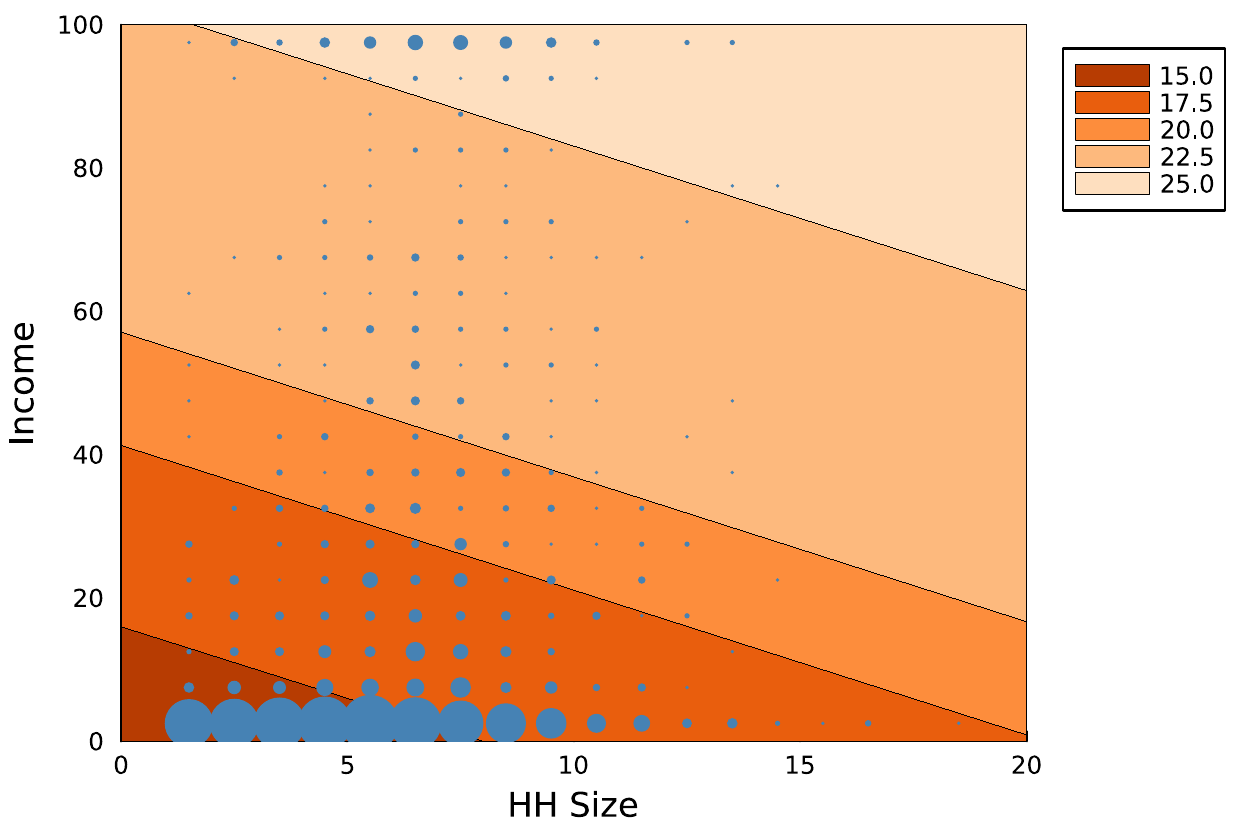}
    \\
    \small \raggedright
    The estimated optimal treatment allocation as a function of household size 
    and earnings.  The size of the dots is proportional to the number of people 
    at each value of covariates.  
    The shaded regions indicate which covariate values are assigned to each treatment.
\end{figure}
\begin{figure}
    \center 
    \caption{$\Delta = 2.5$, $\alpha=45$}
    \includegraphics[width=0.8\textwidth]{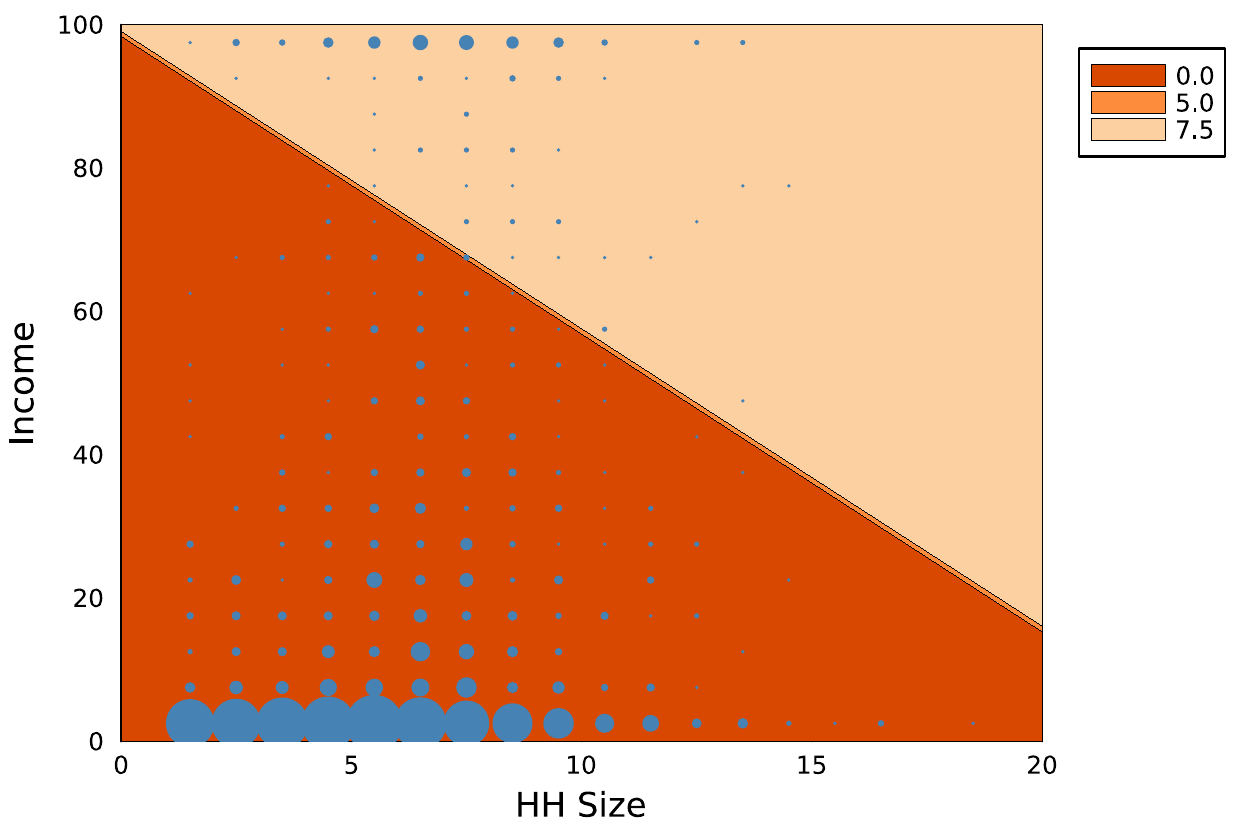}
    \\
    \small \raggedright
    The estimated optimal treatment allocation as a function of household size 
    and earnings.  The size of the dots is proportional to the number of people 
    at each value of covariates.  
    The shaded regions indicate which covariate values are assigned to each treatment.
\end{figure}

\begin{figure}
    \center 
    \caption{$\Delta = 1$, $\alpha=25$}
    \includegraphics[width=0.8\textwidth]{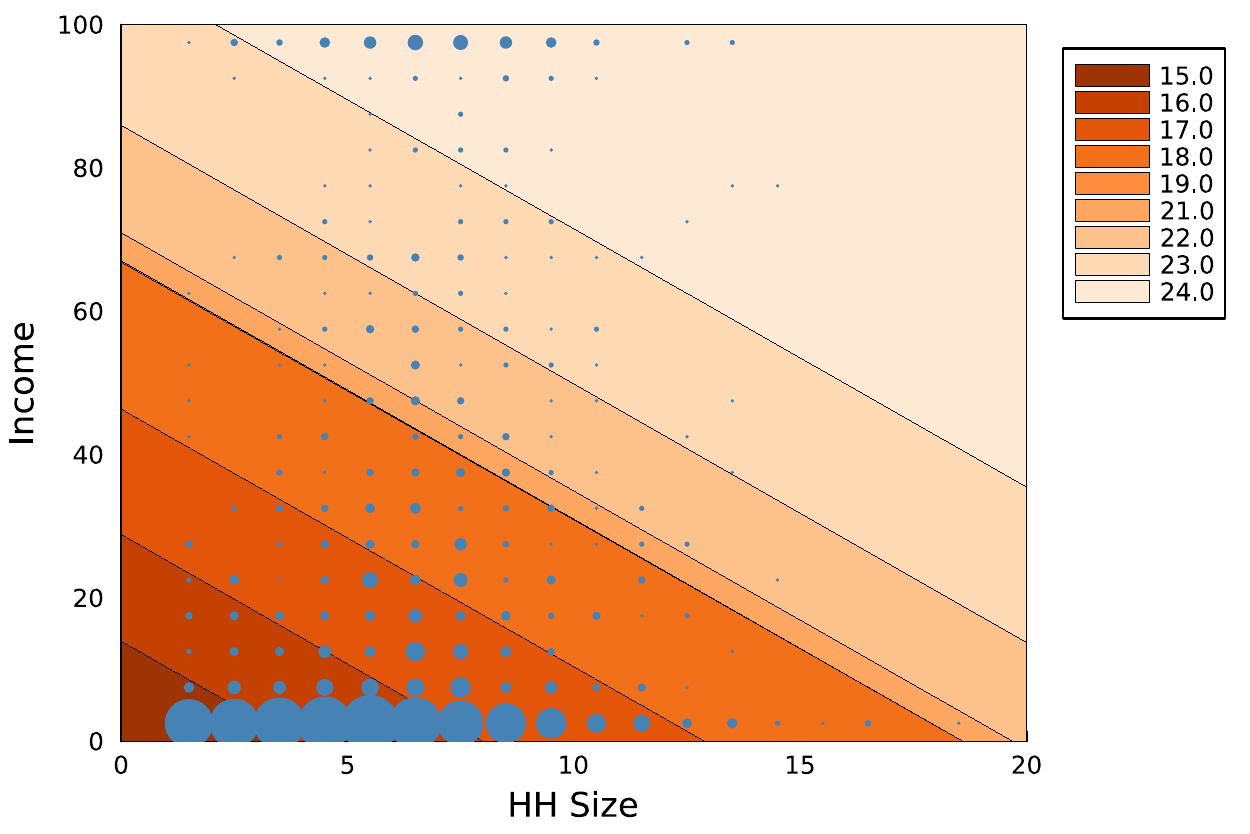}
    \\
    \small \raggedright
    The estimated optimal treatment allocation as a function of household size 
    and earnings.  The size of the dots is proportional to the number of people 
    at each value of covariates.  
    The shaded regions indicate which covariate values are assigned to each treatment.
\end{figure}
\begin{figure}
    \center 
    \caption{$\Delta = 1$, $\alpha=35$}
    \includegraphics[width=0.8\textwidth]{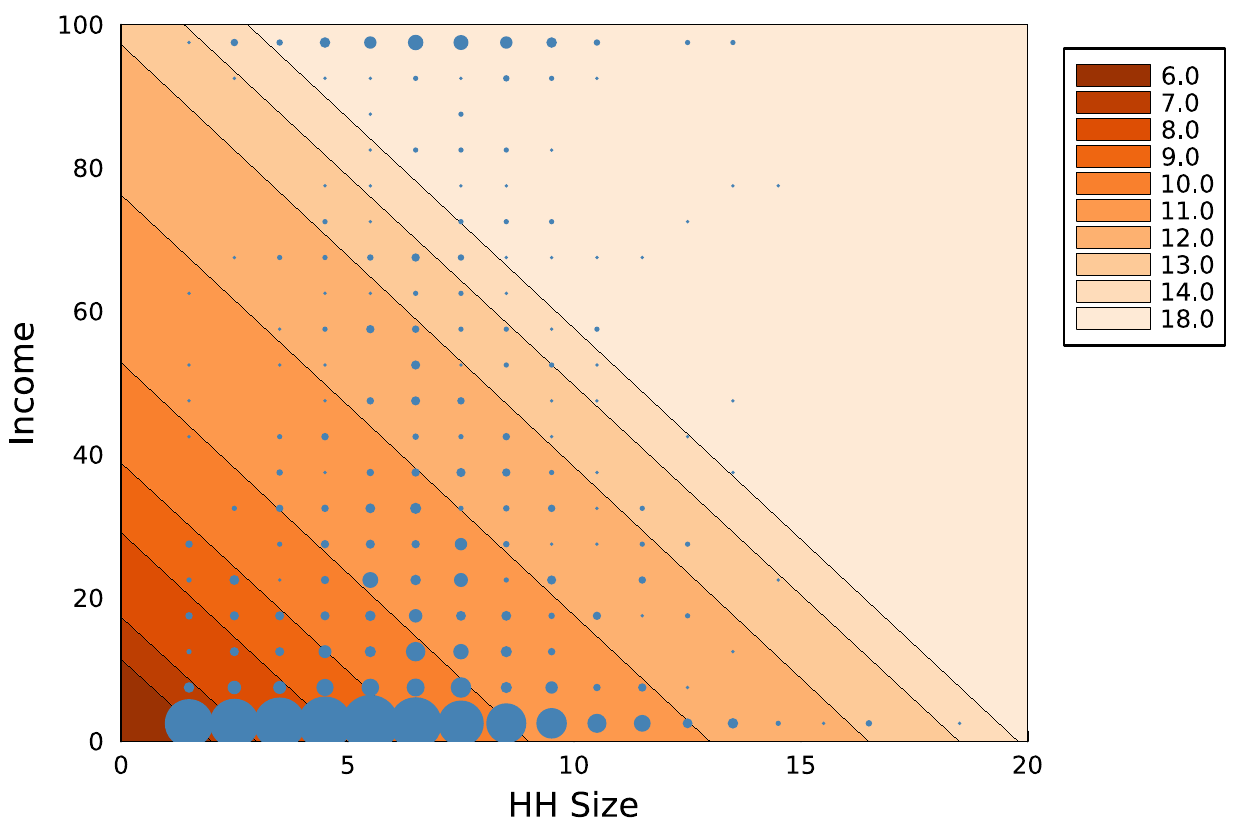}
    \\
    \small \raggedright
    The estimated optimal treatment allocation as a function of household size 
    and earnings.  The size of the dots is proportional to the number of people 
    at each value of covariates.  
    The shaded regions indicate which covariate values are assigned to each treatment.
\end{figure}
\begin{figure}
    \center 
    \caption{$\Delta = 1$, $\alpha=45$}
    \includegraphics[width=0.8\textwidth]{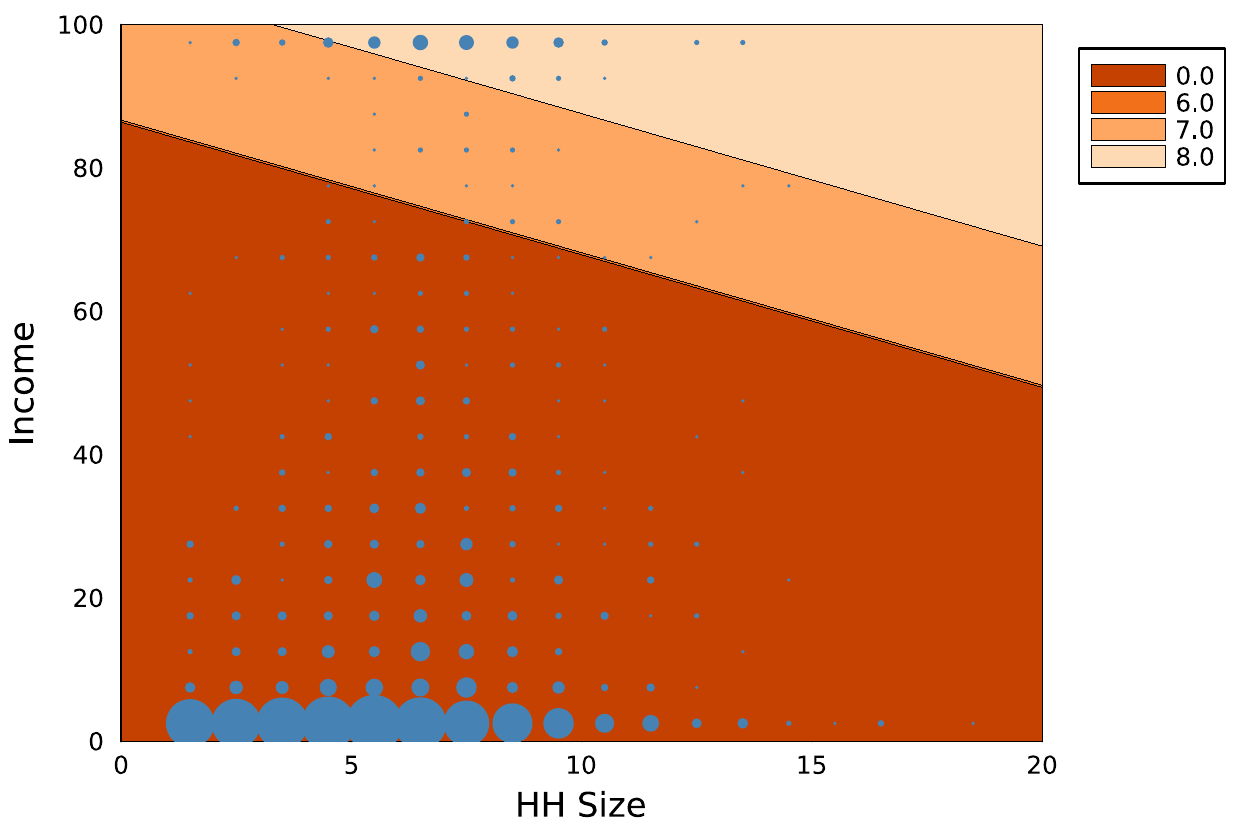}
    \label{lastfig}
    \\
    \small \raggedright
    The estimated optimal treatment allocation as a function of household size 
    and earnings.  The size of the dots is proportional to the number of people 
    at each value of covariates.  
    The shaded regions indicate which covariate values are assigned to each treatment.
\end{figure}

\end{document}